\documentclass[12pt,american]{article}

\usepackage[verbose,tmargin=1in,bmargin=1in,lmargin=1in,rmargin=1in]{geometry}
\usepackage[T1]{fontenc}
\usepackage[utf8]{inputenc}
\usepackage{color}
\usepackage{babel}
\usepackage{refstyle}
\usepackage{amsthm}
\usepackage{amsmath}
\usepackage{amssymb}
\usepackage{bbm}
\usepackage{txfonts}
\usepackage{todonotes}
\usepackage{enumitem}
\usepackage{rotating}
\usepackage[unicode=true, pdfusetitle, bookmarks=true,
  bookmarksnumbered=false, bookmarksopen=false, breaklinks=true, pdfborder={0
  0 0}, backref=false, colorlinks=true, linkcolor=blue, citecolor=blue,
  urlcolor=blue]{hyperref}

\DeclareSymbolFont{largesymbols}{OMX}{yhex}{m}{n}
\DeclareMathAccent{\reallywidetilde}{\mathord}{largesymbols}{"65}
\DeclareMathAccent{\reallywidehat}{\mathord}{largesymbols}{"62}

\makeatletter

\AtBeginDocument{}
\AtBeginDocument{}
\AtBeginDocument{\providecommand\lemref[1]{~\ref{lem:#1}}}

\newcommand{\noun}[1]{\textsc{#1}}
\RS@ifundefined{subref}{\def\RSsubtxt{section~}
\newref{sub}{name = \RSsubtxt}}{}
\RS@ifundefined{thmref}{\def\RSthmtxt{theorem~}
\newref{thm}{name = \RSthmtxt}}{}
\RS@ifundefined{lemref}{\def\RSlemtxt{lemma~}
\newref{lem}{name = \RSlemtxt}}{}

 \theoremstyle{plain}
\newtheorem{thm}{\protect\theoremname}[section] \theoremstyle{plain}
\newtheorem{lem}[thm]{\protect\lemmaname} \ifx\proof\undefined
\newenvironment{proof}[1][\protect\proofname]{%
\par\normalfont\topsep6\p@\@plus6\p@\relax\trivlist\itemindent\parindent
\item[\hskip\labelsep\scshape#1]\ignorespaces}{}

\providecommand{\proofname}{Proof} \fi \theoremstyle{plain}
\newtheorem{cor}[thm]{\protect\corollaryname}
\theoremstyle{plain}
\newtheorem{defn}[thm]{\protect\definitionname}
\theoremstyle{plain}
\newtheorem{rem}[thm]{\protect\remarkname}
\newtheorem{cond}[thm]{\protect\conditionname}
\theoremstyle{plain}

\makeatother

\newref{con}{name = Conjecture\ } \newref{prop}{name = Proposition\ }
\newref{def}{name = Definition\ } \newref{sec}{name = Section\ }
\newref{sub}{name = Section\ } \newref{thm}{name = Theorem\ }
\newref{lem}{name = Lemma\ } \newref{cor}{name = Corollary\ }

\providecommand{\definitionname}{Definition}
\providecommand{\conditionname}{Condition}
\providecommand{\remarkname}{Remark}
\providecommand{\corollaryname}{Corollary}
\providecommand{\lemmaname}{Lemma}
\providecommand{\theoremname}{Theorem}

\newcommand{\db}{}
\newcommand{\ed}{}

\begin{document}

\title{Dynamics of Sound Waves in an Interacting Bose Gas}

\author{D.-A. Deckert, J. Fröhlich, P. Pickl, A. Pizzo}

\maketitle\textbf{}

\begin{abstract} 
    We consider a non-relativistic quantum gas of $N$ bosonic atoms confined to
    a box of volume $\Lambda$ in physical space.  The atoms interact with each
    other through a pair potential whose strength is inversely proportional to
    the density, $\rho=\frac{N}{\Lambda}$, of the gas.  We study the time
    evolution of coherent excitations above the ground state of the gas in a
    regime of large volume $\Lambda$ and small ratio $\frac{\Lambda}{\rho}$. The
    initial state of the gas is assumed to be close to a \textit{product state}
    of one-particle wave functions that are approximately constant throughout
    the box. The initial one-particle wave function of an excitation is assumed
    to have a compact support independent of $\Lambda$.  We derive an effective
    non-linear equation for the time evolution of the one-particle wave function
    of an excitation and establish an explicit error bound tracking the accuracy
    of the effective non-linear dynamics in terms of the ratio
    $\frac{\Lambda}{\rho}$. We conclude with a discussion of the dispersion law
    of low-energy excitations, recovering Bogolyubov's well-known formula for
    the speed of sound in the gas, and a dynamical instability for attractive
    two-body potentials.  
\end{abstract}

\tableofcontents{}

\section{Introduction}

In the study of the intricate dynamics of many-body systems, it is often
convenient, or actually unavoidable, to resort to simpler approximate
descriptions. For quantum-mechanical many-body systems of bosons it is possible
to use effective one-particle equations to track the microscopic evolution of
many-particle states in appropriate regimes. This tends to reduce the complexity
of the problem enormously.  Of course, one has to convince oneself that the
approximation introduced into the analysis is not too crude but resolves the
dynamical features of interest fairly accurately.  To mention an example, the
interaction potential exerted on a test particle in a non-linear one-particle
description of the effective dynamics of a Bose gas can be chosen
self-consistently as the mean potential generated by all the other particles at
the position of the test particle. The mathematical analysis of such so-called
\textit{mean-field limits} goes back to work by Hepp \cite{hepp1974} (quantum
many-body systems), and by Braun and Hepp \cite{braun1977} and Neunzert
\cite{neunzert1977} (classical many-body systems).  Among other results, they
have shown that the Vlasov equation effectively describes a classical many-body
system while the Hartree equation describes a Bose gas in the mean-field limit.
After Hepp's initial work \cite{hepp1974} there has been a lot of effort to
arrive at a mathematically rigorous understanding of quantum-mechanical
mean-field limits; regarding the dynamics see, e.g.,
\cite{spohn,Rodnianski2009,pickl2010,erdos2010,frohlich2009,knowles}, and regarding ground state see, e.g.,
\cite{seiringer,derezinski,lewin} and furthermore \cite{lieb} for an elaborate
overview. 

In oder to clarify the relation between our discussion and previous studies found in
the existing literature, it is necessary to first explain our conventions concerning units
of physical quantities and the use of dimensionless parameters:
\begin{rem}
    All physical quantities appearing in this paper are made dimensionless by
    expressing them in terms of (dimensionful) fundamental constants of Nature or of 
    constants characteristic of the system under consideration.  In this paper, we use units
    in which Planck's constant and the mass of a gas atom are equal to unity. Furthermore,
    distances are expressed as multiples of the diameter of the essential
    support ("range") of the two-body interaction potential, $U$, which equals
    $1$ in our units.  Consequently, to say that the volume $\Lambda$ of the
    region to which the gas is confined equals $1$ would mean that it is comparable
    to the volume of the support of the two-body potential $U$. Furthermore, to say that the density
    fulfills $\rho=1$ would mean that the expected number of particles inside the support of
    $U$ equals $1$.
\end{rem}

With these conventions the situation usually  considered in the
mathematical literature on mean-field limits can be described as follows: The support of 
wave functions is kept fixed while the scattering length of the two-body interaction scales
inversely proportional to the particle number $N$ as the mean-field limit, $N
\rightarrow \infty$, is approached.  In the study of many physically interesting
situations, e.g., of a Bose gas in the thermodynamic limit, one must,
however, consider regimes where $N$ \textit{and} $\Lambda$ tend to
$\infty$. The mean-field regime is then approached by taking the gas density
$\rho=\frac{N}{\Lambda}$ to be large and assuming that the strength of the
two-body potential is ${\cal{O}}(\rho^{-1})$; the mean-field limit 
corresponding to the limit $\rho \rightarrow \infty$.  This ensures that the interaction
energy per particle is of order one and, consequently, the velocity of sound is
kept constant.

A key open problem is to show that the many-body dynamics of a gas of bosonic
atoms can be controlled in terms of an effective equation for a one-particle
wave function when the thermodynamic limit, $\Lambda\to\infty$, is taken at
constant density $\rho$ before the mean-field regime of large $\rho$ is
approached.  While at the present time a satisfactory solution to this problem
appears to be out of reach we propose to make a modest contribution in this
direction by considering an interacting Bose gas at zero temperature in the
regime of large density $\rho$, allowing the volume $\Lambda$ to increase
depending on $\rho$, in such a way that $\frac{\Lambda}{\rho}\ll 1$ as the
mean-field limit is approached. 

More precisely, we propose to study the \emph{microscopic} time evolution of an initial
$N$-particle wave function that is, in a sense to be made precise later, close to a
product wave function of the form
\begin{align} \label{eq:iv-psi-introduction}
  \Psi_{0}(x_{1},x_{2},\ldots,x_{N}) =
\prod_{k=1}^{N}\frac{1}{\Lambda^{1/2}}\left(\phi_{0}^{(\mathrm{ref})}(x_{k})
  +\epsilon_{0}(x_{k})\right).
\end{align}                 
Here, N is the number of atoms in the gas, and $\phi_{0}^{(\mathrm{ref})}$
denotes a slowly varying, compactly supported one-particle wave function chosen
such that its support occupies roughly a region of volume $\Lambda$ and its $L^\infty$
norm is kept constant as $\Lambda$ varies. Its $N$-fold product represents a
so-called \emph{reference state} of the gas, a (Bose-Einstein)
\textit{condensate}, which is then perturbed by a smooth, compactly supported
wave function, $\epsilon_0$, that has a fixed scale-(or $\Lambda$-) independent
support inside the support of $\phi_{0}^{(\mathrm{ref})}$.  The function $\epsilon_0$ is
supposed to describe a localized \emph{excitation} of the reference state. The
time evolution of this initial state is given by the $N$-particle Schrödinger
equation
\begin{equation}
  i\partial_{t}\Psi_{t}(x_{1,\ldots,}x_{N})=H\Psi_{t}(x_{1,\ldots,}x_{N}),\label{eq:microscopic}
\end{equation} 
where the microscopic Hamiltonian, $H$, is given by
\begin{equation} H:=-\frac{1}{2}\sum_{k=1}^{N}\Delta_{x_{k}}
  +\frac{1}{\rho}\sum_{1\leq j<k\leq N}U(x_{j}-x_{k}).
  \label{eq:micro hamiltonian}
\end{equation}

In this work we show that the solution, $\Psi_t$ , corresponding
to equation
(\ref{eq:microscopic}) and initial value (\ref{eq:iv-psi-introduction}) \ed has interesting
features that can be studied with the help of effective one-particle equations
describing the evolution of the reference state $\phi^{(\mathrm{ref})}_t$ and
the excitation $\epsilon_t$; see equations
(\ref{eq:varphi-evolution})-(\ref{eq:epsilon-evolution}) below.  We find that,
in the time evolution of the reference wave function,
quantum-mechanical spreading of the wave packet is suppressed due to the
circumstance that $\phi_{0}^{(\mathrm{ref})}$ is flat.
As a consequence, to leading order, the time-evolved state,
$\phi_{t}^{(\mathrm{ref})}$, equals the initial state $\phi_{0}^{(\mathrm{ref})}$
 up to a time-dependent phase factor.  However, the dynamics of the excitation, i.e., the
behavior of the function $\epsilon_t$, is quite non-trivial.  In particular, its
$L^2$ norm is \textit{not} conserved because of exchange of gas particles between the
condensate (described by the reference state) and the coherent excitation.  Moreover, the
function $\epsilon_t$ disperses according to a law that incorporates a strictly
positive, finite speed of sound in the gas; meaning that sound waves (Goldstone modes)
with arbitrarily small wave number turn out to propagate at a strictly positive speed 
as expected of sound waves in an \textit{interacting} Bose gas, and which has
already be observed in experiments, e.g., \cite{ketterle}.

Excitations of the condensate might be caused by some heavy tracer particles
penetrating into the gas, as considered in \cite{deckert2012}, where the Bose
gas was taken to be an ideal gas. For simplicity we shall not include
such tracer particles in the analysis presented below but
study the dynamics of excitations of the condensate ground-state directly.  The key
analytical ideas used in the analysis of the mean-field limit presented in this
paper are inspired by those introduced in \cite{pickl2011}. They involve some
counting of the number of ``bad particles'', by which we mean particles that do not follow the
(one-particle) effective dynamics.  As compared to \cite{deckert2012}, the
problems addressed in the present work require considerably finer control of the number
of bad particles.  Indeed, since a typical excitation $\epsilon_t$ involves
${\cal O}(\rho)$ many particles, the number of bad particles in a state of
the gas must be controlled in terms of $\rho$ rather than of $N$.  For this
reason, the counting measures used in this work have to be considerably
fine-tuned in order to arrive at useful estimates.

Beside the analysis of dynamics, it should be noted that first steps in the
direction of large volume considering the excitation spectrum of a Bose gas have
also been undertaken in \cite{derezinski} which provides an extension of the
previous results
in \cite{seiringer}.

\paragraph{Outline:} After introducing some important notation in
Section~\ref{sec:notation} we describe our main results in
Section~\ref{sec:main-results} and present the proofs in
Section~\ref{sec:proofs}.\ed

\paragraph{Acknowledgments:} J.F. thanks T. Spencer for hospitality at the
School of Mathematics of the Institute for Advanced Study. J.F.'s stay at the
Institute of Advanced Study has been supported by ``The Fund for Math'' and
``The Robert and Luisa Fernholz Visiting Professorship Fund''. Furthermore,
D-.A.D. and P.P. would like to thank the Mathematical Institute of the LMU
Munich, the Department of Mathematics of UC Davis, and the Institute of
Theoretical Physics of the ETH Zurich for their hospitality.

\subsection{\label{sec:notation}Notation}
\begin{enumerate}
    \item $\left|{\cdot}\right|$ is the standard norm on $\mathbb{R}^d$ or $\mathbb{C}^{d}$, for arbitrary
        $d$; $\left\Vert{\cdot}\right\Vert _{p}$ is the norm on the
  Lebesgue space $L^{p}$, $0\leq p\leq\infty$.  For operators, $O$, acting on the 
  Hilbert space $L^{2}$ 
  we denote by $\left\Vert
  O\right\Vert $ the operator norm of $O$.

  \item Throughout this paper $\Lambda$ denotes both a cube
in physical space $\mathbb{R}^3$ and the volume of this cube. 

    \item For $r>0$ the ball of radius 
$r$ in $\mathbb R^3$ is denoted by ${\mathcal B}_r:=\left\{ x\in{\mathbb R}^3
\,\big|\, |x|<r\right\}$.

\item We denote the Laplace operator and the gradient in the 
  $x-$variable by $\Delta$ and $\nabla$, respectively.

\item The Fourier transform of a function $\eta\in L^{2}$ is denoted by
$\widehat{\eta}$.

\item The convolution of two functions $f$ and $g$ on $\mathbb{R}^{3}$ is defined by
$(f*g)(\cdot):=\int_{\mathbb{R}^3} dy\, f(\cdot-y)g(y)$.  \item By
  ``$F\in\mathrm{Bounds}$'' we mean that $F$ is a continuous, non-decreasing,
  non-negative function on the non-negative reals, i.e.,
  $F:\mathbb{R}^{+}_0\to\mathbb{R}^{+}_0$.
  \item Unless specified otherwise, the
    symbol $C$ denotes a universal constant whose value may change from
    one line to another. In particular, all constants are independent of
    $\Lambda$ and $\rho$. 
\end{enumerate}

\subsection{Main Results} \label{sec:main-results}

As announced in the introduction, the goal pursued in this paper is to understand features of the
time evolution of a many-body wave function, $\Psi_t$, for a given initial
product wave function of the form (\ref{eq:iv-psi-introduction}), which will be
characterized more precisely as follows:
\begin{cond}\label{def:initial-conditions} The many-body wave function of the initial state (at time $t=0$) is given by
  \begin{align}
    \label{eq:iv-psi-phi} \Psi_0(x_1,x_2,\ldots,x_N) = \prod_{k=1}^N
    \frac{1}{\Lambda^{1/2}} \varphi_0(x_k), \qquad \varphi_0:=
    \phi^{(\mathrm{ref})}_0 + \epsilon_0,
  \end{align} where
  $\phi_{0}^{(\mathrm{ref})},\epsilon_0 \in{\cal C}_{c}^{\infty}$
  have the following properties:
  \begin{align} 
      \operatorname{supp}\phi_{0}^{(\mathrm{ref})}\subseteq\Lambda, \qquad
      \Vert\phi_{0}^{(\mathrm{ref})}\Vert_\infty \leq \left\Vert
      \,\reallywidehat{\,|\phi_{0}^{(\mathrm{ref})}|} \, \right\Vert_{1} \leq C,
    \label{eq:conditions-phi}
  \end{align}
  \begin{align}
    \operatorname{supp}\epsilon_{0}\subset{\cal B}_{1/4 \Lambda^{1/3}}, \qquad \Vert
    \epsilon_{0}\Vert _{\infty} \leq \left\Vert \, \widehat{|\epsilon_0|} \,
    \right\Vert_1 \leq C, \qquad \Vert \epsilon_0 \Vert_2 \leq C.
    \label{eq:conditions-epsilon}
  \end{align}
  \begin{align}
\label{eq:psi-varphi-L2-norm}
    \Vert \varphi_0
    \Vert_2 = \Lambda^{1/2}
    \qquad
    \Leftrightarrow
    \qquad
    \Vert \Psi_0
    \Vert_2 = 1.
  \end{align} Furthermore, we assume that the density of the
  gas condensate is essentially constant
  in some large
  region inside the container to which the gas is confined. Therefore, with the
  help of a family of cut-off functions $\chi_r \in {\cal
  C}^2(\mathbb R^3)$, $0< r< 1$,
  \begin{align} \chi_r(x)=
    \begin{cases} 0 & \text{ for }x \in \mathcal B_{r \Lambda^{1/3}} \\ 1 & \text{ for }x
      \notin \mathcal B_{\Lambda^{1/3}}
    \end{cases} \qquad \text{and} \qquad
    \Vert\nabla \chi_r\Vert_\infty\leq C \Lambda^{-1/3},  \label{eq:cutoff}
  \end{align}
  we require
  \begin{equation}\label{eq:phi-cutoff-constraint-epsilon-decay} 
  \left| \,
    \phi_0^{(\mathrm{ref})}(x)-1 \, \right|\leq \chi_{1/2}(x).
  \end{equation}
    This will
  allow us to track the dynamics of the excitation with the properties
  (\ref{eq:conditions-epsilon}) in that region.
  Finally, we require some control of the kinetic energy of the initial reference
  wave function: 
  \begin{equation}
    \label{eq:nabla-phi-2}
      \Vert\nabla\phi^{(\mathrm{ref})}_0\Vert_\infty\leq
            C\Lambda^{-\frac{1}{3}}\;,\hspace{1cm}
      \Vert \nabla\phi^{(\mathrm{ref})}_0 \Vert_2\leq
            C\Lambda^{\frac{1}{6}}\;,\hspace{1cm}
      \Vert \Delta\phi^{(\mathrm{ref})}_0 \Vert_2 
      \leq
      C\Lambda^{-\frac{1}{6}}.
  \end{equation}\ed
\end{cond}

Without further reference we assume Condition~\ref{def:initial-conditions} and
$$U\in\mathcal C^\infty_c(\mathbb R^3,\mathbb R)$$
to hold throughout the entire paper. 
\\

In order to gain control on the dynamics of the many-body wave function $\Psi_t$, 
we show in a first step
that it can be described approximately as a product function of the solution, $\varphi_t$, of the
following nonlinear Schr\"odinger equation
\begin{equation}
  \label{eq:varphi-evolution} i\partial_{t}\varphi_{t}(x) =
  h_x[\varphi_t]\varphi_{t}(x), \qquad h_x[\varphi_t] :=
  -\frac{1}{2}\Delta+U*|\varphi_t|^{2}(x),
\end{equation} 
with initial value $\varphi_t|_{t=0}=\varphi_0$. The sense of the approximation involved in this claim will be made
clear in Section~\ref{sec:proofs}. As already mentioned in the introduction
there are two sources for the dynamics of $\varphi_t$: One is connected to
the evolution of the reference one-particle state $\phi^{(\mathrm{ref})}_t$, and a second one is connected 
to the evolution of the excitation, as described by $\epsilon_t$. In order to conveniently distinguish
between these two sources, the reference state $\phi^{(\mathrm{ref})}_0$
is time-evolved according to the equation
\begin{equation} i\partial_{t}\phi_{t}^{(\mathrm{ref})}(x) =
  \left( -\frac{1}{2}\Delta +U*|\phi_{t}^{(\mathrm{ref})}|^{2}(x)
  -\left\Vert U\right\Vert _{1} \right)
  \phi_{t}^{(\mathrm{ref})}(x)\label{eq:phi-evolution}\;,
\end{equation} and the excitation propagates as described by the equation
\begin{equation}
  \epsilon_{t}:=\varphi_{t}e^{i\left\Vert U\right\Vert
    _{1}t}-\phi_{t}^{(\mathrm{ref})} \label{eq:epsilon-def}.
  \end{equation}
Equations (\ref{eq:varphi-evolution}) and (\ref{eq:phi-evolution}) show that the
evolution of the excitation is given by
\begin{align}
  i\partial_{t}\epsilon_{t}(x) = & \left(-\frac{1}{2}\Delta
  +U*|\phi_{t}^{(\mathrm{ref})}|^{2}(x) -\left\Vert U\right\Vert _{1}
  +U*|\epsilon_{t}|^{2}(x)
  +U*2\Re\left(\epsilon_{t}^{*}\phi_{t}^{(\mathrm{ref})}\right)(x)
  \right)\epsilon_{t}(x) \label{eq:epsilon-evolution} \\ &  + \left(
  U*|\epsilon_{t}|^{2}(x)
  +U*2\Re\left(\epsilon_{t}^{*}\phi_{t}^{(\mathrm{ref})}\right)(x) \right)
  \phi_{t}^{(\mathrm{ref})}(x).  \nonumber
\end{align}
Note that, for a fixed point $x$ deep inside the region $\Lambda$, one has that
\[
  \left| U*|\phi_{t}^{(\mathrm{ref})}|^{2}(x) - \Vert U \Vert_1 \right|
  \approx 0,
\]
which motivates our choice of the phase on the right side of
(\ref{eq:epsilon-def}). Furthermore, in the limit of large $\Lambda$
the reference state $\phi_{t}^{(\mathrm{ref})}$ tends to $1$ so that equation
(\ref{eq:epsilon-evolution}) formally turns into
\begin{align*}
  i\partial_{t}\epsilon_{t}(x) = & \left(-\frac{1}{2}\Delta 
  +U*|\epsilon_{t}|^{2}(x)
  +U*2\Re\left(\epsilon_{t}^{*}\right)(x)
  \right)\epsilon_{t}(x)  + 
  \left( U*|\epsilon_{t}|^{2}(x)
  +U*2\Re\left(\epsilon_{t}^{*}\right)(x)\right).
\end{align*}

We recall the standard facts that, for repulsive $U$, i.e., $U\geq
0$, and given $\Psi_0$, $\varphi_0$, $\phi_{0}^{(\mathrm{ref})}$,
and $\epsilon_0$
as in 
Condition~\ref{def:initial-conditions}, 
there exist unique classical solutions
$\Psi_t$, $\varphi_t$, 
$\phi_{t}^{(\mathrm{ref})}$,
and $\epsilon_t$ to equations (\ref{eq:microscopic}), (\ref{eq:varphi-evolution}), (\ref{eq:phi-evolution}), and
(\ref{eq:epsilon-evolution}), $t\in\mathbb R$,
with initial data $\Psi_{t=0}=\Psi_0$, $\varphi_{t=0}=\varphi_0$, 
$\phi_{t=0}^{(\mathrm{ref})}=\phi_{0}^{(\mathrm{ref})}$, and
$\epsilon_{t=0}=\epsilon_0$, respectively.
In the case of attractive potentials $U$, however, the solution $\varphi_t$, and 
therefore also $\epsilon_t$, may blow up in finite time; see our discussion in
the last paragraph of this
section.\\

In a second step, we show that the control of the $N$-particle wave function
$\Psi_t$ as a function of time $t$ in terms of the one-particle function
$\varphi_t$ is so accurate that the excitation $\epsilon_t$ is ``silhouetted''
against all error terms. In order to compare the microscopic description of the
quantum dynamics with its mean-field description, one must check that the
reduced one-particle density matrix determined by the ``true'' many-body wave
function $\Psi_t$ matches the pure one-particle state given by the one-particle
wave function $\varphi_t$ that one determines by solving equation
(\ref{eq:varphi-evolution}). As discussed in the introduction, the reduced
density matrix of the microscopic (Schr\"{o}dinger) description, 
\[
  \operatorname{Tr}_{x_{2},\ldots,x_{N}}
  \left| {\Psi}_{t}\right\rangle
  \left\langle {\Psi}_{t}\right|,
\]
is given, to leading order, by the projection
$\left|\varphi_{t}\right\rangle\left\langle\varphi_{t}\right|$ onto the one-particle state $|\varphi_t\rangle$.
In order to subtract the contribution from the homogeneous condensate and only track the excitation, we project $\left|\varphi_{t}\right\rangle$ onto the
subspace orthogonal to the reference state. For this purpose we introduce the following notation.
\begin{defn}
  Given a vector $\eta\in L^{2}(\mathbb R^3,\mathbb C)$
  we define the orthogonal projectors
  \[
    p^{\eta}:= \frac{1}{\Vert \eta\Vert_2^2}
    \left|\eta\right\rangle \left\langle \eta\right|,
    \qquad
    q^\eta := 1-p^\eta,
    \qquad
    p_{t}^{(\mathrm{ref})}:=p^{\phi_{t}^{(\mathrm{ref})}},
    \qquad
    q_t^{(\mathrm{ref})}:=q^{\phi_{t}^{(\mathrm{ref})}}.
  \]
\end{defn}
In this notation, the quantities to be compared are the following density
matrices:
\[
    \rho_{t}^{(\mathrm{micro})}:=q_{t}^{(\mathrm{ref})}
    \operatorname{Tr}_{x_{2},\ldots,x_{N}}
    \left|\Lambda^{1/2}{\Psi}_{t}\right\rangle
    \left\langle \Lambda^{1/2}{\Psi}_{t}\right|q_{t}^{(\mathrm{ref})}
    \qquad\text{and}\qquad
    \rho_{t}^{(\mathrm{macro})}:=\left|\epsilon_{t}\right\rangle
    \left\langle \epsilon_{t}\right|.
\]
The additional factor of $\Lambda$ makes up for the
different scalings of $\Psi_t$ and $\epsilon_t$; see
Condition~\ref{def:initial-conditions}.  

Our first result is the following theorem.
\begin{thm}\label{thm:rho-diff-brutal}
  Let $U\in{\cal C}^\infty_c(\mathbb R^3,\mathbb R^+_0)$ be a
  repulsive potential.
  Then there exists a $C\in\operatorname{Bounds}$ such that
  \[
    \left\Vert \rho_t^{(\mathrm{micro})} - \rho_t^{(\mathrm{macro})} \right\Vert
    \leq
    C(t)\frac{\Lambda^{3/2}}{\rho^{1/2}},
  \]
  for all times $t\geq 0$ provided $\Lambda$ is sufficiently large.\ed
\end{thm}
This theorem states that if the thermodynamic limit, $\Lambda \rightarrow
\infty$, and the mean-field limit, $\rho \rightarrow \infty$,
are approached in such a way that $\Lambda\ll\rho^{1/3}$, then the many-body
Schr\"{o}dinger dynamics is well approximated by the non-linear mean-field
dynamics of a one-particle wave function -- at least at the level of
one-particle density matrices. 

Obviously, a key open question is whether the thermodynamic limit can be taken
\textit{before} the mean-field limit is approached.  Concretely, one must ask
how one could possibly improve the rate of convergence established in
Theorem~\ref{thm:rho-diff-brutal}. 
The time evolution necessarily creates some ``bad'' particles,
viz., particles in states that do not follow the mean-field dynamics, throughout the region $\Lambda$
to which the gas is confined. This makes it plausible that, on the one hand, 
the number of bad particles grows with $\Lambda$, while, on the other
hand, it decreases as $\rho$ increases due to our choice of scaling.  Hence, 
when passing to large volumes $\Lambda$, for some fixed $\rho$, it seems 
hopeless to control the norm
\begin{align}
  \label{eq:rhos}
  \left\Vert \rho_t^{(\mathrm{micro})} - \rho_t^{(\mathrm{macro})}
  \right\Vert
\end{align}
directly. In particular, if the thermodynamic limit,
$\Lambda\to\infty$, were taken \textit{before} the mean-field limit,
$\rho\to\infty$, the time evolution would immediately create an infinite number of bad
particles, and (\ref{eq:rhos}) could not possibly be small. 

In this respect it is important to note that a control of
(\ref{eq:rhos}) in the thermodynamic limit is actually stronger than what is
needed when comparing theoretical predictions 
to data about the time evolution of excitations 
gathered in an experiment.
In order to gain access to regimes corresponding to very large
volumes $\Lambda$, one must therefore introduce an appropriate notion of
approximation by mean-field quantities weakening (\ref{eq:rhos}). One such possibility would be to
introduce a semi-norm involving the restrictions of the one-particle density matrices to a bounded 
region $\lambda\subset\Lambda$ of interest with
a volume of order ${\cal O}(1)$, e.g.,
\begin{align}\label{eq:cut}
    \left\Vert \mathbbm{1}_\lambda\left(\rho_t^{(\mathrm{micro})}
    - \rho_t^{(\mathrm{macro})}\right)\mathbbm{1}_\lambda
  \right\Vert,
\end{align}
where $\mathbbm{1}_\lambda(x)$ is some cut-off function with support in $\lambda$. For finite
times, an excitation of the gas created in some bounded region of space can be expected 
to essentially remain localized in a bounded region. Thus, 
control of (\ref{eq:cut}) may turn out to suffice to study its dynamics for a finite interval of times
and compare it to its effective (mean-field) dynamics. The technical control of
a quantity like (\ref{eq:cut}) is however cumbersome as one needs to control the
flow of particles from $\Lambda\setminus\lambda$ into the volume $\lambda$ without having much information
about them.

Another possibility in the direction of large volumes -- the one explored in this paper -- is to show that
(\ref{eq:rhos}) is \emph{typically} small, the precise
mathematical statement being: There is a trajectory of vectors $\widetilde\Psi_t$ with corresponding
reduced density matrix $\widetilde\rho^{(\mathrm{micro})}_t$ such that
$\|\widetilde\Psi_t-\Psi_t\|_2$ and $\|\tilde\rho_t^{(\mathrm{micro})} -
\rho_t^{(\mathrm{macro})}\|$ are both small. Such a result may actually be expected to enable 
one to answer most physical questions in a satisfactory way as only what happens with large probability really
matters for the comparison with an experiment.  Let us try to explain why this
mode of approximation is helpful: If the
volume $\Lambda$ of the region to which the gas is confined is large, the gas
contains a vast number of particles. Suppose that, with a tiny probability, the
positions of all these particles are changed. Such a change may yield a
significant variation of the reduced density matrices of the system.  However, events
that happen with a very small probability are not important physically.  Hence, the
fact that the reduced density matrices may change appreciably is unimportant.

With the next two results we explore this probabilistic idea and demonstrate how
the result in Theorem~\ref{thm:rho-diff-brutal} can be improved. The basis for this improvement forms 
the contents of our second
main result. To state it we make the notion of ``bad'' particles precise. We introduce orthogonal projectors
\begin{align}
    P^{\varphi_t}_{k}=\left(q^{\varphi_t}\right)^{\odot k}\odot
\label{eq:projector}
\left(p^{\varphi_t}\right)^{\odot(N-k)}, \qquad 0\leq k\leq N,
\end{align}
where $\odot$ denotes the symmetric tensor product.  The projector
$p^{\varphi_t}$ is to be thought of as projecting onto one-particle states of
``good'' particles, while $q^{\varphi_t}$ projects onto one-particle states of
``bad'' particles; see equation 
(\ref{eq:projectors}) below.  The probability, $\mathbb P_t$, of the event that the total number of bad
particles described by the many-body wave function  $\Psi_t$ is larger than the density $\rho$ is given
by
\begin{equation}
\label{eq:PtotM}
  \mathbb P_{t}\left(
             \text{total number of bad particles}>\rho
           \right)
  :=
 1-\left|\langle \Psi_t|\widetilde \Psi_t \rangle\right|^{2}, \quad
  \text{where} \quad \widetilde \Psi_t := \sum_{1\leq k\leq
  \rho}P^{\varphi_t}_k\Psi_t.
\end{equation}
This quantity is estimated in our second main result. \ed
\begin{thm}\label{thm:probM}
  Let $U\in{\cal C}^\infty_c(\mathbb R^3,\mathbb R^+_0)$ be a
  repulsive potential.
  Then there is a $C\in\operatorname{Bounds}$ such that
  \[
   \left\| \Psi_t - \widetilde\Psi_t\right\|^2_2     \leq
    C(t)\frac{\Lambda}{\rho},
  \]
 for all times $t\geq 0$, provided $\Lambda$ is sufficiently large. \ed
\end{thm}
%

We pause to interpret this result. 
As a gedanken experiment, we imagine that the density of the Bose gas is
measured, e.g., by shining light into the condensate and then recording the scattered light by means of
a photograph -- as one does in recent experiments with cold atom gases,
where for example a sequence of photographs is taken to record the dynamics of the
Bose gas cloud; see also \cite{ketterle}. As long as one can recognize
a localized excitation on the photograph of the gas, one can argue that there
are at most $\mathcal{O}(\rho)$ bad particles in the state of the gas, and hence that the
state after the measurements is close to the vector
$\widetilde\Psi_t$. Theorem~\ref{thm:probM} then says that if $\frac{\Lambda}{\rho}\ll 1$ the
state of the system is very close to the vector $\widetilde\Psi_t$, and, in this case, the result in
Theorem~\ref{thm:rho-diff-brutal} can be further improved as follows (our third main result).

\begin{thm}\label{thm:improved}
  Let $U\in{\cal C}^\infty_c(\mathbb R^3,\mathbb R^+_0)$ be a
  repulsive potential.
  Then there exists $C\in\operatorname{Bounds}$ such that
  \[
    \widetilde\rho_{t}^{(\mathrm{micro})}:=q_{t}^{(\mathrm{ref})}
    \operatorname{Tr}_{x_{2},\ldots,x_{N}}
    \left|\Lambda^{1/2}{\widetilde\Psi}_{t}\right\rangle
    \left\langle \Lambda^{1/2}{\widetilde\Psi}_{t}\right|q_{t}^{(\mathrm{ref})},
  \]
  fulfills
  \begin{align}
\label{eq:rho-tilde-sense}
        \left\Vert \widetilde \rho^{(\mathrm{micro})}_t -
        \rho^{(\mathrm{macro})}_t \right\Vert \leq
        C(t)
        {\frac{\Lambda^{1/2}}{\rho^{1/2}}},
  \end{align}
for all times $t\geq 0$ provided $\Lambda$ is sufficiently large. \ed
\end{thm}\ed

\begin{rem}\label{rem:general}
    It should be stressed that Theorems~\ref{thm:rho-diff-brutal},
    ~\ref{thm:probM} and ~\ref{thm:improved} also hold (i) for more general
    initial states $\Psi_0$ which, however, must be close to the product state
    in (\ref{eq:iv-psi-phi}), see Remark~\ref{rem:technique} below; and (ii)
    for \underline{attractive} two-body potentials $U$ and times $0\leq
    t<T\leq\infty$ provided  $\Vert\varphi_t\Vert_\infty$ stays bounded for
    $0\leq t\leq T$.  As mentioned above, the case of attractive potentials is
    more subtle because solutions of the evolution equation
    (\ref{eq:epsilon-evolution}) may blow up in finite time. Indeed, for this
    case the Bose gas collapses in the thermodynamic limit, and it is then not
    surprising that convergence to the mean-field limit fails, too. \ed
\end{rem}

In order to further analyze the dynamics of $\Psi_t$, we consider 
excitations $\epsilon_t$ of very small $L^2-$ and bounded $L^\infty-$ norm.
In this case we find that the
evolution of $\epsilon_t$ is well described by a linear version of
equation (\ref{eq:epsilon-evolution}), namely
\begin{equation}
  i\partial_{t}\eta_{t}(x)=-\frac{1}{2}\Delta\eta_{t}(x)+U*2\Re\eta_{t}(x),
  \label{eq:approximate-time-evolution}
\end{equation}
with initial condition $\eta_{t}|_{t=0}=\epsilon_0$. Indeed, in Section~\ref{sec:eta} we
prove the following theorem: \ed
\begin{thm}
  \label{thm:main-epsilon}
  Let $U\in\mathcal C^\infty_c(\mathbb R^3,\mathbb R)$ be a general 
  potential. \ed
  Suppose $\epsilon_t$ and $\eta_t$ solve the equations
  (\ref{eq:epsilon-evolution}) and (\ref{eq:approximate-time-evolution}),
  respectively, for $0\leq t< T\leq \infty$ and initial data $\epsilon_t|_{t=0}=\epsilon_0=\eta_t|_{t=0}$. Then there is a $C\in\mathrm{Bounds}$ such that
    \begin{equation}\label{eq:linearized-nonlinearized} \left\Vert
    \eta_{t}-\epsilon_{t}\right\Vert _{2}\leq C(t)\sup_{s\in[0,t]}\left(
    \Lambda^{-\frac{1}{6}} + \left\Vert
    {\epsilon}_{s}\right\Vert _{2}^{2}+\left\Vert {\epsilon}_{s}\right\Vert
    _{2}^{3}\right),
  \end{equation}
  for times $0\leq t< T$ provided $\Lambda$ is sufficiently large. \ed
\end{thm}
The evolution equation (\ref{eq:epsilon-evolution}) is then quite easy to analyze.
After a Fourier transformation, 
\[ \widehat{\eta}_{t}(k):=(2\pi)^{-3/2}\int d^3x\,e^{-ikx}\eta_{t}(x), \] 
of $\eta_{t}$, we rewrite (\ref{eq:approximate-time-evolution}) in
momentum space
\begin{equation}
  i\partial_{t}\widehat{\eta}_{t}(k)=\omega_{0}(k)\widehat{\eta}_{t}(k)+\widehat{U}(k)\left(\widehat{\eta}(k)+\widehat{\eta}^{*}(-k)\right),\label{eq:eta-fourier}
\end{equation} where we have used that
$\widehat{\eta^{*}}(k)=\widehat{\eta}^{*}(-k)$, and where 
\[
\omega_{0}(k)=\frac{k^2}{2}
\] is
the symbol of the differential operator $-\frac{1}{2}\Delta$ in momentum space.
The complex conjugate of this equations is given by
\[
i\partial_{t}\widehat{\eta}_{t}^{*}(-k)=-\omega_{0}(k)\widehat{\eta}_{t}^{*}(-k)-\widehat{U}(k)\left(\widehat{\eta}^{*}(-k)+\widehat{\eta}(k)\right),
\]
where we have used that \[
  \omega_{0}(k)\equiv\omega_{0}(|k|)
  \qquad
  \text{and}
  \qquad
  \widehat{U}(k)=\widehat{U}^{*}(-k)
\] 
as the potential $U(x)$ is real-valued. The evolution equations for
$\widehat{\eta}_{t}(k)$ and $\widehat{\eta}_{t}^{*}(-k)$ can then be written in
closed form as \[ i\partial_{t}\begin{pmatrix}\widehat{\eta}_{t}(k)\\
    \widehat{\eta}_{t}^{*}(-k)
  \end{pmatrix}={\cal
  H}(k)\begin{pmatrix}\widehat{\eta}_{t}(k)\\ \widehat{\eta}_{t}^{*}(-k)
  \end{pmatrix},\qquad \text{with} \qquad {\cal H}(k):=\begin{pmatrix}\omega_{0}(k)+\widehat{U}(k)
    & \widehat{U}(k)\\ -\widehat{U}(k) & -\omega_{0}(k)-\widehat{U}(k)
\end{pmatrix}.  \] 
Note that $\cal H$ is not self-adjoint, and hence, the $L^2$ norm of $\eta_t$ is
not preserved. However, one can still find a basis w.r.t.\  which $\cal H$ is
diagonal.
For arbitrary $\widehat{U}(k)$, an eigenvalue,
$\omega(k)$, of ${\cal H}(k)$ fulfills
\begin{equation}
  \omega(k)^{2}=\omega_{0}(k)\left(\omega_{0}(k)+2\widehat{U}(k)\right).
  \label{eq:effective-dispersion}
\end{equation} This shows how the dispersion law, $\omega(k)$, of sound waves 
in the gas depends on the pair potential $U$. We consider
two interesting cases:\\

\textbf{Repulsive potential, e.g., $\widehat{U}(0)>0$:}  \[
    |\omega(k)|=|k|\sqrt{\frac{k^{2}}{4}+
  {\widehat{U}(k)}}.\] 
  Apparently, the speed of sound at small values of $|k|$ is then given by
\[ v_{\mathrm{sound}}=\sqrt{{\widehat{U}(0)}}, \] 
which is a well-known result due to Bogolyubov \cite{Bogolyubov1947}. Note that the fact that $v_{\mathrm{sound}}$ 
does \textit{not} depend on the density $\rho$ of the gas is owed to 
the scaling in (\ref{eq:micro hamiltonian}).\\

\textbf{Attractive potential, e.g., $\widehat{U}(k)<0$:}  
For such potentials $U$, modes with wave vectors $k$ fulfilling
 $\omega_0(k)=-2\widehat{U}(k)$ become static
according to the effective dispersion relation
\[
 \omega(k)=\omega_{0}(k)^{1/2}\sqrt{\omega_{0}(k)+2\widehat{U}(k)},
\]
 while modes corresponding to
 wave vectors $k$ with $\omega_{0}(k)<-2\widehat{U}(k)$ are
 dynamically \textit{unstable}. This instability causes the gas to implode at a
 finite time. As noted in Remark~\ref{rem:general}, 
our main results about the $N$-particle time evolution also hold for attractive
two-body potentials $U$,
as long as $\Vert\varphi_t\Vert_\infty$ remains bounded, i.e., for sufficiently
short times, which is why for those times $\eta_t$ also gives insights into the
microscopic dynamics of $\Psi_t$. \ed

\begin{rem}
    We note that the proofs provided in this paper
    also work for dispersion relations other than
    $\omega_0(k)=\frac{k^2}{2}$. While the propagation estimates given in
    Section~\ref{sec:propagationestimates} would have to be adapted,
    the mean-field estimates hold for any dispersion relation as
    all one-particle terms in the
    Hamiltonian drop out immediately; see (\ref{eq:dt-m}) below.
\end{rem}

\section{Proofs}\label{sec:proofs}

In this section, we present the proofs of our results. The organization of our reasoning process is as follows. 
\begin{itemize} 
    \item Section~\ref{sec:bad-particles}: Our first technical result,
        Lemma~\ref{lem:m-bound}, aims at controlling the number of bad
        particles present in the state of the gas.  This lemma will be proven under the assumption
        that $\Vert\varphi_t\Vert_\infty$ is bounded following ideas of \cite{pickl2011}. Note that the control of
        the Hartree dynamics (\ref{eq:varphi-evolution}) is well understood. One
        might then ask why  Lemma~\ref{lem:m-bound} is needed. The reason is
        that we are ultimately interested in the dynamics of
        \textit{excitations}, and for this it turns out in the proofs of
        Theorem~\ref{thm:probM} and Theorem~\ref{thm:improved} that considerably stronger bounds on the number of bad particles
        are necessary.

\item Section~\ref{sec:proofs-main-results}: Using
  Lemma~\ref{lem:m-bound} we proceed to proving our first three main results, namely 
  Theorems ~\ref{thm:rho-diff-brutal}, ~\ref{thm:probM}, and ~\ref{thm:improved}.
  These results hold provided the assumptions (\ref{eq:varphi-Linfty-bound}),
(\ref{eq:epsilon-L2-bound}) and (\ref{eq:p-epsilon-L2-bound}) hold true.

\item Section~\ref{sec:propagationestimates}: Here 
  ``propagation estimates'' justifying
  the assumptions (\ref{eq:varphi-Linfty-bound}),
(\ref{eq:epsilon-L2-bound}) and (\ref{eq:p-epsilon-L2-bound}) will be derived.

\item Section~\ref{sec:eta}:
  To conclude, we provide the proof of
  Theorem~\ref{thm:main-epsilon} which is also based on those propagation
  estimates.

\end{itemize}

\subsection{Controlling the number of ``bad''
particles}\label{sec:bad-particles}

For any $\varphi\in L^2$, we use the notation
  \begin{equation}
    q_{k}^{\varphi}:=1-p_{k}^{\varphi},\qquad\left(p_{k}^{\varphi}\Psi\right)(x_{1},\ldots,x_{N}):=\frac{\varphi(x_{k})}{\Vert
    \varphi\Vert_2}\int
    d^{3}x_{k}\,\frac{\varphi^{*}(x_{k})}{\Vert
    \varphi\Vert_2}\Psi(x_{1},\ldots,x_{N}),\qquad1\leq k\leq
    N.\label{eq:projectors}
  \end{equation}

To begin with, we need to define a convenient measure to count ``bad''
particles, i.e., those particles that do not evolve according to the effective non-linear
dynamics (\ref{eq:varphi-evolution}). 
For this purpose we introduced the orthogonal projectors 
 \begin{equation} P^\varphi_{k}=(q_{\cdot}^\varphi)^{\odot k}\odot
     (p_{\cdot}^\varphi)^{\odot(N-k)}, \tag{\ref{eq:projector}}
 \end{equation}
 for $0\leq k\leq N$.
To simplify our notation we use the convention
\begin{equation} P^\varphi_{k}\equiv0,\qquad\forall\ k\notin\left\{
  0,1,\ldots,N\right\}.\label{def-Pk-2}
\end{equation}
Later we will replace $\varphi$ by
the solution $\varphi_t$ of equation (\ref{eq:varphi-evolution}). One may then think of
$p_{\cdot}^{\varphi_t}$ as projecting on a ``good'' one-particle state
and $q_{\cdot}^{\varphi_t}$ as projecting on a ``bad'' one-particle state.

For an arbitrary weight function \[ w:\mathbb{Z}\to\mathbb{R}_{0}^{+}\]
we then define
weighted counting operators
\begin{equation}
  \widehat{w^{\varphi}}:=\sum_{k=0}^{N}w(k)P_{k}^{\varphi},\qquad\widehat{w_{d}^{\varphi}}:=\sum_{k=-d}^{N-d}w(k+d)P_{k}^{\varphi},
  \qquad d\in\mathbb Z.
  \label{eq:def-w}
\end{equation}
The role of the integer $d$ will become clear in
(\ref{eq:hat-Q-commutation-1}) and (\ref{eq:hat-Q-commutation-1bis})\ed. Note that, in the language introduced above,
$P_{k}^{\varphi}$ projects on that
part of the wave function that describes exactly $k$ bad particles. Hence,
one of the obvious candidates for a convenient counting measure is 
$\widehat{w^{\varphi}}$,  with $w(k)=k/N$. The expectation value
$\left\langle\Psi,\widehat{w^{\varphi}}\Psi\right\rangle$ then represents the
expected relative number of bad particles in the gas. 
However, control of this quantity will not suffice to track the excitation $\epsilon_t$: The total
number of particles in the gas is given by $N=\Lambda \rho$, and the number of
particles participating in an excitation is ${\cal O}(\rho)$. Consequently, we
will 
have to control the number of bad particles as compared to
$\rho$. This means that we have to adjust our
weight in a such a way that it counts the number of bad particles
relatively to $\rho$. The explicit weight function we use is given by
\begin{equation}
  m(k):=\begin{cases} \frac{k}{\rho} & \forall\,0\leq k\leq\rho\\ 1 &
    \forall\,\rho<k\\ 0 & \mbox{otherwise.}
  \end{cases}\label{eq:m-weight}
\end{equation}

When setting $w(k):=m(k)$ we denote the corresponding operator
$\widehat{w^\varphi}$ by $\widehat{m^\varphi}$.
Now, if  $\left\langle\Psi,\widehat{m^{\varphi}}\Psi\right\rangle$ is small, \ed
the probability of finding approximately $\rho$ bad particles in the gas is
small. As time goes by more and more particles in the gas will become
bad, due to interactions with other particles. Even for a perfect product
state there will always be a small
deviation of the true field from the mean field. The more bad particles there are
in the gas the stronger this deviation will be, and one may expect that the
rate of ``infection'' of formerly good particles is proportional to the number
of bad particles, up to a small term. The strategy of our proof is thus to show,
with the help of a Grönwall argument,
that if, initially, the number of bad particles is small, it will remain small for
any finite time interval. \\

Before we can start presenting the proofs of our results we must recall some properties of
the weighted counting measures, which have originally been studied 
in Lemma 1 in \cite{pickl2010a}. We summarize those properties that will be
needed in our analysis here while postponing their proofs to the appendix.
\begin{enumerate} \item
      \begin{equation}\label{eq:hat-multipilication}
        \widehat{v^{\varphi}}\widehat{w^{\varphi}}=\widehat{\left(vw\right)^{\varphi}}=\widehat{w^{\varphi}}\widehat{v^{\varphi}}
      \end{equation}

\item
  \begin{equation}
    \left[\widehat{w^{\varphi}},p_{k}^{\varphi}\right]=\left[\widehat{w^{\varphi}},q_{k}^{\varphi}\right]=0\label{eq:hat-pq-commutation}
  \end{equation}

\item \[ \left[\widehat{w^{\varphi}},P_{k}^{\varphi}\right]=0 \]

\item For $n(k)=\sqrt{\frac{k}{N}}$ we have
  \begin{equation}
    \left(\widehat{n^{\varphi}}\right)^{2}=\frac{1}{N}\sum_{k=1}^{N}q_{k}^{\varphi}\label{eq:old-q1}
  \end{equation}

\item For $\Psi\in\left(L^{2}\right)^{\odot N}$ we have that
  \begin{eqnarray}
    \left\Vert \widehat{w^{\varphi}}q_{1}^{\varphi}\Psi\right\Vert_2  & = &
    \left\Vert \widehat{w^{\varphi}}\widehat{n^{\varphi}}\Psi\right\Vert_2
    \label{eq:old-q1-2}\\ \left\Vert
    \widehat{w^{\varphi}}q_{1}^{\varphi}q_{2}^{\varphi}\Psi\right\Vert_2  & \leq &
    \sqrt{\frac{N}{N-1}}\left\Vert
    \widehat{w^{\varphi}}\left(\widehat{n^{\varphi}}\right)^{2}\Psi\right\Vert_2
    \label{eq:old-q1-3}
  \end{eqnarray}

\item For any  function $Y\in L^\infty(\mathbb{R}^3)$ and $Z\in L^\infty(\mathbb{R}^6)$ \ed and
  \[ A_{0}^{\varphi}=p_{1}^{\varphi},\qquad
    A_{1}^{\varphi}=q_{1}^{\varphi},\qquad
    B_{0}^{\varphi}=p_{1}^{\varphi}p_{2}^{\varphi},\qquad
    B_{1}^{\varphi}=p_{1}^{\varphi}q_{2}^{\varphi},\qquad
    B_{2}^{\varphi}=q_{1}^{\varphi}q_{2}^{\varphi} \] we have
    \begin{equation}
      \widehat{w^{\varphi}}A_{j}^{\varphi}Y(x_{1})A_{l}^{\varphi}=A_{j}^{\varphi}Y(x_{1})A_{l}^{\varphi}\widehat{w_{j-l}^{\varphi}}\quad \text{with}\,\, j,l=0,1,\label{eq:hat-Q-commutation-1}
\end{equation}
and
\begin{equation}
\widehat{w^{\varphi}}B_{j}^{\varphi}Z(x_{1},x_{2})B_{l}^{\varphi}=B_{j}^{\varphi}Z(x_{1},x_{2})B_{l}^{\varphi}\widehat{w_{j-l}^{\varphi}}\quad \text{with}\,\,j,l=0,1,2.\label{eq:hat-Q-commutation-1bis}
    \end{equation}
\end{enumerate}

In the following lemma the weighted number of bad particles encountered in the course of
time evolution is estimated. The proofs 
of our main results in Section~\ref{sec:proofs-main-results} rely on this fundamental lemma.
Another crucial point will be to justify assumption (\ref{eq:varphi-Linfty-bound-assumption})
below, which will be address in Section~\ref{sec:propagationestimates}.
\begin{lem} \label{lem:m-bound} 
    Let $U\in\mathcal C^\infty_c(\mathbb R^3,\mathbb R)$. 
    Let $\Psi_{t}$ be the solution to
    equation (\ref{eq:microscopic}) for initial data as in
    Condition \ref{def:initial-conditions}. Assume that, for some $T \leq \infty$, there
    is a $C\in\mathrm{Bounds}$ such that
  \begin{align}
    \label{eq:varphi-Linfty-bound-assumption}
    \Vert\varphi_t\Vert_\infty\leq C(t), \qquad 0\leq t< T.
  \end{align}
  \ed Then there is a
  $C\in\mathrm{Bounds}$ such that
  \begin{equation} 
     \left\langle
        \widehat{m^{\varphi_{t}}}\right\rangle_{t}  
      :=
      \left\langle
        \Psi_{t},\widehat{m^{\varphi_{t}}}\Psi_{t}\right\rangle
      \leq
        C(t)\frac{\Lambda}{\rho}, \qquad 0\leq t<T
        ,\label{eq:m-bound}
  \end{equation}
  where the weight function $m$ corresponding to counting operator
  $\widehat{m^{\varphi_t}}$ is defined in (\ref{eq:m-weight}).
\end{lem}
\begin{proof}
    The heart of the proof is a Grönwall argument for which we need
  to control the time derivative of $\left\langle
      \widehat{m^{\varphi_{t}}}\right\rangle _{t}$. Note that we have so-called
      ``intermediate picture'' here as 
      both the wave function and the operator are time dependent.
      
      The time derivative of $p_{k}^{\varphi_t}$ is given by
      $\frac{d}{dt}p_{k}^{\varphi_t}
      =-i[h_{x_{k}}[\varphi_{t}],p_{k}^{\varphi_t}]$ which can be seen best by
      noting that  
      in bra-ket notation $p_{k}^{\varphi_t}$ is given by
      $|\varphi_t\rangle\langle\varphi_t|$ acting on the $k^{\it th}$
      particle;
       see (\ref{eq:projectors}). Since $q_{k}^{\varphi_t}=1-p_{k}^{\varphi_t}$
       it follows that  $\frac{d}{dt}q_{k}^{\varphi_t}
       =-i[h_{x_{k}}[\varphi_{t}],q_{k}^{\varphi_t}]$. Consequently, as
       $P_k^{\varphi_t}$ is a symmetric product of $p$'s and $q$'s, one has
      $$\frac{d}{dt}P_k^{\varphi_t}=-i\left[\sum_{k=1}^{N}h_{x_{k}}[\varphi_{t}],P_k^{\varphi_t}\right]\;.$$
     Since any weighted counting operator is a sum of operators
     $P_k^{\varphi_t}$ multiplied by real numbers (see (\ref{eq:m-weight})), it
     follows that
     $\frac{d}{dt}\widehat{m^{\varphi_{t}}}=-i\left[\sum_{k=1}^{N}h_{x_{k}}[\varphi_{t}],\widehat{m^{\varphi_{t}}}\right]$ and thus           
  \begin{align} 
      \frac{d}{dt}\left\langle
    \widehat{m^{\varphi_{t}}}\right\rangle _{t} & =  i\left\langle
    \left[H-\sum_{k=1}^{N}h_{x_{k}}[\varphi_{t}],\widehat{m^{\varphi_{t}}}\right]\right\rangle
    _{t}\nonumber \\ & = i\left\langle \left[\frac{1}{\rho}\sum_{1\leq
    j<k\leq
  N}U(x_{j}-x_{k})-\sum_{k=1}^{N}\frac{N}{\rho}
U*\frac{|\varphi_{t}|^{2}}{\Lambda}(x_k),\widehat{m^{\varphi_{t}}}\right]\right\rangle
  _{t}.\label{eq:dt-m}
 \end{align}
  Using the symmetry in the bosonic degree
of freedom we find
\begin{align} \left|(\ref{eq:dt-m})\right| & \leq
  \frac{N(N-1)}{2\rho}\left|\left\langle
  \left[\underbrace{U(x_{1}-x_{2})-U*\frac{|\varphi_{t}|^{2}}{\Lambda}(x_{1})-U*\frac{|\varphi_{t}|^{2}}{\Lambda}(x_{2})}_{=:Z(x_{1},x_{2})},\widehat{m^{\varphi_{t}}}\right]\right\rangle
  _{t}\right|\label{eq:m-main-term}\\  &
  \qquad+\frac{N}{\rho}\left|\left\langle
  \left[\underbrace{U*\frac{|\varphi_{t}|^{2}}{\Lambda}(x_{1})}_{=:Y(x_{1})},\widehat{m^{\varphi_{t}}}\right]\right\rangle
  _{t}\right|.\label{eq:m-small-term}
\end{align} The first term, viz.
(\ref{eq:m-main-term}), in the expression above is the physically relevant one.
The second term, (\ref{eq:m-small-term}), only gives rise to a small correction.
But we shall estimate this term first, because this actually permits us to
demonstrate a crucial technique without too much additional ballast. We start
by inserting identity operators, in the form of $\mathrm{id}_{{\cal
H}}=p_{1}^{\varphi_{t}}+q_{1}^{\varphi_{t}}$, on the left- and right side
of the scalar product in (\ref{eq:m-small-term}), i.e.,
\begin{eqnarray}
 (\ref{eq:m-small-term}) & = & \frac{N}{\rho}\left|\left\langle
  \left(p_{1}^{\varphi_{t}}+q_{1}^{\varphi_{t}}\right)\left(Y(x_{1})\widehat{m^{\varphi_{t}}}-\widehat{m^{\varphi_{t}}}Y(x_{1})\right)\left(p_{1}^{\varphi_{t}}+q_{1}^{\varphi_{t}}\right)\right\rangle
  _{t}\right|.\label{eq:m-small-term-pq}\\
 &\leq &  \frac{N}{\rho}\left|\left\langle
  p_{1}^{\varphi_{t}}\left(Y(x_{1})\widehat{m^{\varphi_{t}}}-\widehat{m^{\varphi_{t}}}Y(x_{1})\right)p_{1}^{\varphi_{t}}\right\rangle
  _{t}\right|\label{eq:m-small-term-pq-1}
  \\
 & &+ \frac{N}{\rho}\left|\left\langle
  q_{1}^{\varphi_{t}}\left(Y(x_{1})\widehat{m^{\varphi_{t}}}-\widehat{m^{\varphi_{t}}}Y(x_{1})\right)q_{1}^{\varphi_{t}}\right\rangle
  _{t}\right|\label{eq:m-small-term-pq-2}
  \\
 &&+  \frac{2N}{\rho}\left|\left\langle
  p_{1}^{\varphi_{t}}\left(Y(x_{1})\widehat{m^{\varphi_{t}}}-\widehat{m^{\varphi_{t}}}Y(x_{1})\right)q_{1}^{\varphi_{t}}\right\rangle
  _{t}\right|\\
 &= &   \frac{2N}{\rho}\left|\left\langle
  p_{1}^{\varphi_{t}}\left(Y(x_{1})\widehat{m^{\varphi_{t}}}-\widehat{m^{\varphi_{t}}}Y(x_{1})\right)q_{1}^{\varphi_{t}}\right\rangle
  _{t}\right|\,.
  \end{eqnarray}
Here, (\ref{eq:m-small-term-pq-1} ) and (\ref{eq:m-small-term-pq-2}) are seen to be identically zero using (\ref{eq:hat-pq-commutation}) and 
(\ref{eq:hat-Q-commutation-1}) for $j=l=0$, e.g.,
$$p_{1}^{\varphi_{t}}Y(x_{1})\widehat{m^{\varphi_{t}}}p_{1}^{\varphi_{t}}=p_{1}^{\varphi_{t}}Y(x_{1})p_{1}^{\varphi_{t}}\widehat{m^{\varphi_{t}}}=\widehat{m^{\varphi_{t}}}p_{1}^{\varphi_{t}}Y(x_{1})p_{1}^{\varphi_{t}}=p_{1}^{\varphi_{t}}\widehat{m^{\varphi_{t}}}Y(x_{1})p_{1}^{\varphi_{t}}\,.$$
  
Without further
notice we will frequently use that \begin{align}
  \Vert \varphi_t \Vert_2^2 = \Lambda,
  \label{eq:varphi-L2-norm}
\end{align}
as implied by
(\ref{eq:psi-varphi-L2-norm}) and (\ref{eq:varphi-evolution}). 

Next, we apply the commutation relations in
(\ref{eq:hat-pq-commutation}) and after that the pull-through formula in
(\ref{eq:hat-Q-commutation-1}) for $j=0$ and $l=1$ to find
\begin{eqnarray}
  (\ref{eq:m-small-term}) & \leq & \frac{2N}{\rho}\left|\left\langle
  p_{1}^{\varphi_{t}}\left(Y(x_{1})\widehat{m^{\varphi_{t}}}-\widehat{m^{\varphi_{t}}}Y(x_{1})\right)q_{1}^{\varphi_{t}}\right\rangle
  _{t}\right| \label{eq:uno}\\ 
 &=&  \frac{2N}{\rho}\left|\left\langle
  p_{1}^{\varphi_{t}}Y(x_{1})q_{1}^{\varphi_{t}}\widehat{m^{\varphi_{t}}}-\widehat{m^{\varphi_{t}}}p_{1}^{\varphi_{t}}Y(x_{1})q_{1}^{\varphi_{t}}\right\rangle
 _{t}\right| \label{eq:due} \\ 
   & = & \frac{2N}{\rho}\left|\left\langle
  p_{1}^{\varphi_{t}}Y(x_{1})q_{1}^{\varphi_{t}}\left(\widehat{m^{\varphi_{t}}}-\widehat{m_{-1}^{\varphi_{t}}}\right)\right\rangle
  _{t}\right|.\label{eq:tre}
\end{eqnarray} Using the definition in (\ref{eq:def-w}) we find 
\begin{eqnarray} (\ref{eq:m-small-term}) & = &
  \frac{2N}{\rho}\left|\left\langle
  p_{1}^{\varphi_{t}}Y(x_{1})q_{1}^{\varphi_{t}}\left(\sum_{k=0}^{N}m(k)P^{\varphi_t}_k-\sum_{k=1}^{N+1}m(k-1)P^{\varphi_t}_k\right)\right\rangle
  _{t}\right| \label{second-term-1}\\ 
  &=& \frac{2N}{\rho}\left|\left\langle
  p_{1}^{\varphi_{t}}Y(x_{1})q_{1}^{\varphi_{t}}\left(\sum_{k=1}^{N}(m(k)-m(k-1))P^{\varphi_t}_k\right)\right\rangle
  _{t}\right| \label{second-term-2}\\ 
 & = & \frac{2N}{\rho}\left|\left\langle
  p_{1}^{\varphi_{t}}Y(x_{1})q_{1}^{\varphi_{t}}\left(\sum_{1\leq k\leq \rho}
  \frac{P^{\varphi_t}_k}{\rho}\right)\right\rangle
  _{t}\right| \label{second-term-3}\\ 
  & \leq & \frac{N}{\rho}C\left\Vert
  U*\frac{|\varphi_{t}|^{2}}{\Lambda}\right\Vert \frac{1}{\rho}\label{second-term-4}
  \\
   & \leq & \frac{C(t)}{\rho},
  \label{eq:small-term-one-over-rho}
\end{eqnarray} 
where we have used the following ingredients:
\begin{itemize}
\item for the step from (\ref{second-term-1}) to (\ref{second-term-2}) we have used that $m(0)=0$ and $P^{\varphi_t}_{N+1}=0$;
\item for the step from (\ref{second-term-2}) to (\ref{second-term-3}) we have used that $m(k)-m(k-1)=\frac{1}{\rho}$ for $k=1,\dots, \rho$ and $m(k)-m(k-1)=0$ for $k>\rho$; see (\ref{eq:m-weight});
\item for the step from (\ref{second-term-3}) to (\ref{second-term-4}) we have used  the definition of $Y(x_1)$ in (\ref{eq:m-small-term}) and 
    that $P^{\varphi_t}_k$, $1\leq k\leq N$, are pairwise orthogonal projectors;
\item in the last step we have made use of  assumption
(\ref{eq:varphi-Linfty-bound-assumption})
to infer the bound \[ \left\Vert U*|\varphi_{t}|^{2}\right\Vert \leq\left\Vert
  U\right\Vert _{1}\left\Vert \varphi_{t}\right\Vert _{\infty}^{2}\leq
  C(t)\left\Vert U\right\Vert _{1}.  \]
\end{itemize}

 In what comes next we will invoke assumption
  (\ref{eq:varphi-Linfty-bound-assumption}) without further mentioning.\\

A similar technique is used to estimate (\ref{eq:m-main-term}).  Again, we begin
by inserting identity operators, in the form of $\mathrm{id}_{{\cal
H}}=p_{1}^{\varphi_{t}}+q_{1}^{\varphi_{t}}$ and $\mathrm{id}_{{\cal
H}}=p_{2}^{\varphi_{t}}+q_{2}^{\varphi_{t}}$, in order to extract different
types of processes from the interaction which have to be treated separately:
\begin{eqnarray*} (\ref{eq:m-main-term}) & = &
  \frac{N(N-1)}{2\rho}\left|\left\langle
  Z(x_{1},x_{2})\widehat{m^{\varphi_{t}}}-\widehat{m^{\varphi_{t}}}Z(x_{1},x_{2})\right\rangle
  _{t}\right| \\ & = &
  \frac{N(N-1)}{2\rho}\bigg|\bigg<\left(p_{1}^{\varphi_{t}}+q_{1}^{\varphi_{t}}\right)\left(p_{2}^{\varphi_{t}}+q_{2}^{\varphi_{t}}\right)\times
  \\ &  &
  \qquad\times\left(Z(x_{1},x_{2})\widehat{m^{\varphi_{t}}}-\widehat{m^{\varphi_{t}}}Z(x_{1},x_{2})\right)\left(p_{1}^{\varphi_{t}}+q_{1}^{\varphi_{t}}\right)\left(p_{2}^{\varphi_{t}}+q_{2}^{\varphi_{t}}\right)\bigg>_{t}\bigg|.
\end{eqnarray*}
Due to symmetry 
\begin{eqnarray} (\ref{eq:m-main-term}) & \leq&
  \frac{N(N-1)}{2\rho}\bigg|\bigg<p_{1}^{\varphi_{t}}p_{2}^{\varphi_{t}}\left(Z(x_{1},x_{2})\widehat{m^{\varphi_{t}}}-\widehat{m^{\varphi_{t}}}Z(x_{1},x_{2})\right)p_{1}^{\varphi_{t}}p_{2}^{\varphi_{t}}\bigg>_{t}\bigg|
  \\&&+\frac{N(N-1)}{2\rho}\bigg|\bigg<\left(p_{1}^{\varphi_{t}}q_{2}^{\varphi_{t}}+q_{1}^{\varphi_{t}}p_{2}^{\varphi_{t}}\right)\left(Z(x_{1},x_{2})\widehat{m^{\varphi_{t}}}-\widehat{m^{\varphi_{t}}}Z(x_{1},x_{2})\right)\left(p_{1}^{\varphi_{t}}q_{2}^{\varphi_{t}}+q_{1}^{\varphi_{t}}p_{2}^{\varphi_{t}}\right)\bigg>_{t}\bigg|
      \\&&+\frac{N(N-1)}{2\rho}\bigg|\bigg<q_{1}^{\varphi_{t}}q_{2}^{\varphi_{t}}\left(Z(x_{1},x_{2})\widehat{m^{\varphi_{t}}}-\widehat{m^{\varphi_{t}}}Z(x_{1},x_{2})\right)q_{1}^{\varphi_{t}}q_{2}^{\varphi_{t}}\bigg>_{t}\bigg|
          \\&&+\frac{2N(N-1)}{\rho}\bigg|\bigg<p_{1}^{\varphi_{t}}p_{2}^{\varphi_{t}}\left(Z(x_{1},x_{2})\widehat{m^{\varphi_{t}}}-\widehat{m^{\varphi_{t}}}Z(x_{1},x_{2})\right)p_{1}^{\varphi_{t}}q_{2}^{\varphi_{t}}\bigg>_{t}\bigg|
              \\&&+\frac{N(N-1)}{\rho}\bigg|\bigg<p_{1}^{\varphi_{t}}p_{2}^{\varphi_{t}}\left(Z(x_{1},x_{2})\widehat{m^{\varphi_{t}}}-\widehat{m^{\varphi_{t}}}Z(x_{1},x_{2})\right)q_{1}^{\varphi_{t}}q_{2}^{\varphi_{t}}\bigg>_{t}\bigg|
                  \\&&+\frac{2N(N-1)}{\rho}\bigg|\bigg<p_{1}^{\varphi_{t}}q_{2}^{\varphi_{t}}\left(Z(x_{1},x_{2})\widehat{m^{\varphi_{t}}}-\widehat{m^{\varphi_{t}}}Z(x_{1},x_{2})\right)q_{1}^{\varphi_{t}}q_{2}^{\varphi_{t}}\bigg>_{t}\bigg|\;.
\end{eqnarray}

Using the pull-through formula in (\ref{eq:hat-Q-commutation-1bis}) and the commutation relations given in
(\ref{eq:hat-pq-commutation}) we can recast
the last expression to get that 
\begin{align} (\ref{eq:m-main-term}) 
   \leq & \quad C\frac{N(N-1)}{\rho}\left|\left\langle
    p_{1}^{\varphi_{t}}p_{2}^{\varphi_{t}}Z(x_{1},x_{2})p_{1}^{\varphi_{t}}p_{2}^{\varphi_{t}}\left(\widehat{m^{\varphi_{t}}}-\widehat{m^{\varphi_{t}}}\right)\right\rangle
    _{t}\right|\label{eq:term-0-1}\\ 
    &   +C\frac{N(N-1)}{\rho}\left|\left\langle
    \left(p_{1}^{\varphi_{t}}q_{2}^{\varphi_{t}}+q_{1}^{\varphi_{t}}p_{2}^{\varphi_{t}}\right)Z(x_{1},x_{2})\left(p_{1}^{\varphi_{t}}q_{2}^{\varphi_{t}}+q_{1}^{\varphi_{t}}p_{2}^{\varphi_{t}}\right)\left(\widehat{m^{\varphi_{t}}}-\widehat{m^{\varphi_{t}}}\right)\right\rangle
    _{t}\right|\label{eq:term-0-2}\\ 
    &   +C\frac{N(N-1)}{\rho}\left|\left\langle
    q_{1}^{\varphi_{t}}q_{2}^{\varphi_{t}}Z(x_{1},x_{2})q_{1}^{\varphi_{t}}q_{2}^{\varphi_{t}}\left(\widehat{m^{\varphi_{t}}}-\widehat{m^{\varphi_{t}}}\right)\right\rangle
    _{t}\right|\label{eq:term-0-3}\\
    &+C\frac{N(N-1)}{\rho}\left|\left\langle
  p_{1}^{\varphi_{t}}p_{2}^{\varphi_{t}}Z(x_{1},x_{2})p_{1}^{\varphi_{t}}q_{2}^{\varphi_{t}}\left(\widehat{m^{\varphi_{t}}}-\widehat{m_{-1}^{\varphi_{t}}}\right)\right\rangle
  _{t}\right|\label{eq:term-I}\\ 
 &  +C\frac{N(N-1)}{\rho}\left|\left\langle
  p_{1}^{\varphi_{t}}p_{2}^{\varphi_{t}}Z(x_{1},x_{2})q_{1}^{\varphi_{t}}q_{2}^{\varphi_{t}}\left(\widehat{m^{\varphi_{t}}}-\widehat{m_{-2}^{\varphi_{t}}}\right)\right\rangle
  _{t}\right|\label{eq:term-II}\\ 
 &  +C\frac{N(N-1)}{\rho}\left|\left\langle
  p_{1}^{\varphi_{t}}q_{2}^{\varphi_{t}}Z(x_{1},x_{2})q_{1}^{\varphi_{t}}q_{2}^{\varphi_{t}}\left(\widehat{m^{\varphi_{t}}}-\widehat{m_{-1}^{\varphi_{t}}}\right)\right\rangle
  _{t}\right|\label{eq:term-III}\;.
\end{align} 
Lines (\ref{eq:term-0-1})-(\ref{eq:term-0-3}) all contain the factor
$\left(\widehat{m^{\varphi_{t}}}-\widehat{m^{\varphi_{t}}}\right)$. Hence, they are identically equal to zero. 
\ed
 In the following we provide
estimates for the terms (\ref{eq:term-I})-(\ref{eq:term-III}). We
use that, for any $f\in L^2$,
\begin{equation}\label{eq:pfp}
  p_{1}^{\varphi_{t}} f(x_1-x_2)p_1^{\varphi_{t}} 
  =
  p_{1}^{\varphi_{t}} \int dx_1 \, \frac{\varphi_t^*(x_1)}{\|\varphi_t\|_2}
  f(x_1-x_2) \frac{\varphi_t(x_1)}{\|\varphi_t\|_2} p_{1}^{\varphi_{t}} 
  =
  \Lambda^{-1}f*|\varphi_{t}|^{2}(x_{2}) p_{1}^{\varphi_{t}}\;,
\end{equation} 
holds so that we can estimate
\begin{equation}\label{eq:fpop}
  \Vert p_{1}^{\varphi_{t}} f(x_1-x_2)\Vert =
  \left\Vert p_{1}^{\varphi_{t}} |f(x_1-x_2)|^2  p_{1}^{\varphi_{t}} \right\Vert^{1/2}
\leq C(t)\Lambda^{-1/2} \Vert f\Vert_2
\end{equation}
and
\begin{equation}\label{eq:fpop2}
\Vert p_{1}^{\varphi_{t}} f(x_1) \Vert =\Vert p_{1}^{\varphi_{t}} 
  |f(x_1)|^2  p_{1}^{\varphi_{t}} \Vert^{1/2}
\leq C(t)\Lambda^{-1/2} \Vert f\Vert_2\;.
\end{equation}
\ed

\noun{Term (\ref{eq:term-I}):} Using (\ref{eq:pfp}), the equation
\emph{\noun{\begin{multline*}
  p_{1}^{\varphi_{t}}p_{2}^{\varphi_{t}}Z(x_{1},x_{2})p_{1}^{\varphi_{t}}q_{2}^{\varphi_{t}}\\
  =p_{1}^{\varphi_{t}}p_{2}^{\varphi_{t}}
  \left(\underbrace{p_{1}^{\varphi_{t}}U(x_{1}-x_{2})p_{1}^{\varphi_{t}}}_{=\Lambda^{-1}U*|\varphi_{t}|^{2}(x_{2})
  p_{1}^{\varphi_{t}}}-U*\frac{|\varphi_{t}|^{2}}{\Lambda}(x_{2})p_{1}^{\varphi_{t}}\right)
  q_{2}^{\varphi_{t}}-p_{1}^{\varphi_{t}}U*\frac{|\varphi_{t}|^{2}}{\Lambda}(x_{1})p_{1}^{\varphi_{t}}\underbrace{p_{2}^{\varphi_{t}}q_{2}^{\varphi_{t}}}_{=0}=0
\end{multline*} }}implies that
\begin{equation}
  (\ref{eq:term-I})=0.\label{eq:term-I-final-bound}
\end{equation}

\noun{Term (\ref{eq:term-III}):}
We need some preliminary results on operator norms and $L^2$-norms that are used in the next steps. By (\ref{eq:fpop}) we can estimate
\begin{align*}
  \Vert p_1^{\varphi_t} U(x_1-x_2) \Vert  
  \leq C(t)\Lambda^{-1/2}.
\end{align*}
Furthermore, using Young's inequality and the conservation of the $L_2$-norm of $\varphi_t$ we get
\begin{align}
\label{eq:u-conv-bound}
 \left\Vert U*\frac{|\varphi_t|^2}{\Lambda} \right\Vert_2 \leq \Vert U \Vert_1 \,\frac{\|\varphi_t^2\|_2}{\Lambda}\leq \Vert U \Vert_1 \, \frac{\Vert
  \varphi_t \Vert_\infty\|\varphi_t\|_2}{\Lambda} \leq  \Vert U \Vert_1 \, \|\frac{\Vert
  \varphi_t \Vert_\infty}{\Lambda^{\frac{1}{2}}}\;.
 \end{align}
Finally, starting from the definition of $Z(x_1,x_2)$ in (\ref{eq:m-main-term}), (\ref{eq:fpop2}) and (\ref{eq:u-conv-bound}) are seen to imply
\begin{equation}\label{eq:u-conv-bound-1}
\Vert p_1^{\varphi_t}Z(x_1,x_2)\Vert \leq \frac{C(t)}{\Lambda^{1/2}}.
\end{equation}

 Next, let $r:\mathbb{Z}\to\mathbb{R}^+_0$ be given by $r(k):=\sqrt{m(k)-m(k-1)}$ 
which is well defined because $m(k)$ is monotone increasing. Relation (\ref{eq:hat-multipilication}) 
implies that
$\left(\widehat{r^{\varphi_{t}}}\right)^2=\widehat{m^{\varphi_{t}}}-\widehat{m_{-1}^{\varphi_{t}}}$. Then we can write
  \begin{eqnarray}   (\ref{eq:term-III})& = &
     C\frac{N(N-1)}{\rho}\left|\left\langle
    p_{1}^{\varphi_{t}}q_{2}^{\varphi_{t}}Z(x_{1},x_{2})q_{1}^{\varphi_{t}}q_{2}^{\varphi_{t}}\left(\widehat{r^{\varphi_{t}}}\right)^2\right\rangle
    _{t}\right|\nonumber  \\
 &=&   C\frac{N(N-1)}{\rho}\left|\left\langle
    p_{1}^{\varphi_{t}}q_{2}^{\varphi_{t}}Z(x_{1},x_{2})q_{1}^{\varphi_{t}}q_{2}^{\varphi_{t}}  \widehat{r^{\varphi_{t}}} \,\widehat{r^{\varphi_{t}}}\right\rangle
    _{t}\right|
\nonumber    
 \end{eqnarray} 
Using the pull-through formula in (\ref{eq:hat-Q-commutation-1bis}) with $j=1$ and $l=2$ we get that 
\noun{
  \begin{eqnarray} (\ref{eq:term-III}) & = &
     C\frac{N(N-1)}{\rho}\left|\left\langle
    \widehat{r_{1}^{\varphi_{t}}}p_{1}^{\varphi_{t}}q_{2}^{\varphi_{t}}Z(x_{1},x_{2})q_{1}^{\varphi_{t}}q_{2}^{\varphi_{t}}\widehat{r^{\varphi_{t}}}\right\rangle
    _{t}\right|\nonumber   \end{eqnarray} }
Finally, using the commutation relations in (\ref{eq:hat-pq-commutation}), the   bounds in (\ref{eq:u-conv-bound-1}), and Schwartz inequality we can estimate
\noun{
  \begin{eqnarray} (\ref{eq:term-III}) 
  & =&  C\frac{N(N-1)}{\rho}\left|\left\langle
    \widehat{r_{1}^{\varphi_{t}}}p_{1}^{\varphi_{t}}q_{2}^{\varphi_{t}}Z(x_{1},x_{2})q_{1}^{\varphi_{t}}q_{2}^{\varphi_{t}}\widehat{r^{\varphi_{t}}}\right\rangle
    _{t}\right|\nonumber\\
   & = & C\frac{N(N-1)}{\rho}\left|\left\langle
    q_{2}^{\varphi_{t}}\widehat{r_{1}^{\varphi_{t}}}p_{1}^{\varphi_{t}}Z(x_{1},x_{2})\widehat{r^{\varphi_{t}}}q_{1}^{\varphi_{t}}q_{2}^{\varphi_{t}}\right\rangle
    _{t}\right|\nonumber \\ 
     & \leq & C\frac{N(N-1)}{\rho}\left\Vert
    \widehat{r_{1}^{\varphi_{t}}}q_{2}^{\varphi_{t}}\Psi_{t}\right\Vert
    _{2}\;\;\;\Big\|p_1^{\varphi_t}Z(x_1, x_2)\Big\|\left\Vert
    \widehat{r^{\varphi_{t}}}q_{1}^{\varphi_{t}}q_{2}^{\varphi_{t}}\Psi_{t}\right\Vert
    _{2}\nonumber\\
    & \leq & C\frac{N(N-1)}{\rho}\left\Vert
    \widehat{r_{1}^{\varphi_{t}}}q_{2}^{\varphi_{t}}\Psi_{t}\right\Vert
    _{2}\;\;\;\frac{C(t)}{\Lambda^{1/2}}\left\Vert
    \widehat{r^{\varphi_{t}}}q_{1}^{\varphi_{t}}q_{2}^{\varphi_{t}}\Psi_{t}\right\Vert
    _{2}\label{eq:term-III-1}\;.
\end{eqnarray} }
   Using properties (\ref{eq:old-q1}) and
    (\ref{eq:old-q1-2}) of the counting measures and the definitions
 in (\ref{eq:m-weight}) and (\ref{eq:def-w}) we find that
    \begin{eqnarray}
      \left\Vert
       \widehat{r_{1}^{\varphi_{t}}}q_{2}^{\varphi_{t}}\Psi_{t}\right\Vert
      _{2} & = & \left\Vert
     \widehat{r_{1}^{\varphi_{t}}}\widehat{n^{\varphi_{t}}}\Psi_{t}\right\Vert
      _{2}\nonumber \\ & = & \left\Vert
      \sum_{k=1}^{N-1}\left(\left[m(k+1)-m(k)\right]\frac{k}{N}\right)^{1/2}P_{k}^{\varphi_{t}}\Psi_{t}\right\Vert
      _{2}\nonumber \\ & \leq & \frac{C}{N^{1/2}}\left\Vert
      \sum_{0\leq k<\rho}\left(\frac{k}{\rho}\right)^{1/2}P_{k}^{\varphi_{t}}\Psi_{t}\right\Vert _{2}\nonumber
      \\ & \leq & C\left(\frac{\left\langle
        \widehat{m^{\varphi_{t}}}\right\rangle
        _{t}}{N}\right)^{1/2},\label{eq:term-III-mprimeq_estimate}
      \end{eqnarray} 
      where we have used that $m(k)-m(k-1)=\frac{1}{\rho}$ for $k=1,\dots, \rho$ and $m(k)-m(k-1)=0$ for $k>\rho$.
      Quite similarly, and using (\ref{eq:old-q1-3}), we see that
      \begin{eqnarray}
        \left\Vert
         \widehat{r^{\varphi_{t}}}q_{1}^{\varphi_{t}}q_{2}^{\varphi_{t}}\Psi_{t}\right\Vert
        _{2} & \leq & \sqrt{\frac{N}{N-1}}\left\Vert
        \left(\widehat{m^{\varphi_{t}}}-\widehat{m_{-1}^{\varphi_{t}}}\right)^{1/2}\left(\widehat{n^{\varphi_{t}}}\right)^{2}
        \Psi_{t}\right\Vert
        _{2} \label{eq:-eq:term-I-one}\\ & = & \left\Vert
        \sum_{k=1}^{N}\left(\left[m(k)-m(k-1)\right]\frac{k^{2}}{N^{2}}\right)^{1/2}P_{k}^{\varphi_t}\Psi_{t}\right\Vert _{2} \label{eq:-eq:term-I-two}\\ & \leq &
        \frac{C}{N^{1/2}}\left\Vert
        \sum_{0\leq
        k<\rho}\left(\frac{k}{\rho}\frac{k}{N}\right)^{1/2}P^{\varphi_t}_{k}\Psi_{t}\right\Vert _{2}\nonumber \\ & \leq &
        C\left(\frac{\left\langle \widehat{m^{\varphi_{t}}}\right\rangle
        _{t}}{N}\right)^{1/2}\left(\frac{\rho}{N}\right)^{1/2}.\label{eq:term-III-2}
      \end{eqnarray} 
     As a consequence, going back to (\ref{eq:term-III-1}), the
      bounds (\ref{eq:term-III-mprimeq_estimate}), (\ref{eq:term-III-2}), and
      (\ref{eq:varphi-Linfty-bound}) are seen to imply
      \begin{equation}
        (\ref{eq:term-III})\leq C\frac{N^{2}}{\rho}\left(\frac{\left\langle
          \widehat{m^{\varphi_{t}}}\right\rangle
          _{t}}{N}\right)^{1/2}\frac{C(t)}{\Lambda^{1/2}}\left(\frac{\left\langle
            \widehat{m^{\varphi_{t}}}\right\rangle
            _{t}}{N}\right)^{1/2}\left(\frac{\rho}{N}\right)^{1/2}\leq
            C(t)\left\langle \widehat{m^{\varphi_{t}}}\right\rangle
            _{t}.\label{eq:term-III-final-bound}
          \end{equation}

\noun{Term (\ref{eq:term-II}): }Again, we write
$\left(\widehat{m^{\varphi_{t}}}-\widehat{m_{-2}^{\varphi_{t}}}\right)$ as the 
square of its square root and we use the pull-through formula in (\ref{eq:hat-Q-commutation-1bis}) for $j=0$ and $l=2$:

\begin{eqnarray} (\ref{eq:term-II}) & = &
  C\frac{N(N-1)}{\rho}\left|\left\langle
  p_{1}^{\varphi_{t}}p_{2}^{\varphi_{t}}Z(x_{1},x_{2})q_{1}^{\varphi_{t}}q_{2}^{\varphi_{t}}\left(\widehat{m^{\varphi_{t}}}-\widehat{m_{-2}^{\varphi_{t}}}\right)\right\rangle
  _{t}\right|\nonumber \\ & = &
  C\frac{N(N-1)}{\rho}\bigg|\bigg<\left(\widehat{m_{2}^{\varphi_{t}}}-\widehat{m^{\varphi_{t}}}\right)^{1/2}p_{1}^{\varphi_{t}}p_{2}^{\varphi_{t}}\times\nonumber\\
  &  & \qquad\qquad\qquad\qquad\qquad\qquad\times
  Z(x_{1},x_{2})q_{1}^{\varphi_{t}}q_{2}^{\varphi_{t}}\left(\widehat{m^{\varphi_{t}}}-\widehat{m_{-2}^{\varphi_{t}}}\right)^{1/2}\bigg>_{t}.\nonumber
\end{eqnarray} 
Next, we use the symmetry in the bosonic degrees of freedom of the wave function $\Psi_t$ and of the counting measures to arrive at
\begin{eqnarray} (\ref{eq:term-II}) & = &
  C\frac{N}{\rho}\bigg|\bigg<\left(\widehat{m_{2}^{\varphi_{t}}}-\widehat{m^{\varphi_{t}}}\right)^{1/2}
  p_{1}^{\varphi_{t}}\sum_{k=2}^{N}p_{k}^{\varphi_{t}}\times\nonumber\\
  &  & \qquad\qquad\qquad\qquad\times
  Z(x_{1},x_{k})q_{k}^{\varphi_{t}}q_{1}^{\varphi_{t}}
  \left(\widehat{m^{\varphi_{t}}}-\widehat{m_{-2}^{\varphi_{t}}}\right)^{1/2}\bigg>_{t},\nonumber
\end{eqnarray} 
and finally use Schwarz inequality\noun{\begin{eqnarray}
  (\ref{eq:term-II}) & \leq &
  C\frac{N}{\rho}\left\Vert
  \sum_{k=2}^{N}q_{k}^{\varphi_{t}}Z(x_{1},x_{k})p_{k}^{\varphi_{t}}p_{1}^{\varphi_{t}}
  \left(\widehat{m^{\varphi_{t}}_{2}}-\widehat{m^{\varphi_{t}}}\right)^{1/2}\Psi_{t}\right\Vert
  _{2}\times\label{eq:term-II-L2norm}\\ &  &
  \qquad\qquad\qquad\qquad\times\left\Vert
  q_{1}^{\varphi_{t}}\left(\widehat{m^{\varphi_{t}}}-\widehat{m_{-2}^{\varphi_{t}}}\right)^{1/2}\Psi_{t}\right\Vert
  _{2}.\label{eq:term-II-1}
\end{eqnarray} }
 Furthermore, a computation similar to the one leading to (\ref{eq:term-III-mprimeq_estimate}) shows that
\begin{eqnarray}
  (\ref{eq:term-II-1})=\left\Vert
  q_{1}^{\varphi_{t}}\left(\widehat{m^{\varphi_{t}}}-\widehat{m_{-2}^{\varphi_{t}}}\right)^{1/2}\Psi_{t}\right\Vert
  _{2} & \leq & C\left(\frac{\left\langle
    \widehat{m^{\varphi_{t}}}\right\rangle _{t}}{N}\right)^{1/2}.\label{estim.qmm}
  \end{eqnarray} 
  Next, we estimate the square of the $L^{2}-$ norm in
  (\ref{eq:term-II-L2norm}).  In order to obtain a good estimate, we rewrite this
  expression according to
  \begin{align} & \left\Vert
    \sum_{k=2}^{N}q_{k}^{\varphi_{t}}Z(x_{1},x_{k})p_{k}^{\varphi_{t}}p_{1}^{\varphi_{t}}
    \left(\widehat{m_{2}^{\varphi_{t}}}-\widehat{m^{\varphi_{t}}}\right)^{1/2}\Psi_{t}\right\Vert
    _{2}^{2}\label{eq:term-II-1.5}\\ = &
    \sum_{k=2}^{N}\bigg\langle\left(\widehat{m_{2}^{\varphi_{t}}}-\widehat{m^{\varphi_{t}}}\right)^{1/2}p_{1}^{\varphi_{t}}p_{k}^{\varphi_{t}}Z(x_{1},x_{k})q_{k}^{\varphi_{t}}\times\nonumber
    \\ & \qquad\qquad\qquad\qquad\qquad\qquad\times
    Z(x_{1},x_{k})p_{k}^{\varphi_{t}}p_{1}^{\varphi_{t}}
    \left(\widehat{m_{2}^{\varphi_{t}}}-\widehat{m^{\varphi_{t}}}\right)^{1/2}\bigg\rangle_{t}\nonumber
    \\ & +\sum_{j,k=2,j\neq
    k}^{N}\bigg\langle\left(\widehat{m_{2}^{\varphi_{t}}}-\widehat{m^{\varphi_{t}}}\right)^{1/2}p_{1}^{\varphi_{t}}p_{k}^{\varphi_{t}}Z(x_{1},x_{k})q_{k}^{\varphi_{t}}q_{j}^{\varphi_{t}}\times\nonumber
    \\ & \qquad\qquad\qquad\qquad\qquad\qquad\times
    Z(x_{1},x_{j})p_{j}^{\varphi_{t}}p_{1}^{\varphi_{t}}
    \left(\widehat{m_{2}^{\varphi_{t}}}-\widehat{m^{\varphi_{t}}}\right)^{1/2}\bigg\rangle_{t}.\nonumber
  \end{align} 
  Furthermore, we exploit the symmetry in the bosonic degrees of freedom
  and split the summations into diagonal- and off-diagonal parts, with the
  result that 
  \begin{align} 
      (\ref{eq:term-II-1.5})\leq &
    N\bigg\langle\left(\widehat{m_{2}^{\varphi_{t}}}-\widehat{m^{\varphi_{t}}}\right)^{1/2}p_{1}^{\varphi_{t}}p_{2}^{\varphi_{t}}Z(x_{1},x_{2})q_{2}^{\varphi_{t}}
    Z(x_{1},x_{2})p_{2}^{\varphi_{t}}p_{1}^{\varphi_{t}}\left(\widehat{m_{2}^{\varphi_{t}}}-\widehat{m^{\varphi_{t}}}\right)^{1/2}\bigg\rangle_{t}
    \label{eq:term-II-2-2}\\
    &
    +N^{2}\bigg\langle\left(\widehat{m_{2}^{\varphi_{t}}}-\widehat{m^{\varphi_{t}}}\right)^{1/2}p_{1}^{\varphi_{t}}p_{2}^{\varphi_{t}}Z(x_{1},x_{2})q_{2}^{\varphi_{t}}q_{3}^{\varphi_{t}}
    Z(x_{1},x_{3})p_{3}^{\varphi_{t}}p_{1}^{\varphi_{t}}\left(\widehat{m_{2}^{\varphi_{t}}}-\widehat{m^{\varphi_{t}}}\right)^{1/2}\bigg\rangle_{t}.\label{eq:term-II-3-2}
  \end{align} 
   Using (\ref{eq:u-conv-bound}) we find
   \begin{align}
      \left\Vert
    Z(x_{1},x_{2})p_{1}^{\varphi_{t}}p_{2}^{\varphi_{t}}\right\Vert
    \leq  \left\Vert
    Z(x_{1},x_{2})p_{1}^{\varphi_{t}} \right\Vert \;  \Vert p_{2}^{\varphi_{t}}\Vert\leq
    \frac{C(t)}{\Lambda^{1/2}},\label{eq:Z12p1p2-bound}
  \end{align}
  We observe also that, using the definitions in (\ref{eq:def-w})
  and (\ref{eq:m-weight}), for any $\Psi$ with $\|\Psi\|_2=1$ one has
 \begin{eqnarray}
  \left\Vert
  \left(\widehat{m^{\varphi_{t}}}-\widehat{m_{-2}^{\varphi_{t}}}\right)^{1/2}\Psi \right\Vert^2_{2} & = &  \left\langle \Psi,   (\widehat{m^{\varphi_{t}}}-\widehat{m_{-2}^{\varphi_{t}}}) \Psi\right\rangle  \\
  &= &   \left\langle \Psi, \left(\sum_{k=0}^{N}m(k)P_{k}^{\varphi_{t}}-
  \sum_{k=2}^{N+2}m(k-2)P_{k}^{\varphi_{t}} \right)\Psi \right\rangle \\
  &\leq  & C \left\langle
  \Psi,\left(\sum_{k=0}^{N}\frac{1}{\rho}P_{k}^{\varphi_{t}} \right)\Psi \right\rangle \\
  &\leq & \frac{C}{\rho}\label{eq:norm-difference-counting-measures}
  \end{eqnarray} 
 because $\sum_{k=0}^{N}P_{k}^{\varphi_{t}}$ coincides with the identity operator.
Therefore, using (\ref{eq:Z12p1p2-bound}) and (\ref{eq:norm-difference-counting-measures}), we can estimate the diagonal terms by
  \begin{eqnarray}
    (\ref{eq:term-II-2-2}) & \leq & CN\left\Vert
    \left(\widehat{m^{\varphi_{t}}_{2}}-\widehat{m^{\varphi_{t}}}\right)^{1/2}\right\Vert
    ^{2}\left\Vert
    Z(x_{1},x_{2})p_{2}^{\varphi_{t}}p_{1}^{\varphi_{t}}\right\Vert
    _{2}^{2}\nonumber \\ & \leq & CN\left\Vert
   \left(\widehat{m^{\varphi_{t}}_{2}}-\widehat{m^{\varphi_{t}}}\right)^{1/2}\right\Vert
    ^{2} \frac{C(t)}{\Lambda} \nonumber \\ & \leq &
    C(t).\label{eq:term-II-4}
  \end{eqnarray}  For the
  off-diagonal terms we find
  \begin{eqnarray}
    (\ref{eq:term-II-3-2}) & = &
    N^{2}\bigg<\left(\widehat{m_{2}^{\varphi_{t}}}-\widehat{m^{\varphi_{t}}}\right)^{1/2}q_{3}^{\varphi_{t}}p_{1}^{\varphi_{t}}p_{2}^{\varphi_{t}}Z(x_{1},x_{2})\times
    \nonumber
    \\
    &  & \qquad\qquad\qquad\qquad\times
    Z(x_{1},x_{3})p_{3}^{\varphi_{t}}p_{1}^{\varphi_{t}}q_{2}^{\varphi_{t}}
    \left(\widehat{m_{2}^{\varphi_{t}}}-\widehat{m^{\varphi_{t}}}\right)^{1/2}\bigg>_{t}
    \nonumber
    \\
    & \leq & N^{2}\left\Vert
    q_{3}^{\varphi_{t}}\left(\widehat{m_{2}^{\varphi_{t}}}-\widehat{m^{\varphi_{t}}}\right)^{1/2}\Psi_{t}\right\Vert
    _{2}
    \nonumber
    \\
    & & \qquad\qquad\times
    \left\Vert
    p_{1}^{\varphi_{t}}p_{2}^{\varphi_{t}}Z(x_{1},x_{2})    Z(x_{1},x_{3})p_{3}^{\varphi_{t}}p_{1}^{\varphi_{t}}\right\Vert \times
    \label{eq:Z-term}
    \\ &
    & \qquad\qquad\qquad\qquad\times\left\Vert
    q_{2}^{\varphi_{t}}\left(\widehat{m_{2}^{\varphi_{t}}}-\widehat{m^{\varphi_{t}}}\right)^{1/2}\Psi_{t}\right\Vert
    _{2}.
    \nonumber
  \end{eqnarray}
  Here it becomes apparent why the splitting of (\ref{eq:term-II-1.5}) into a
  diagonal- and an off-diagonal part is necessary: A rough estimate of the term
  (\ref{eq:Z-term}), using (\ref{eq:term-II-1.5}), leads to a
  $\Lambda^{-1}-$decay. As it will turn out in
  (\ref{eq:term-II-final-bound}), this decay is not good enough. Fortunately,
  the situation is better than that, as
  the following analysis shows. First, we note that for non-negative $U$ one
  finds
  \begin{align}
    & \left\Vert
      p_{1}^{\varphi_{t}}p_{2}^{\varphi_{t}}U(x_{1}-x_{2})  U(x_{1}-x_{3})p_{3}^{\varphi_{t}}p_{1}^{\varphi_{t}}
    \right\Vert
    \nonumber
    \\
    & = \left\Vert
      p_{1}^{\varphi_{t}}p_{2}^{\varphi_{t}}\sqrt{U(x_{1}-x_{3})}\sqrt{U(x_{1}-x_{2})}
      \sqrt{U(x_{1}-x_{3})}\sqrt{U(x_{1}-x_{2})}p_{3}^{\varphi_{t}}p_{1}^{\varphi_{t}}
    \right\Vert 
    \nonumber
    \\
    & = \left\Vert
      p_{1}^{\varphi_{t}}\sqrt{U(x_{1}-x_{3})}p_{2}^{\varphi_{t}}\sqrt{U(x_{1}-x_{2})}
      \sqrt{U(x_{1}-x_{3})}p_{3}^{\varphi_{t}}\sqrt{U(x_{1}-x_{2})}p_{1}^{\varphi_{t}}
    \right\Vert 
    \nonumber
    \\
    & \leq \left\Vert p_{1}^{\varphi_{t}}\sqrt{U(x_{1}-x_{3})} \right\Vert^4_2 
    \leq \frac{C(t)}{\Lambda^2}\Vert U\Vert_1^2,
    \label{eq:ppUUpp-term}
  \end{align}
  where in the last step we have used (\ref{eq:fpop}) and $\Vert \sqrt
  U\Vert_2^2=\Vert U\Vert_1$. Choosing the branch cut of the square root
  conveniently one observes that the formula holds for general $U$.
  
  Second, due to (\ref{eq:fpop}) and (\ref{eq:fpop2})
  \begin{align*}
    \Vert p_j p_k U(x_j - x_k) \Vert
    \leq \Vert p_k U(x_j - x_k) \Vert
   \leq  \frac{C(t)}{\Lambda^{1/2}},
    \\
  \left\Vert p_j p_k \frac{|\varphi_t|^2}{\Lambda} (x_j) \right\Vert \leq
   \left\Vert p_j \frac{|\varphi_t|^2}{\Lambda} (x_j) \right\Vert \leq  \frac{C(t)}{\Lambda^{3/2}}
  \end{align*}
 that together with (\ref{eq:ppUUpp-term}) imply 
  \begin{equation}\label{estim.ppZZpp}
    \Vert p_1 p_2 Z(x_1,x_2) Z(x_1,x_3) p_1 p_3 \Vert
    \leq
    \frac{C(t)}{\Lambda^2}.
  \end{equation}
 Analogously to (\ref{estim.qmm}),  one can prove that
  \begin{equation}
 \left\Vert
    q_{k}^{\varphi_{t}}\left(\widehat{m_{2}^{\varphi_{t}}}-\widehat{m^{\varphi_{t}}}\right)^{1/2}\Psi_{t}\right\Vert_{2}\leq C\left(\frac{\left\langle
    \widehat{m^{\varphi_{t}}}\right\rangle _{t}}{N}\right)^{1/2}.\label{estim.qmm-bis}
  \end{equation}
 Hence, invoking the estimates in (\ref{estim.qmm-bis}) and (\ref{estim.ppZZpp}), we arrive at
  \begin{eqnarray} 
    (\ref{eq:term-II-3-2}) & \leq &
    N^{2}\left\Vert
    q_{3}^{\varphi_{t}}\left(\widehat{m_{2}^{\varphi_{t}}}-\widehat{m^{\varphi_{t}}}\right)^{1/2}\Psi_{t}\right\Vert
    _{2}
   \frac{C(t)}{\Lambda^{2}}
    \left\Vert
    q_{2}^{\varphi_{t}}\left(\widehat{m_{2}^{\varphi_{t}}}-\widehat{m^{\varphi_{t}}}\right)^{1/2}\Psi_{t}\right\Vert
    _{2}\nonumber \\ & \leq & C(t)N^{2}\left(\frac{\left\langle
      \widehat{m^{\varphi_{t}}}\right\rangle
      _{t}}{N}\right)^{1/2}\frac{1}{\Lambda^{2}}\left(\frac{\left\langle
        \widehat{m^{\varphi_{t}}}\right\rangle _{t}}{N}\right)^{1/2}\nonumber
        \\ & \leq & C(t)N\frac{1}{\Lambda^{2}}\left\langle
        \widehat{m^{\varphi_{t}}}\right\rangle \label{eq:term-II-3-2-1}\;.
      \end{eqnarray} 
       Thus
      \begin{eqnarray} (\ref{eq:term-II}) & \leq &
        C(t)\frac{N}{\rho}\sqrt{(\ref{eq:term-II-1.5})}\times(\ref{eq:term-II-1})\nonumber
        \\ & \leq &
        C(t)\frac{N}{\rho}\sqrt{(\ref{eq:term-II-4})+(\ref{eq:term-II-3-2-1})}\times(\ref{eq:term-II-1})\nonumber
        \\ & \leq &
        C(t)\frac{N}{\rho}\left(1+N\frac{1}{\Lambda^{2}}\left\langle
        \widehat{m^{\varphi_{t}}}\right\rangle
        \right)^{1/2}\left(\frac{\left\langle
          \widehat{m^{\varphi_{t}}}\right\rangle
          _{t}}{N}\right)^{1/2}\nonumber \\ & \leq &
          C(t)\left(\frac{\Lambda}{\rho}+\left\langle
          \widehat{m^{\varphi_{t}}}\right\rangle
          _{t}\right).\label{eq:term-II-final-bound}
        \end{eqnarray} The bounds (\ref{eq:small-term-one-over-rho}),
        (\ref{eq:term-I-final-bound}), (\ref{eq:term-III-final-bound}), and
        (\ref{eq:term-II-final-bound}) yield
        \begin{eqnarray*}
          \left|\frac{d}{dt}\left\langle
          \widehat{m^{\varphi_{t}}}\right\rangle_t \right| & \leq &
          (\ref{eq:m-main-term})+(\ref{eq:m-small-term})\\
          & \leq & 
          C(t)\left(\left\langle\widehat{m^{\varphi_{t}}}\right\rangle_t +\frac{1+\Lambda}{\rho}\right).
        \end{eqnarray*}
        Finally, for any initial wave function $\Psi_0$ with the property that
        \begin{align}
\label{eq:ini-gronwall}
          \left\langle
          \widehat{m^{\varphi_{t}}}\right\rangle_t\bigg|_{t=0}
          \leq C\frac{\Lambda}{\rho},
        \end{align}
        Grönwall's Lemma yields the claim
        (\ref{eq:m-bound}).
        According to Condition~\ref{def:initial-conditions} we have
        $\left\langle\widehat{m^{\varphi_{t}}}\right\rangle\big|_{t=0}=0$ so
        that the bound (\ref{eq:ini-gronwall}) is
        fulfilled which
        concludes the proof of Lemma \ref{lem:m-bound}. \ed
      \end{proof}
\begin{rem}\label{rem:technique}
  (i) The proof can be extended to more general initial
  conditions than those specified in Condition~\ref{def:initial-conditions},
  namely to all wave functions, $\Psi_0$, for which the bound
  (\ref{eq:ini-gronwall}) holds.
  (ii) Note that (\ref{eq:term-II-4}) is the crucial estimate
  that determines the right-hand side of claim (\ref{eq:m-bound}). It follows from the 
  auxiliary  bound (\ref{eq:Z12p1p2-bound}), which cannot be
  improved without new insights into the dynamics of Bose gases.
  (iii) Provided $\|\varphi_t\|_\infty$ is bounded, the proof holds also for
  attractive potentials.
\end{rem}


\subsection{Proofs of Theorem~\ref{thm:rho-diff-brutal},
Theorem~\ref{thm:probM}, and Theorem~\ref{thm:improved}} \label{sec:proofs-main-results}

Lemma~\ref{lem:m-bound} immediately implies that, for a suitable class of initial wave functions,
 the microscopic and the macroscopic descriptions of the dynamics are close to one another, 
which is the content of our main results, Theorems \ref{thm:rho-diff-brutal}, \ref{thm:probM} and
Theorem~\ref{thm:improved}. 
Since we assume that the potential $U$ is repulsive,
Corollary~\ref{lem:propagation-varphi} and
Lemma~\ref{lem:propagation-epsilon} of Section~\ref{sec:propagationestimates}
below
provide the following estimates:
 There are $C_1,C_2,C_3\in\mathrm{Bounds}$ such that 
\begin{align} \Vert\varphi_t\Vert_\infty&\leq
  C_1(t),\label{eq:varphi-Linfty-bound} \\\label{eq:epsilon-L2-bound} \Vert
  \epsilon_t \Vert_2 & \leq C_2(t), \\ \label{eq:p-epsilon-L2-bound}\Vert
  p^{(\mathrm{ref})}_t \epsilon_t \Vert_2 & \leq
  \frac{C_3(t)}{\Lambda^{1/2}},
\end{align}
  for all $t\geq 0$ provided $\Lambda$ is sufficiently large.
  We temporarily assume the bounds in (\ref{eq:varphi-Linfty-bound}),
(\ref{eq:epsilon-L2-bound}) and (\ref{eq:p-epsilon-L2-bound}) and proceed to
proving our second and third main result; the first main results,
Theorem~\ref{thm:rho-diff-brutal},
will latter be proven
as a corollary.\ed
\begin{proof}[Proof of Theorem~\ref{thm:probM}]
  Because of (\ref{eq:varphi-Linfty-bound}), Lemma~\ref{lem:m-bound}
  implies that
  \begin{equation} \left|\left\langle
    \widehat{m^{\varphi_{t}}}\right\rangle _{t}\right|\leq
    C(t)\frac{\Lambda}{\rho}.
  \end{equation}
  In (\ref{eq:PtotM}) we have introduced a wave function $\tilde{\Psi}_{t}$ by setting
  \begin{align*}
    \widetilde{\Psi}_{t}:=\sum_{0\leq k\leq
    \rho}P_{k}^{\varphi_{t}}\Psi_{t}.
    \end{align*}
  Using the definition of the counting measure $m(k)$, see (\ref{eq:m-weight}), we see that
      \begin{eqnarray*} \left\Vert
        \Psi_{t}-\widetilde{\Psi}_{t}\right\Vert _{2}^{2} & = &
        \sum_{\rho<k\leq N}\left\Vert P_{k}^{\varphi_{t}}\Psi_{t}\right\Vert
        _{2}^{2}=\sum_{\rho<k\leq N}m(k)\left\Vert
        P_{k}^{\varphi_{t}}\Psi_{t}\right\Vert _{2}^{2}\\ & \leq &
        \sum_{k=0}^{N}m(k)\left\Vert P_{k}^{\varphi_{t}}\Psi_{t}\right\Vert
        _{2}^{2}=\left\langle
        \Psi_{t},\widehat{m^{\varphi_{t}}}\Psi_{t}\right\rangle.
      \end{eqnarray*} 
      By \lemref{m-bound}, there is a $C\in\mathrm{Bounds}$
      such that \[ \left\Vert \Psi_{t}-\widetilde{\Psi}_{t}\right\Vert
    _{2}\leq C(t)\sqrt{\frac{\Lambda}{\rho}}, \] 
    which concludes the proof of Theorem \ref{thm:probM}.
\end{proof}

\begin{proof}[Proof of
    Theorem~\ref{thm:improved}]
 Notice that  $\left\Vert
 P_{k}^{\varphi_{t}}\Psi_{t}\right\Vert\geq\left\Vert
 P_{k}^{\varphi_{t}}\widetilde{\Psi}_{t}\right\Vert$, for any $0\leq k\leq N$.
 This fact and definition (\ref{eq:def-w}) yield
 \begin{align}\left\langle
     \widehat{m^{\varphi_{t}}}\right\rangle_t
      =\sum_{0\leq k\leq N}m(k)\left\langle
      \Psi_{t},P_{k}^{\varphi_{t}}\Psi_{t}\right\rangle \geq \sum_{0\leq k\leq
      N}m(k)\left\langle
      \widetilde{\Psi}_{t},P_{k}^{\varphi_{t}}\widetilde{\Psi}_{t}\right\rangle
      =&\Lambda\sum_{0\leq k\leq N}\frac{k}{N}\left\langle
      \widetilde{\Psi}_{t},P_{k}^{\varphi_{t}}\widetilde{\Psi}_{t}\right\rangle
      \nonumber\\
      &-\sum_{0\leq
      k\leq N}\left(\frac{k}{\rho}-m(k)\right)\left\langle
      \widetilde{\Psi}_{t},P_{k}^{\varphi_{t}}\widetilde{\Psi}_{t}\right\rangle
      \label{eq:braket-m-zero}.
  \end{align} 
  Since $ P_{k}^{\varphi_{t}}\widetilde{\Psi}_{t}= 0$,  for $k>\rho$,
  and $\frac{k}{\rho}-m(k)=0$, for $0\leq k\leq \rho$ -- see (\ref{eq:m-weight})
  -- term (\ref{eq:braket-m-zero}) vanishes.  Using (\ref{eq:old-q1}) and the symmetry of bosonic wave functions, we get
       \[
        \Lambda\sum_{0\leq k\leq N}\frac{k}{N}\left\langle
                          \widetilde{\Psi}_{t},P_{k}^{\varphi_{t}}\widetilde{\Psi}_{t}\right\rangle
        =\Lambda\sum_{0\leq k\leq N}\frac{1}{N}\left\langle
                    \widetilde{\Psi}_{t},q_{k}^{\varphi_{t}}\widetilde{\Psi}_{t}\right\rangle=\Lambda\left\langle
        \widetilde{\Psi}_{t},q_{1}^{\varphi_{t}}\widetilde{\Psi}_{t}\right\rangle.\] \ed
        This implies that
        \begin{equation} \Lambda\left\langle
          \widetilde{\Psi}_{t},q_{1}^{\varphi_{t}}\widetilde{\Psi}_{t}\right\rangle
                \leq\left\langle \widehat{m^{\varphi_{t}}}\right\rangle_t.
          \label{eq:q-and-m-relation}
        \end{equation} Furthermore, upon
        inserting identity operators, in the form of 
        $\mathrm{id}_{{\cal H}}=p_{1}^{\varphi_{t}}+q_{1}^{\varphi_{t}}$, 
        the difference of the density matrices can be bounded by
        \begin{eqnarray} \left\Vert
          \widetilde\rho_{t}^{(\mathrm{micro})}-\rho_{t}^{(\mathrm{macro})}\right\Vert
          & \equiv & \left\Vert \Lambda q_{t}^{(\mathrm{ref})}\,
          tr{}_{x_{2},\ldots,x_{N}}\left|\widetilde{\Psi}_{t}\right\rangle
          \left\langle
          \widetilde{\Psi}_{t}\right|q_{t}^{(\mathrm{ref})}-\left|\epsilon_{t}\right\rangle
          \left\langle \epsilon_{t}\right|\right\Vert \nonumber \\ & \leq &
          \left\Vert \Lambda q_{t}^{(\mathrm{ref})}\,
          tr{}_{x_{2},\ldots,x_{N}}\left[p_{1}^{\varphi_{t}}\left|\widetilde{\Psi}_{t}\right\rangle
          \left\langle
        \widetilde{\Psi}_{t}\right|p_{1}^{\varphi_{t}}\right]q_{t}^{(\mathrm{ref})}-\left|\epsilon_{t}\right\rangle
        \left\langle \epsilon_{t}\right|\right\Vert \label{eq:density_PP}\\ &
        & +2\Lambda \left\Vert q_{t}^{(\mathrm{ref})}\,
        tr{}_{x_{2},\ldots,x_{N}}\left[p_{1}^{\varphi_{t}}\left|\widetilde{\Psi}_{t}\right\rangle
        \left\langle
      \widetilde{\Psi}_{t}\right|q_{1}^{\varphi_{t}}\right]q_{t}^{(\mathrm{ref})}\right\Vert
      \label{eq:density_PQ}\\ &  & +\Lambda \left\Vert
      q_{t}^{(\mathrm{ref})}\,
      tr{}_{x_{2},\ldots,x_{N}}\left[q_{1}^{\varphi_{t}}\left|\widetilde{\Psi}_{t}\right\rangle
      \left\langle
    \widetilde{\Psi}_{t}\right|q_{1}^{\varphi_{t}}\right]q_{t}^{(\mathrm{ref})}\right\Vert.\label{eq:density_QQ}
  \end{eqnarray} 
  In order to estimate (\ref{eq:density_PP}), we shall need the preliminary bound
  \begin{align} &
    \left\Vert
    q_{t}^{(\mathrm{ref})}\left|\varphi_{t}\right\rangle
    \left\langle
    \varphi_{t}\right|q_{t}^{(\mathrm{ref})}-\left|\epsilon_{t}\right\rangle
    \left\langle \epsilon_{t}\right|\right\Vert \nonumber \\ = & \left\Vert
    q_{t}^{(\mathrm{ref})}\left|\phi_{t}^{(\mathrm{ref})}+\epsilon_{t}\right\rangle
    \left\langle
    \phi_{t}^{(\mathrm{ref})}+\epsilon_{t}\right|q_{t}^{(\mathrm{ref})}-\left|\epsilon_{t}\right\rangle
    \left\langle \epsilon_{t}\right|\right\Vert \nonumber \\ \leq & \left\Vert
    q_{t}^{(\mathrm{ref})}\left|\epsilon_{t}\right\rangle \left\langle
    \epsilon_{t}\right|q_{t}^{(\mathrm{ref})}-\left|\epsilon_{t}\right\rangle
    \left\langle \epsilon_{t}\right|\right\Vert \nonumber \\ \leq & \left\Vert
    p_{t}^{(\mathrm{ref})}\left|\epsilon_{t}\right\rangle \left\langle
    \epsilon_{t}\right|p_{t}^{(\mathrm{ref})}\right\Vert +2\left\Vert
    p_{t}^{(\mathrm{ref})}\left|\epsilon_{t}\right\rangle \left\langle
    \epsilon_{t}\right|\right\Vert \nonumber \\ \leq &
    \frac{C(t)^{2}}{\Lambda }+\frac{C(t)}{\Lambda ^{1/2}},\label{eq:phi-eps-dist}
  \end{align} where, in the last two lines, we have used
  (\ref{eq:epsilon-L2-bound}) and (\ref{eq:p-epsilon-L2-bound}) of 
  Lemma~\ref{lem:propagation-epsilon}, (see Subsection \ref{sec:propagation-epsilon}). 
  We are now prepared to provide the estimates of 
  terms (\ref{eq:density_PP}), (\ref{eq:density_PQ}) and (\ref{eq:density_QQ}):\\

  \noun{Term (\ref{eq:density_PP})}: Fubini's Theorem justifies the identity
\[
  \left\langle
  \frac{\varphi_{t}}{\Lambda^{1/2}}\right|tr{}_{x_{2},\ldots,x_{N}}\left|\Psi_{t}\right\rangle
  \left\langle
\Psi_{t}\right|\left|\frac{\varphi_{t}}{\Lambda^{1/2}}\right\rangle
  =
  \left\langle \Psi_{t},\left|\frac{\varphi_{t}}{\Lambda^{1/2}}\right\rangle \left\langle
  \frac{\varphi_{t}}{\Lambda^{1/2}}\right|\Psi_{t}\right\rangle
  =
  1-\left\langle \Psi_{t},q_{1}^{\varphi_{t}}\Psi_{t} \right\rangle.
\]
The right side can be bounded according to \[ \left|1-\left\langle
\Psi_{t},q_{1}^{\varphi_{t}}\Psi_{t}\right\rangle
\right|\leq1+\frac{\left\langle \widehat{m^{\varphi_{t}}}\right\rangle_t
}{\Lambda }\leq 2, \] 
provided $\Lambda$ is sufficiently large.
Hence, (\ref{eq:q-and-m-relation}) and (\ref{eq:phi-eps-dist}), together with
(\ref{eq:epsilon-L2-bound}) and (\ref{eq:p-epsilon-L2-bound}) of 
Lemma~\ref{lem:propagation-epsilon}, guarantee that
\begin{eqnarray} (\ref{eq:density_PP}) & = & \left\Vert
  \Lambda q_{t}^{(\mathrm{ref})}\left|\frac{\varphi_{t}}{\Lambda^{1/2}}\right\rangle
  \left\langle
  \frac{\varphi_{t}}{\Lambda^{1/2}}\right|tr{}_{x_{2},\ldots,x_{N}}\left[\left|\Psi_{t}\right\rangle
  \left\langle \Psi_{t}\right|\right]\left|\frac{\varphi_{t}}{\Lambda^{1/2}}\right\rangle
  \left\langle
  \frac{\varphi_{t}}{\Lambda^{1/2}}\right|q_{t}^{(\mathrm{ref})}-\left|\epsilon_{t}\right\rangle
  \left\langle \epsilon_{t}\right|\right\Vert \nonumber \\ & = & \left\Vert
  \left(1-\left\langle \Psi_{t},q_{1}^{\varphi_{t}}\Psi_{t}\right\rangle
  \right)\left[q_{t}^{(\mathrm{ref})}\left|\varphi_{t}\right\rangle
  \left\langle
\varphi_{t}\right|q_{t}^{(\mathrm{ref})}-\left|\epsilon_{t}\right\rangle
\left\langle \epsilon_{t}\right|\right]+\left\langle
\Psi_{t},q_{1}^{\varphi_{t}}\Psi_{t}\right\rangle
\left|\epsilon_{t}\right\rangle \left\langle \epsilon_{t}\right|\right\Vert
\nonumber \\ & \leq & \left\Vert \left(1-\left\langle
\Psi_{t},q_{1}^{\varphi_{t}}\Psi_{t}\right\rangle
\right)\left[q_{t}^{(\mathrm{ref})}\left|\varphi_{t}\right\rangle
\left\langle
\varphi_{t}\right|q_{t}^{(\mathrm{ref})}-\left|\epsilon_{t}\right\rangle
\left\langle \epsilon_{t}\right|\right]\right\Vert +\left|\left\langle
\Psi_{t},q_{1}^{\varphi_{t}}\Psi_{t}\right\rangle \right|\left\Vert
\epsilon_{t}\right\Vert _{2}^{2}\nonumber \\ & \leq &
2\left(\frac{C(t)^{2}}{\Lambda }+\frac{C(t)}{\Lambda ^{1/2}}\right)+\frac{\left\langle
  \widehat{m^{\varphi_{t}}}\right\rangle_t
}{\Lambda }C(t)^{2}.\label{eq:density_PP_est}
\end{eqnarray} \noun{Term (\ref{eq:density_PQ}):} Thanks to
(\ref{eq:epsilon-L2-bound}) of Lemma~\ref{lem:propagation-epsilon} we have that
\begin{eqnarray} (\ref{eq:density_PQ}) & = & 2\Lambda \left\Vert
  q_{t}^{(\mathrm{ref})}\,
  tr{}_{x_{2},\ldots,x_{N}}\left[p_{1}^{\varphi_{t}}\left|\Psi_{t}\right\rangle
  \left\langle
\Psi_{t}\right|q_{1}^{\varphi_{t}}\right]q_{t}^{(\mathrm{ref})}\right\Vert
\nonumber \\ & \leq & 2\Lambda \left\Vert
q_{t}^{(\mathrm{ref})}\frac{\varphi_{t}}{\Lambda^{1/2}}\right\Vert _{2}\left\Vert
q_{1}^{\varphi_{t}}\Psi_{t}\right\Vert _{2}\nonumber \\ & = &
2\Lambda \left\Vert
q_{t}^{(\mathrm{ref})}\frac{\phi^{(\mathrm{ref})}+\epsilon_{t}}{\Lambda ^{1/2}}\right\Vert
_{2}\left\Vert q_{1}^{\varphi_{t}}\Psi_{t}\right\Vert _{2}\nonumber \\ & \leq
& 2\Lambda \left\Vert
\frac{\epsilon_{t}}{\Lambda ^{1/2}}\right\Vert
_{2}\sqrt{\frac{\left\langle \widehat{m^{\varphi_{t}}}\right\rangle_t
}{\Lambda }}\nonumber \\ & \leq & 2\sqrt{\left\langle
  \widehat{m^{\varphi_{t}}}\right\rangle_t
}C(t).\label{eq:density_PQ_est}
\end{eqnarray} \noun{Term
(\ref{eq:density_QQ}):} A straight-forward computation yields
\begin{eqnarray} (\ref{eq:density_QQ}) & = &
  \Lambda \left\Vert q_{t}^{(\mathrm{ref})}\,
  tr{}_{x_{2},\ldots,x_{N}}\left[q_{1}^{\varphi_{t}}\left|\Psi_{t}\right\rangle
  \left\langle
\Psi_{t}\right|q_{1}^{\varphi_{t}}\right]q_{t}^{(\mathrm{ref})}\right\Vert
\nonumber \\ & \leq & \Lambda \left\Vert
q_{1}^{\varphi_{t}}\Psi_{t}\right\Vert ^{2}\nonumber \\ & \leq & \left\langle
\widehat{m^{\varphi_{t}}}\right\rangle_t. \label{eq:density_QQ_est}
\end{eqnarray}

Collecting estimates (\ref{eq:density_PP_est}), (\ref{eq:density_PQ_est}) and
(\ref{eq:density_QQ_est}) we find
\begin{eqnarray*} 
    \left\Vert
  \widetilde\rho_{t}^{(\mathrm{micro})}-\rho_{t}^{(\mathrm{macro})}\right\Vert  & \leq &
  \frac{C(t)^{2}}{\Lambda }+\frac{C(t)}{\Lambda ^{1/2}}+\frac{\left\langle
    \widehat{m^{\varphi_{t}}}\right\rangle_t
  }{\Lambda }C(t)^{2}+2\sqrt{\left\langle
    \widehat{m^{\varphi_{t}}}\right\rangle_t }C(t)+\left\langle
    \widehat{m^{\varphi_{t}}}\right\rangle_t.
\end{eqnarray*}
  However, thanks to (\ref{eq:varphi-Linfty-bound}), Lemma~\ref{lem:m-bound}
  shows that
  \begin{equation} \left|\left\langle
    \widehat{m^{\varphi_{t}}}\right\rangle _{t}\right|\leq
    C(t)\frac{\Lambda}{\rho}, \qquad 0\leq t\leq T.
  \end{equation}
 As a consequence, there is a $C\in\mathrm{Bounds}$ such that
  \begin{eqnarray*}
\left\Vert
  \widetilde\rho_{t}^{(\mathrm{micro})}-\rho_{t}^{(\mathrm{macro})}\right\Vert
   & \leq &
    C(t)\sqrt{\frac{\Lambda}{\rho}}.
  \end{eqnarray*} 
\end{proof}

To conclude this section, we note that our first main result is an immediate consequence 
of Theorem~\ref{thm:probM} and Theorem~\ref{thm:improved}.

\begin{proof}[Proof of Theorem~\ref{thm:rho-diff-brutal}]
    Theorems ~\ref{thm:probM} and ~\ref{thm:improved} imply that
  \begin{align*}
    \left\Vert \rho_t^{(\mathrm{micro})} - \rho_t^{(\mathrm{macro})} \right\Vert
    & \leq 
    \left
      \Vert \rho_t^{(\mathrm{micro})}
      - \widetilde\rho_t^{(\mathrm{micro})}
    \right\Vert
    + \left\Vert
        \widetilde\rho_t^{(\mathrm{micro})} - \rho_t^{(\mathrm{macro})} 
      \right\Vert
    \\
    & \leq
        C \Lambda \left\Vert \Psi - \widetilde \Psi \right\Vert_2
        + C(t) \sqrt{\frac{\Lambda}{\rho}}
    \\
    & \leq
    C \Lambda C(t) \sqrt{\frac{\Lambda}{\rho}}
    + C(t) \sqrt{\frac{\Lambda}{\rho}}
    \leq
    C(t) \frac{\Lambda^{3/2}}{\rho^{1/2}}.
  \end{align*}
\end{proof}

%

\subsection{A Priori Propagation Estimates}\label{sec:propagationestimates}

In this section we prove the propagation estimates (\ref{eq:varphi-Linfty-bound}),
(\ref{eq:epsilon-L2-bound}) and (\ref{eq:p-epsilon-L2-bound}) --
Corollary~\ref{lem:propagation-varphi} and Lemma~\ref{lem:propagation-epsilon}
-- that have been required in the proofs of our first three main results.

To gain 
the required control of the solutions to the
non-linear equations (\ref{eq:varphi-evolution}), (\ref{eq:phi-evolution}),
and (\ref{eq:epsilon-evolution}) turns out
to be quite involved. \ed
Therefore, it is convenient, to first
study the dynamics on a torus, $\mathbb{T}$, meaning that we view the region $\Lambda$ 
as a torus and impose periodic boundary conditions.  In order to distinguish these two different situations in
our notations, we use the following convention.  On $\mathbb R^3$ we refer to
the solutions of equations (\ref{eq:microscopic}), (\ref{eq:varphi-evolution}),
(\ref{eq:phi-evolution}), and (\ref{eq:epsilon-evolution}) as before, i.e., as
$$
  t\mapsto\Psi_{t}, \quad t\mapsto\varphi_{t},
  \quad t\mapsto\phi^{(\mathrm{ref})}_t, \quad t\mapsto\epsilon_t,
$$ 
whereas, on $\mathbb T$, we write
$$
  t\mapsto\Psi_{t}^{\mathbb T},
  \quad t\mapsto\varphi_{t}^{\mathbb T}, \quad t\mapsto\phi^{\mathbb T,(\mathrm{ref})}_t,
  \quad t\mapsto\epsilon^{\mathbb T}_t.
$$
The corresponding initial conditions on the torus are
\begin{align}
  \label{eq:iv-torus}
  e^{i\Vert U\Vert_1 t}\varphi^{\mathbb T}_0
  :=
  \phi^{\mathbb T,(\mathrm{ref})}_0+\epsilon_0^{\mathbb T},
  \qquad
  \phi^{\mathbb T,(\mathrm{ref})}_0
  :=
  1, \qquad \epsilon_0^{\mathbb T}
  := \epsilon_0;
\end{align}
see Condition~\ref{def:initial-conditions}. Note that we neither distinguish the
differential operators on $\mathbb T$ and $\mathbb R^3$ in our notation, nor we make the
domain, $\Lambda$, of integration explicit in the integrals. Both can be
unambiguously inferred from context. 
Furthermore, for some $T\leq \infty$ we
assume the above solutions to
exist on the time interval $[0,T)$ and consider only times $t\in[0,T)$. \ed

One of the main goals of this section is to provide $L^\infty$ norms on
$\phi^{\mathrm{(ref)}}_t$, $\varphi_t$, and $\epsilon_t$. 
The advantage of the torus is that the respective reference state $\phi^{\mathbb T,\mathrm{(ref)}}_t$ is
simply a constant, whereas $\phi^{\mathrm{(ref)}}_t$ on $\mathbb R^3$ has tails.
In consequence, on the torus the only kinetic energy there is stems from the
excitation. It can be readily estimated by energy conservation and provides an estimate
that is good enough to prevent excessive clustering of particles. 
Heuristically, the same is true for the reference
state in $\mathbb R^3$ as it is very flat. However, there it is more difficult to distinguish the
kinetic energy due to the excitation and the one due to the tails of the
reference state in the technical estimates.
Therefore, we first study $\phi^{\mathbb T,\mathrm{(ref)}}_t$, $\varphi^{\mathbb
T}_t$, and $\epsilon^{\mathbb T}_t$ on the torus in
Section~\ref{sec:propagation-on-torus}. Afterwards we construct 
auxiliary wave functions on $\mathbb R^3$ by means of the torus wave functions
which are
already in some sense close $\phi^{\mathrm{(ref)}}_t$, $\varphi_t$, and
$\epsilon_t$, 
respectively.
The propagation of errors is then controlled by Grönwall arguments which allow
to extend the results in the case of the torus to the one of $\mathbb R^3$; see
Sections~\ref{sec:propagation-phi} and \ref{sec:propagation-varphi}. The latter
sections also
provide the required control of the excitations which is discussed in Section~\ref{sec:propagation-epsilon}. 

While the quantum mechanical
spreading due to the Laplace term usually tends to relax bad situations,
the pair-interaction due to $U$ could give rise to such, and a strategy
is needed to control the $L^\infty$ norms of solutions over time.
Here it is important to recall that the respective $L^2$ norms 
$\phi^{\mathrm{(ref)}}_t$ and $\varphi_t$ scale proportionally to $\Lambda^{1/2}$.
Hence, over time the growth of the solutions due to the interaction can not simply be
controlled by using an $L^2$ estimate in a Cook's argument.
For this reason we introduce the following Lemma \ref{lem:L1-norm-evolution} 
which will be applied frequently below.
It holds on $\mathbb R^3$ as well as on the torus
$\mathbb T$ and makes use of the following convenient norms: \ed
\begin{defn}
 For $0\leq
  p_1,p_2,p_3,\ldots \leq \infty$ we define the norms
  \begin{align*}
    \Vert\zeta\Vert_{p_1\wedge p_2\wedge p_3\wedge\ldots} & :=
    \inf_{\zeta=\zeta_{p_1}+\zeta_{p_2}+\zeta_{p_3}+\ldots} \left(
    \Vert\zeta_{p_1}\Vert_{p_1}+\Vert\zeta_{p_2}\Vert_{p_2}
    +\Vert\zeta_{p_3}\Vert_{p_3}+\ldots \right). 
  \end{align*}
In order to compress the notation we also use
\[
    \Vert\zeta\Vert_{p_1, p_2,
    p_3, \ldots} := \Vert\zeta\Vert_{p_1}+\Vert\zeta\Vert_{p_2} + \ldots\;.
\]
\end{defn}
\begin{lem}\label{lem:L1-norm-evolution}
  Let $U\in\mathcal C^\infty_c(\mathbb R^3,\mathbb R)$ be a general
  potential. \ed
  Let $\zeta_t$ be solution of the nonlinear equation 
  $$
    i\partial_t\zeta_t(x) 
    =
    \left( -\frac{1}{2}\Delta + U*|\zeta_t|^2(x) \right)\zeta_{t}(x).
  $$
  for an initial value
  $\zeta_t|_{t=0}=\zeta_0$ such that:
  \begin{align}\label{eq:initial-constraint}
    \db \left(\left\Vert{\zeta_0}\right\Vert_\infty \leq \right) \ed
    \left\Vert\widehat{\zeta_0}\right\Vert_1 \leq C_1
    \text{ and }
    \db \left(\left\Vert\zeta_t\right\Vert_{2\wedge
    \infty}\leq\right)\ed
    \left\Vert\widehat{\,|\zeta_t|\,}\right\Vert_{1\wedge 2} \leq C_2(t)
  \end{align}
  for some $C_1,C_2\in\mathrm{Bounds}$. Then there exists a
  $C_3\in\mathrm{Bounds}$ such that
  $$
    \db \left(\left\Vert{\zeta_t}\right\Vert_\infty \leq \right) \ed
    \left\Vert\widehat{\zeta_t}\right\Vert_1
    \leq
    C_3(t).
  $$
\end{lem}
\begin{proof}
  Grönwall's Lemma, the bound on the time derivative
  \begin{align*}
    \partial_t \left\Vert \widehat\zeta_t \right\Vert_1
    & \leq
    \int dk \,
      \frac{\Im \widehat\zeta^*_t(k)
        \left( \frac{k^2}{2}\widehat\zeta_t(k) +
             \reallywidehat{U*|\zeta_t|^2\zeta_t}(k)
        \right)}
      {\left| \widehat\zeta_t(k) \right|}
    \\
    & \leq
    \int dk \int dl \int dp \,
    \left|
      \widehat U(l) \widehat\zeta_t(l-p)
      \widehat\zeta_t(p)\widehat\zeta_t(k-l)
    \right|
    \\
    & \leq
    \int dl \int dp\,
    \left|
      \widehat U(l)
      \widehat\zeta_t(l-p) \widehat\zeta_t(p)
    \right|
    \, \left\Vert 
         \widehat\zeta_t
       \right\Vert_1
    \\
    & \leq
    C \Vert U \Vert_{1,2,\infty} \,
    \left\Vert \widehat{\,|\zeta_t|\,} \right\Vert_{1\wedge 2}^2 \,
    \left\Vert \widehat\zeta_t \right\Vert_1
    \\
    & \leq
    C C_1 C_2(t)^2
    \left\Vert \widehat\zeta_t \right\Vert_1
    =:C_3(t)
    \left\Vert \widehat\zeta_t \right\Vert_1,
  \end{align*}
  and the assumption on the initial condition (\ref{eq:initial-constraint}) imply the claim.
\end{proof}
The lemma states that an a priori bound in the $\Vert\cdot\Vert_{2\wedge
\infty}$ norm is sufficient to maintain control over the $L^\infty$ norm over time. The strategy will
therefore be to establish such a priori norms in the cases of 
$\phi^{\mathrm{(ref)}}_t$, $\varphi_t$, and $\epsilon_t$ and then apply the above
lemma.

\subsubsection{Estimates on the Torus}\label{sec:propagation-on-torus}

\db As discussed this section provides the needed properties of the evolution equations on the torus $\mathbb T$
for initial values (\ref{eq:iv-torus}) and repulsive potentials $U$, i.e.,\ed
\begin{align}
    U\geq 0.
    \label{eq:U-repulsive}
\end{align}
On $\mathbb T$ the unique solution to the evolution equation (\ref{eq:phi-evolution}) of the
reference state that corresponds to initial value (\ref{eq:iv-torus})
is given by the constant, i.e.,
\begin{equation}
  \phi^{\mathbb T,(\mathrm{ref})}_t=1
  \qquad \text{for all } t\in\mathbb R.
  \label{eq:phi-solution-torus}
\end{equation}
In consequence, Condition~\ref{def:initial-conditions} and (\ref{eq:iv-torus}) imply
\begin{equation}
  \left\Vert 
    \, \widehat{\, |\varphi^{\mathbb T}_0| \,} \,
  \right\Vert_1 
  \leq C,  
  \label{eq:varphi-flatness}
\end{equation}
and because of (\ref{eq:U-repulsive}), we have
\begin{equation} E_{\varphi_0^{\mathbb T}} =
  \left\langle \varphi^{\mathbb T}_0, h_x[\varphi^{\mathbb T}_0] \varphi^{\mathbb T}_0 \right\rangle 
  \geq 0.
  \label{eq:varphi-energy}
\end{equation}
The evolution of the excitation wave function on the torus $\mathbb T$ is,
analogously as in the case of $\mathbb R^3$,
defined by
\begin{equation}
  \epsilon^{\mathbb T}_t 
  = 
  \varphi^{\mathbb T}_t e^{i\Vert U\Vert_1 t} - \phi^{\mathbb
    T,(\mathrm{ref})}_t
    \db =
\varphi^{\mathbb T}_t e^{i\Vert U\Vert_1 t} - 1.\ed
  \label{eq:epsilon-evolution-torus}
\end{equation}
Together with (\ref{eq:phi-solution-torus}) and (\ref{eq:epsilon-evolution})
this implies
\begin{align}
  \label{eq:epsilon-torus}
  i\partial_{t}\epsilon_{t}^{\mathbb T}(x) 
  = &
  \left( -\frac{1}{2}\Delta+
    U*|\epsilon_{t}^{\mathbb T}|^{2}(x) + U*2\Re{\epsilon^{\mathbb T}_{t}}^{*}(x)
  \right) \epsilon_{t}(x)
  \\
  & + U*\left( |\epsilon^{\mathbb T}_{t}|^{2}(x) +
  2\Re{\epsilon^{\mathbb T}_{t}}^{*}(x) \right). 
  \nonumber
\end{align}

\begin{lem}\label{lem:torus}
  \db Let $U\in\mathcal C^\infty_c(\mathbb R^3,\mathbb R^+_0)$ be a repulsive
  potential.
  There are \ed $C_1,C_2,C_4\in\mathrm{Bounds}$ such that for all $1/4\leq r<1$
  \begin{align}
    \label{eq:kinetic-estimate}
    \Vert\nabla\varphi^{\mathbb T}_t\Vert_2 
    =
    \Vert\nabla\epsilon^{\mathbb T}_t\Vert_2 
    &\leq C_1,
    \\
    \label{eq:varphi-L2infty-torus} 
    \Vert\varphi^{\mathbb T}_t\Vert_{2\wedge\infty} 
    \leq
    \left\Vert \widehat{\,|\varphi^{\mathbb T}_t|\,} \right\Vert_{1\wedge 2} 
    & \leq
    C_2(t), 
    \\
    \label{eq:epsilon-cutoff}
    \Vert \chi_r\epsilon^{\mathbb T}_t\Vert_2
    & \leq
    \Lambda^{-\frac{1}{3}} C_3(t).
  \end{align}
\end{lem}

\begin{proof} To see (\ref{eq:kinetic-estimate}) we begin by noting that the
  evolution equation (\ref{eq:varphi-evolution}) conserves the energy so that
  due to (\ref{eq:phi-solution-torus}), (\ref{eq:varphi-energy}), and $U\geq 0$ one finds \db
  \begin{equation*}
      \Vert \nabla  \varphi^{\mathbb T}_t \Vert_2^2 = \Vert \nabla
    \epsilon^{\mathbb T}_t \Vert_2^2 = E_{\varphi_0^{\mathbb T}} - \left\langle \varphi_t^{\mathbb T},
    U*|\varphi^{\mathbb T}_t|^2 \varphi_t^{\mathbb T}\right\rangle 
    \leq E_{\varphi_0^{\mathbb T}}.
  \end{equation*}
  Hence, the claim (\ref{eq:kinetic-estimate}) holds for the choice of constant
  $C_1^2=E_{\varphi_0^{\mathbb T}}$.\ed \\

  In order to provide the estimate (\ref{eq:varphi-L2infty-torus}) we exploit
  that the Schrödinger dispersion effectively acts only on that part of the
  wave function which is not constant. It is therefore convenient to split
  $\varphi^{\mathbb T}_t$ into two parts. For this purpose we introduce
  \db the auxiliary wave function $\widetilde{\varphi}^{\mathbb T}_t$ by
  \begin{equation} \label{eq:tildevarphi-T}
    \widetilde\varphi^{\mathbb T}_t =
    \exp\left( -i\int_0^{t}ds\, U*|\varphi^{\mathbb T}_s|^2 \right) \varphi^{\mathbb T}_0
  \end{equation}
  \ed
  so that
  \begin{equation}
    |\widetilde\varphi^{\mathbb T}_t|=|\varphi^{\mathbb T}_0|. \label{eq:tilde-varphi-evolution}
  \end{equation} 
  Next, we split the desired norm of $\varphi^{\mathbb T}_t$ as follows \db
  \begin{align}
    \Vert \varphi^{\mathbb T}_t \Vert_{2\wedge\infty} 
    & =
    \inf_{\varphi^{\mathbb T}_t=\varphi^{\mathbb T}_{t,\infty}
    + \varphi^{\mathbb T}_{t,2}} 
    \left( 
       \Vert \varphi^{\mathbb T}_{t,2} \Vert_2 + \Vert \varphi^{\mathbb T}_{t,\infty} \Vert_\infty
    \right)
    \leq
    \inf_{\varphi^{\mathbb T}_t=\varphi^{\mathbb T}_{t,\infty}
    + \varphi^{\mathbb T}_{t,2}} 
    \left( 
       \Vert \varphi^{\mathbb T}_{t,2} \Vert_2 + \left\Vert \widehat{|\varphi^{\mathbb
       T}_{t,\infty}|} \right\Vert_1
    \right)
    =
    \left\Vert \widehat{|\varphi^{\mathbb T}_t|} \right\Vert_{1\wedge 2}
    \label{eq:2inf-12}
\end{align}
for which we find
\begin{align}
    \left\Vert \widehat{|\varphi^{\mathbb T}_t|} \right\Vert_{1\wedge 2}
   & \leq
        \Vert \varphi^{\mathbb T}_t-\widetilde\varphi_t^{\mathbb T} \Vert_2
        +
        \left\Vert \widehat{|\widetilde \varphi^{\mathbb T}_t|} \right\Vert_1
    \\
    & =
        \Vert \varphi^{\mathbb T}_t-\widetilde\varphi_t^{\mathbb T} \Vert_2
        +
        \left\Vert \widehat{|\varphi^{\mathbb T}_0|} \right\Vert_1
    \\
    & \leq 
    \Vert \varphi_t^{\mathbb T} 
    - \widetilde\varphi_t^{\mathbb T} \Vert_2 + C,
    \label{eq:norm-splitting}
\end{align} \ed
where we used (\ref{eq:tilde-varphi-evolution}) and
(\ref{eq:varphi-flatness}).  It is left to control the difference of
$\varphi^{\mathbb T}_t$ and $\widetilde\varphi^{\mathbb T}_t$ in the $L^2$ norm.  \db Thanks to the conservation of the $L^2$ norms of
$\varphi^{\mathbb T}_t$ and $\widetilde \varphi^{\mathbb T}_t$, the
evolution equation (\ref{eq:varphi-evolution}), (\ref{eq:tildevarphi-T}), and (\ref{eq:kinetic-estimate}) we
find \ed
\begin{align} 
  \partial_t
  \left\Vert\varphi^{\mathbb T}_t-\widetilde \varphi^{\mathbb
  T}_t\right\Vert_2^2 & \leq 2|\partial_t \langle 
  \varphi^{\mathbb T}_t,\widetilde\varphi^{\mathbb T}_t\rangle| = | \langle
  \varphi^{\mathbb T}_t,\Delta\widetilde\varphi^{\mathbb T}_t\rangle| \leq
  \Vert\nabla\varphi^{\mathbb T}_t\Vert_2 \Vert\nabla\widetilde\varphi^{\mathbb T}_t\Vert_2
  \leq \db C_1 \Vert\nabla\widetilde\varphi^{\mathbb T}_t\Vert_2 \;.
  \label{eq:growall-L2-diff}
\end{align} 
\db Using (\ref{eq:tildevarphi-T}), the kinetic energy of
$\widetilde\varphi^{\mathbb T}_t$ can be estimated by
\begin{align}
  \Vert\nabla\widetilde \varphi^{\mathbb T}_t\Vert_2 
  & \leq
  \Vert\nabla\varphi^{\mathbb T}_0\Vert_2 
  + \db \int_0^t ds \, \left\Vert U*\left(2\Re {\varphi^{\mathbb
  T}_s}^*\nabla\varphi^{\mathbb T}_s \right) \varphi^{\mathbb T}_0 \right\Vert_2 
  \\ 
  & \leq
  \Vert\nabla\varphi^{\mathbb T}_0\Vert_2 + 
  \db 2\Vert U\Vert_{1,2}\int_0^t ds \, \Vert\varphi^{\mathbb T}_s \,
  \Vert_{2\wedge\infty}
  \Vert\nabla\varphi^{\mathbb T}_s\Vert_{2} \, \Vert\varphi^{\mathbb T}_0\Vert_\infty 
  \\ 
  & \leq C(t)\left( 1+ \int_0^t ds \, \Vert\varphi^{\mathbb T}_s \,
  \Vert_{2\wedge\infty}\right),\label{eq:1inf-int-est}
\end{align} 
where we used (\ref{eq:kinetic-estimate}).
Thus, collecting the estimates (\ref{eq:growall-L2-diff}) and
(\ref{eq:1inf-int-est}) yields
\begin{align*}
  \partial_t
  \left\Vert\varphi^{\mathbb T}_t-\widetilde \varphi^{\mathbb
  T}_t\right\Vert_2^2 \leq C(t)\left( 1+
  \int_0^t ds\, \Vert \varphi_s^{\mathbb T} - \widetilde\varphi_s^{\mathbb T}
  \Vert_2^2 \right) 
\end{align*}
where we have used the inequality $x\leq 1+x^2$, $\forall\, x\in\mathbb R$, 
to get a quadratic exponent under the integral.
Grönwall's Lemma then ensures the existence of a $C\in\mathrm{Bounds}$ such that
\begin{equation}
  \Vert\varphi^{\mathbb T}_t-\widetilde \varphi^{\mathbb T}_t\Vert_2^2 \leq
  C(t),
  \label{eq:varphi-tildevarphi-diff-torus}
\end{equation}
which together with (\ref{eq:norm-splitting}) and
$\widetilde\varphi_0^{\mathbb T}=\varphi_0^{\mathbb T}$ implies the claim
(\ref{eq:varphi-L2infty-torus}). \ed \\

  We now prove the remaining claim (\ref{eq:epsilon-cutoff}).  First, we note
  that according to (\ref{eq:epsilon-torus})
  \begin{align}
    \partial_t\Vert \chi_r\epsilon^{\mathbb T}_t\Vert^2_2 \leq & \left|
    \left\langle \epsilon^{\mathbb T}_t,\left[
    \frac{\Delta}{2},\chi_r^2\right ] \epsilon^{\mathbb T}_t
    \right\rangle \right| \label{eq:laplace-term} \\
    & + 2\left| \left\langle
    U*\left(|\epsilon^{\mathbb T}_t|^2 + \Re {\epsilon^{\mathbb T}_t}^*\right)
   ,\chi_r^2 \epsilon_t^{\mathbb T}  
    \right\rangle
    \right|.\label{eq:convolution-term}
  \end{align}
  Using partial integration, (\ref{eq:cutoff}), and
  (\ref{eq:kinetic-estimate}) we find
  \begin{align} (\ref{eq:laplace-term}) 
    & =
    \left|\left\langle\epsilon^{\mathbb T}_t, \chi_r\nabla
    \chi_r\nabla\epsilon^{\mathbb T}_t\right\rangle\right| \nonumber 
    \\ 
    & \leq 
    \Vert \chi_r\epsilon^{\mathbb T}_t\Vert_2 \, \Vert \nabla
    \chi_r\Vert_\infty \, \Vert \nabla\epsilon^{\mathbb T}_t\Vert_2
    \nonumber 
    \\ 
    & \leq 
    \Vert \chi_r\epsilon^{\mathbb T}_t\Vert_2 
    C\Lambda^{-\frac{1}{3}}
    C_1.  
    \label{eq:kin-est}
  \end{align} Next,
  equation (\ref{eq:epsilon-evolution-torus}) together with
  (\ref{eq:phi-solution-torus}) imply 
  \[ \left|\, |\epsilon^{\mathbb T}_t|^2 +
      2\Re {\epsilon^{\mathbb T}_t}^* \,\right| \leq |\epsilon^{\mathbb
      T}_t|\left( 1+|\varphi^{\mathbb T}_t| \right), \] 
  which yields the estimate
  \begin{align} (\ref{eq:convolution-term}) & \leq
  2\left\Vert \chi_r U*\left[ |\epsilon^{\mathbb T}_t|(1 + |\varphi^{\mathbb T}_t|) \right]
  \right\Vert_2 \Vert \chi_r \epsilon^{\mathbb T}_t \Vert_2 \\ & \leq \phantom{+} 2
  \left[ \int dx \left| \int dy\, U(x-y) \left( |\epsilon^{\mathbb T}_t(y)|(1 +
    |\varphi^{\mathbb T}_t(y)|) \right) \chi_r(y) \right|^2 \right]^{1/2} \Vert
    \chi_r \epsilon^{\mathbb T}_t \Vert_2 \label{eq:conv-1} \\ & \phantom{\leq} + 2
    \left[ \int dx \left| \int dy\, U(x-y) \left( |\epsilon^{\mathbb T}_t(y)|(1 +
      |\varphi^{\mathbb T}_t(y)|) \right) \left( \chi_r(x)-\chi_r(y) \right)
    \right|^2 \right]^{1/2} \Vert \chi_r \epsilon^{\mathbb T}_t \Vert_2.
    \label{eq:conv-2}
  \end{align}
  Furthermore,
  \begin{align} 
    (\ref{eq:conv-1})
    \leq 
    C\Vert U \Vert_{1,2} \; \Vert\, 1+|\varphi^{\mathbb T}_t|
    \,\Vert_{2\wedge\infty} \; \Vert \chi_r \epsilon^{\mathbb T}_t \Vert_2^2,
    \label{eq:conv-1-1}
  \end{align}
  \begin{align} (\ref{eq:conv-2}) \leq
    C\Lambda^{-\frac{1}{3}}D\Vert U \Vert_{1,2} \, \Vert\, 1+|\varphi^{\mathbb T}_t|
    \,\Vert_{2\wedge\infty} \, \Vert \epsilon^{\mathbb T}_t \Vert_2 \Vert \chi_r
    \epsilon^{\mathbb T}_t \Vert_2,
    \label{eq:conv-1-2}
  \end{align} where we have used
  that $U$ is supported in a ball of radius $D\geq0$ around the origin so that
  by (\ref{eq:cutoff})
  \begin{align} \db |U(x-y)(\chi_r(x)-\chi_r(y))|\leq C\Lambda^{-\frac{1}{3}}|U(x-y)|D.
    \label{eq:Um-decay}
  \end{align} Now equation
  (\ref{eq:tilde-varphi-evolution}) and the bound in
  (\ref{eq:varphi-tildevarphi-diff-torus}) ensure
  \begin{align} \Vert\,
    1+|\varphi^{\mathbb T}_t| \,\Vert_{2\wedge\infty} & \leq \Vert\, 1+
    |\widetilde\varphi_t^{\mathbb T}| +|\varphi^{\mathbb
    T}_t-\widetilde\varphi^{\mathbb T}_t| \,\Vert_{2\wedge\infty} \leq \Vert
    1+|\widetilde\varphi^{\mathbb T}_t| \, \Vert_\infty +\Vert \varphi^{\mathbb T}_t
    - \widetilde\varphi_t^{\mathbb T} \Vert_2 \nonumber \\ & \leq 1+\Vert \varphi^{\mathbb
    T}_0 \Vert_\infty + C(t) \leq C(t).
    \label{eq:1plusPhi}
  \end{align}
  Finally, a similar computation as the one used in (\ref{eq:laplace-term})
  gives
  \begin{align*} \partial_t\Vert \epsilon^{\mathbb T}_t\Vert^2_2 & \leq 2\left|
      \left\langle U*\left(|\epsilon^{\mathbb T}_t|^2 + \Re {\epsilon^{\mathbb
      T}_t}^*\right), \epsilon_t^{\mathbb T} \right\rangle \right| \\ & \leq
      \db 2\Vert U \Vert_{1,2} \, \left\Vert \, 1 + |\varphi^{\mathbb T}_t|\,
      \right\Vert_{2\wedge\infty} \, \Vert \epsilon^{\mathbb T}_t \Vert_2^2
  \end{align*} which thanks to (\ref{eq:1plusPhi}) and Grönwall's Lemma means
  \begin{align}\label{eq:epsilon-L2-bound-torus} \Vert \epsilon_t^{\mathbb T} \Vert_2
    \leq C(t).
  \end{align} 
  Hence, (\ref{eq:conv-1-1}) and (\ref{eq:conv-1-2}) imply
  \begin{align*} (\ref{eq:convolution-term}) \leq C(t)\Lambda^{-\frac{1}{3}}\Vert
    \chi_r \epsilon^{\mathbb T}_t \Vert_2 +C(t)\Vert \chi_r
    \epsilon^{\mathbb T}_t \Vert_2^2.
  \end{align*} 
  Finally, (\ref{eq:laplace-term}), which was estimated in (\ref{eq:kin-est}), and (\ref{eq:convolution-term})
  guarantee
  \[ \Vert \chi_r \epsilon^{\mathbb T}_t \Vert_2^2 \leq
  C(t)\left(\Lambda^{-\frac{1}{3}}+\Vert \chi_r \epsilon^{\mathbb T}_0
  \Vert_2\right).  \] 
  Note that by initial
  constraint (\ref{eq:conditions-epsilon}) one has $\chi_r \epsilon^{\mathbb
  T}_0=0$ for $r\geq 1/4$.
  In conclusion,
  the claim (\ref{eq:epsilon-cutoff}) is a consequence of Grönwall's
  Lemma.
\end{proof}

Lemma~\ref{lem:L1-norm-evolution} and Lemma~\ref{lem:torus} imply the following
corollary.
\begin{cor}\label{cor:epsilon-varphi-torus} 
  \db Let $U\in\mathcal C^\infty_c(\mathbb R^3,\mathbb R^+_0)$ be a repulsive
  potential. \ed
    There is a
  $C\in\mathrm{Bounds}$ such that 
  \[ \Vert \epsilon_t^{\mathbb T} \Vert_\infty \leq 1 +
    \Vert \varphi_t^{\mathbb T} \Vert_\infty \leq 1 + \left\Vert \widehat{\varphi^{\mathbb T}_t}
  \right\Vert_1 \leq C(t).  
  \]
\end{cor}

\subsubsection{Estimates for
  $\phi^{(\mathrm{ref})}_{t}$}\label{sec:propagation-phi}

\begin{lem}\label{lem:phi-propagation} 
\db Let $U\in\mathcal C^\infty_c(\mathbb R^3,\mathbb R)$ be a general 
  potential, and let $\Lambda$ be sufficiently large. \ed
    There are
  $C_1,C_2\in\mathrm{Bounds}$ such that
  \begin{align}
    \label{eq:phi-L2infty-bound}
    \Vert\phi_{t}^{(\mathrm{ref})}\Vert_{2\wedge\infty} \leq \left\Vert \,
    \reallywidehat{\,|\phi_{t}^{(\mathrm{ref})}|} \, \right\Vert_{1\wedge 2} & \leq C_1(t).
    \\ \label{eq:moduls-phi-phi0-L2-bound}
    \left\Vert |\phi_{t}^{(\mathrm{ref})}| -
      |\phi_{0}^{(\mathrm{ref})}| 
    \right\Vert_{2} & \leq C_2(t) \Lambda^{-\frac{1}{6}}.
  \end{align}
\end{lem}

\begin{proof} In order to provide the bound (\ref{eq:phi-L2infty-bound}) we
  introduce the auxiliary wave function
  \begin{align} \widetilde\phi_{t} :=
    \exp\left(-it U*\left(|\phi_{0}^{(\mathrm{ref})}|^2-1\right)\right)
    \phi_{0}^{(\mathrm{ref})}, \label{eq:tilde-phi-def}
  \end{align} and using
  the evolution equation (\ref{eq:phi-evolution}) we estimate the time
  derivative
  \begin{align} 
    \partial_t \left\Vert \phi^{(\mathrm{ref})}_t -
    \widetilde\phi_t \right\Vert_2 
    & \leq 
    \left\Vert
    -\frac{1}{2}\Delta\widetilde\phi_t + U*\left(
    |\phi^{(\mathrm{ref})}_t|^2-|\phi_{0}^{(\mathrm{ref})}|^2 \right)\widetilde\phi_t
    \right\Vert_2 \nonumber 
    \\ 
    & \leq 
    \left\Vert
    \frac{1}{2}\Delta\widetilde\phi_t \right \Vert_2 + \Vert U \Vert_{1,2}
    \, \left\Vert |\phi^{(\mathrm{ref})}_t|^2-|\phi_{0}^{(\mathrm{ref})}|^2
    \right\Vert_{1\wedge2} \, 
    \left\Vert\widetilde\phi_t \right\Vert_\infty.
    \label{eq:phi-tildephi-dervative}
  \end{align} 
  We estimate the terms on the
  right-hand side of (\ref{eq:phi-tildephi-dervative}) individually:\\

  Noting that \db
  \begin{align*} 
        \nabla  \widetilde\phi_t 
    = & 
        \left(
            \left[ -it U*\nabla| \phi^{(\mathrm{ref})}_0|^2 \right]
            \phi^{(\mathrm{ref})}_0 
            +
            \nabla \phi^{(\mathrm{ref})}_0
        \right)
        \exp\left(-it U*\left(|\phi_{0}^{(\mathrm{ref})}|^2-1\right)\right)
        ,
    \\ 
    \Delta \widetilde\phi_t 
    =  &
        \bigg(
        \left[ -it U*\Delta|\phi^{(\mathrm{ref})}_0|^2 \right]  
        \phi^{(\mathrm{ref})}_0 
    +
        \left[ -it U*\nabla|\phi^{(\mathrm{ref})}_0|^2 \right]^2 
        \phi^{(\mathrm{ref})}_0 
    \\ 
    & +
        2\left[ -it U*\nabla|\phi^{(\mathrm{ref})}_0|^2 \right] 
        \nabla \phi^{(\mathrm{ref})}_0
    + 
        \Delta\phi^{(\mathrm{ref})}_0
    \bigg)
    \exp\left(-itU*\left(|\phi_{0}^{(\mathrm{ref})}|^2-1\right)\right),
  \end{align*} 
  \ed and recalling
  (\ref{eq:conditions-phi}) and (\ref{eq:nabla-phi-2}), we find
  \begin{align} 
    \Vert \nabla  \widetilde
    \phi_t \Vert_\infty 
    & \leq 
    \left( 1+2|t|\,\Vert U \Vert_1 \, \Vert
    \phi_{0}^{(\mathrm{ref})} \Vert_\infty^2 \right) \Vert \nabla
    \phi_{0}^{(\mathrm{ref})} \Vert_\infty \leq C(t)\Lambda^{-\frac{1}{3}},
    \label{eq:nabla-tilde-phi-infty} 
    \\ 
    \Vert \nabla  \widetilde \phi_t \Vert_2 
    & \leq 
    \left( 1+2|t|\,\Vert U \Vert_1 \, \Vert \phi_{0}^{(\mathrm{ref})}
    \Vert_\infty^2 \right) \Vert \nabla  \phi_{0}^{(\mathrm{ref})} \Vert_2
    \leq 
    C(t) \Lambda^{\frac{1}{6}},  
    \label{eq:nabla-tilde-phi-2} 
    \\ 
    \Vert \Delta
    \widetilde \phi_t \Vert_2 
    & \leq 
    2|t| \, \Vert U \Vert_1 \, \Vert \phi^{(\mathrm{ref})}_{0} \Vert_\infty
    \left(
      \Vert \phi^{(\mathrm{ref})}_{0} \Vert_\infty \, \Vert \Delta \phi^{(\mathrm{ref})}_{0} \Vert_2
      +
      \Vert \nabla \phi^{(\mathrm{ref})}_{0} \Vert_\infty \, \Vert \nabla \phi^{(\mathrm{ref})}_{0} \Vert_2  
    \right)
    \nonumber
    \\
    & \db \quad +
    4t^2 \, \Vert U \Vert_1^2 \, \Vert \phi^{(\mathrm{ref})}_{0} \Vert_\infty^3 \,
    \Vert \nabla \phi^{(\mathrm{ref})}_{0} \Vert_\infty \,
    \Vert \nabla \phi^{(\mathrm{ref})}_{0} \Vert_2
    \nonumber\ed
    \\
    & \db \quad +
    4|t| \, \Vert U \Vert_1  \, \Vert \phi^{(\mathrm{ref})}_{0} \Vert_\infty \, 
    \Vert \nabla \phi^{(\mathrm{ref})}_{0} \Vert_\infty \,
    \Vert \nabla \phi^{(\mathrm{ref})}_{0} \Vert_2
    \nonumber\ed
    \\
    & \quad +
    \Vert \Delta \phi^{(\mathrm{ref})}_{0} \Vert_2
    \nonumber
    \\
    & \leq 
    C(t) \Lambda^{-\frac{1}{6}}.
    \label{eq:delta-tilde-phi}
  \end{align}

  These estimates together with \db (\ref{eq:conditions-phi}),
  $|\widetilde \phi_t|=|\phi^{(\mathrm{ref})}_0|$,
  and \ed
  (\ref{eq:phi-tildephi-dervative}) ensure\db
  \begin{align*} 
    \partial_t \left\Vert \phi^{(\mathrm{ref})}_t -
    \widetilde\phi_t \right\Vert_2 
    & \leq
     C(t)\Lambda^{-\frac{1}{6}} 
    + C \left\Vert
          |\phi^{(\mathrm{ref})}_t-\widetilde\phi_t|^2
          +
          2\Re \widetilde\phi_t^*
          \left(\phi^{(\mathrm{ref})}_t-\widetilde\phi_t\right)
        \right\Vert_{1\wedge2}
    \\
    & \leq
    C(t)\Lambda^{-\frac{1}{6}} 
    + C \left(\left\Vert \phi^{(\mathrm{ref})}_t-\widetilde\phi_t
    \right\Vert_2^2
    +2\Vert\phi^{(\mathrm{ref})}_0\Vert_\infty \, \left\Vert
    \phi^{(\mathrm{ref})}_t-\widetilde\phi_t \right\Vert_2
    \right). 
  \end{align*}
  Assume that there is a $0\leq \overline t\leq \infty$ such $\left\Vert
    \phi^{(\mathrm{ref})}_t-\phi^{(\mathrm{ref})}_0 \right\Vert_2\leq 1$ for all
    $t\in[0,\overline t]$. In this case we find
  \begin{align*}
    \partial_t \left\Vert \phi^{(\mathrm{ref})}_t -
    \widetilde\phi_t \right\Vert_2& \leq
    C(t)\Lambda^{-1/6} + C\left\Vert
    \phi^{(\mathrm{ref})}_t-\phi^{(\mathrm{ref})}_0 \right\Vert_2,
  \end{align*} 
  which thanks to Grönwall's Lemma and
  $\phi^{(\mathrm{ref})}_{0}=\widetilde\phi_0$ implies
  \begin{align}
    \left\Vert \phi^{(\mathrm{ref})}_t - \widetilde\phi_t \right\Vert_2 \leq
    C(t)\Lambda^{-\frac{1}{6}}
    \qquad
    \text{for }t\in[0,\overline t].\label{eq:phi-tildephi-L2}
  \end{align}
  Clearly, upon choosing $\Lambda$ sufficiently
  large the supremum of such times $\overline t$ in infinite. Hence,
  (\ref{eq:phi-tildephi-L2}) holds for all $t\in\mathbb R$ provided $\Lambda$ is
  sufficient large. 
  \ed
  In conclusion, due to (\ref{eq:conditions-phi}) we observe
  \begin{align*} \left\Vert
    \phi_{t}^{(\mathrm{ref})} \right\Vert_{2\wedge\infty} 
    & \leq 
    \left\Vert \,
    \reallywidehat{\, |\phi_{t}^{(\mathrm{ref})}|} \, \right\Vert_{1,2} 
    \leq \left\Vert
    \,\reallywidehat{\,|\widetilde\phi_t|}\, \right\Vert_1 + \left\Vert
    \phi_{t}^{(\mathrm{ref})} - \widetilde\phi_t \right\Vert_2 \leq \left\Vert
    \,\reallywidehat{\,|\phi^{(\mathrm{ref})}_0|}\, \right\Vert_1 + \left\Vert
    \phi_{t}^{(\mathrm{ref})} - \widetilde\phi_t \right\Vert_2 \\ & \leq
    C+C(t)\Lambda^{-\frac{1}{6}},
  \end{align*}
  which implies that the claim
  (\ref{eq:phi-L2infty-bound}) is true.
  \\

  Moreover, claim (\ref{eq:moduls-phi-phi0-L2-bound}) can be seen by
  (\ref{eq:phi-tildephi-L2}) and
  \begin{align*} \left\Vert
    |\phi_{t}^{(\mathrm{ref})}| - |\phi_{0}^{(\mathrm{ref})}| \right\Vert_2 &
    \leq \left\Vert \phi_{t}^{(\mathrm{ref})}-\widetilde\phi_t \right\Vert_2.
  \end{align*} 
\end{proof}

Lemma~\ref{lem:L1-norm-evolution} and Lemma~\ref{lem:phi-propagation} imply the following
corollary.
\begin{cor}\label{lem:phi-propagationa} 
    \db Let $U\in\mathcal C^\infty_c(\mathbb R^3,\mathbb R)$ be a general 
  potential, and let $\Lambda$ be sufficiently large. \ed
    There is a
  $C\in\mathrm{Bounds}$ such that \[ \Vert \phi^{(\mathrm{ref})}_{t}
    \Vert_\infty \leq \left\Vert \widehat{\phi^{(\mathrm{ref})}_{t}}
    \right\Vert_1 \leq C(t).  \]
  \end{cor}

\subsubsection{Estimates for
  $\varphi_t$}\label{sec:propagation-varphi}

\begin{lem}\label{lem:varphi-L2infty-bound} 
  \db Let $U\in\mathcal C^\infty_c(\mathbb R^3,\mathbb R^+_0)$ be a repulsive
  potential. \ed
    There exists a
  $C\in\mathrm{Bounds}$ such that 
  \[ 
    \Vert\varphi_t\Vert_{2\wedge\infty}\leq
    \left\Vert \, \reallywidehat{\,|\varphi_t|\,} \, \right\Vert_{1\wedge2} \leq C(t).  
  \]
\end{lem}

\begin{proof} In order to provide the desired bound we introduce the auxiliary
  wave function 
  \begin{align} 
    \label{eq:varphi-tilde}
    \widetilde\varphi_{t} := \widetilde \phi_{t}
    \varphi_t^{\mathbb T}. 
  \end{align}
  Using the evolution equation (\ref{eq:varphi-evolution}) on
  $\mathbb R^3$, the corresponding one on the torus $\mathbb T$, and definition
  (\ref{eq:tilde-phi-def}), we compute the time derivative
  \begin{align*}
    i\partial_t ( \varphi_t - \widetilde\varphi_t ) 
    = & 
    \db \left(
    -\frac{1}{2}\Delta + U*|\varphi_t|^2 \right) \ed ( \varphi_t -
    \widetilde\varphi_t ) 
    \\ 
    & - \frac{1}{2}\Delta \widetilde\varphi_t +
    \widetilde\phi_t \frac{1}{2}\Delta \varphi_t^{\mathbb T} 
    \\ 
    & + U*\left(
    |\varphi_t|^2 - |\phi^{(\mathrm{ref})}_0|^2+1-|\varphi^{\mathbb T}_t|^2 \right)
    \widetilde\varphi_t.
  \end{align*} 
  \db Recall that $|\widetilde \phi_t|=|\phi^{(\mathrm{ref})}_0|$. \ed In consequence, we get the estimate
  \begin{align*} \partial_t \Vert \varphi_t - \widetilde\varphi_t \Vert_2 \leq
    & \, \Vert \nabla  \widetilde\phi_t \Vert_\infty \, \Vert \nabla
    \varphi_t^{\mathbb T} \Vert_2 + \frac{1}{2} \Vert \Delta \widetilde\phi_t \Vert_2
    \, \Vert \varphi^{\mathbb T}_t \Vert_\infty\\ & +\Vert U \Vert_{1,2} \, \left\Vert
    |\varphi_t|^2 - \db |\widetilde\phi_t|^2\ed + 1-|\varphi^{\mathbb T}_t|^2
    \right\Vert_{1\wedge2} \, \Vert \widetilde\varphi_t \Vert_\infty.
    \end{align*} Furthermore, we consider the bounds:
  \begin{itemize} 
\item The bounds in (\ref{eq:kinetic-estimate}), (\ref{eq:nabla-tilde-phi-infty}), (\ref{eq:delta-tilde-phi}) and 
  Corollary~\ref{cor:epsilon-varphi-torus} ensure
  \begin{align*} \Vert \nabla
    \widetilde\phi_t \Vert_\infty \, \Vert \nabla  \varphi_t^{\mathbb T} \Vert_2 +
    \frac{1}{2} \Vert \Delta \widetilde\phi_t \Vert_2 \, \Vert \varphi^{\mathbb T}_t
    \Vert_\infty \leq C(t)\Lambda^{-\frac{1}{6}};
\end{align*}
      \item
          Definition (\ref{eq:tilde-phi-def}) and 
          Corollary~\ref{cor:epsilon-varphi-torus} imply
          \begin{align*} \Vert
          \widetilde\varphi_t \Vert_\infty \leq \Vert \widetilde\phi_t
        \Vert_\infty \, \Vert \varphi^{\mathbb T}_t \Vert_\infty \leq C(t);
      \end{align*}
  \item\db
    \begin{align}
    \left\Vert |\varphi_t|^2 - |\widetilde\phi_t|^2 + 1-|\varphi^{\mathbb T}_t|^2
    \right\Vert_{1\wedge2}
    & \leq 
    \left\Vert |\widetilde\varphi_t|^2 -
    |\widetilde\phi_t|^2 + 1-|\varphi^{\mathbb T}_t|^2 \right\Vert_{2} +
    \left\Vert |\varphi_t|^2 - |\widetilde\varphi_t|^2 \right\Vert_{1\wedge2}
    \label{eq:pre-sum-of-L12norms};
  \end{align} \ed
  \item Recall definition (\ref{eq:epsilon-evolution-torus}). Using the identity
      \begin{align*}
        |\widetilde\varphi_t|^2-|\widetilde\phi_t|^2+1-|\varphi^{\mathbb T}_t|^2 
        & =
        \db 
        |1+\epsilon_t^{\mathbb T}|^2 \, |\widetilde\phi_t|^2 
        -|\widetilde\phi_t|^2
        + 1
        -|1+\epsilon_t^{\mathbb
        T}|^2 
        \\ 
        & = 
        \left(\epsilon_t^{\mathbb T}+{\epsilon_t^{\mathbb T}}^*+|\epsilon_t^{\mathbb T}|^2\right)
        \left(|\widetilde\phi_t|^2-1\right)
      \end{align*} we find
      \begin{align*} \left\Vert
        |\widetilde\varphi_t|^2-|\widetilde\phi_t|^2+1
        -|\varphi^{\mathbb T}_t|^2 \right\Vert_{2} \leq & 
        \,
        \left(2+\Vert\epsilon_t^{\mathbb T}\Vert_\infty
        \right)
        \left\Vert
        \epsilon_t^{\mathbb T}\left( |\widetilde\phi_t|-1 \right) \right\Vert_{2}
        \left\Vert |\widetilde\phi_t|+1 \right\Vert_{\infty}.
      \end{align*} 
      Moreover,
      $|\widetilde\phi_t|-1=|\phi^{(\mathrm{ref})}_0|-1\leq \chi_\Lambda$ as required in
      (\ref{eq:phi-cutoff-constraint-epsilon-decay}), so that by
      Lemma~\ref{lem:torus} \[
      \left\Vert\epsilon_t^{\mathbb T}\left(|\widetilde\phi_0|-1\right)\right\Vert_{2} \leq
    \left\Vert\epsilon_t^{\mathbb T}\chi_{\Lambda}\right\Vert_{2} \leq
  C(t)\Lambda^{-\frac{1}{3}}, \] and hence, by
 $|\widetilde \phi_t|=|\phi^{(\mathrm{ref})}_0|$,
(\ref{eq:conditions-phi}),
and
  Corollary~\ref{cor:epsilon-varphi-torus}
  \begin{align}\label{eq:phidifferenzen} \left\Vert
    |\widetilde\varphi_t|^2 - |\widetilde\phi_t|^2 + 1-|\varphi^{\mathbb T}_t|^2
  \right\Vert_{2} \leq C(t)\Lambda^{-\frac{1}{3}};
  \end{align} 
  \item \db This implies
    \begin{align}
    (\ref{eq:pre-sum-of-L12norms})  
    & \leq C(t)\Lambda^{-\frac{1}{3}} 
    + \left\Vert 
        |\varphi_t - \widetilde\varphi_t|^2 
        + 2\Re \widetilde\varphi_t^* (\varphi_t-\widetilde\varphi_t) 
      \right\Vert_{1\wedge 2} 
    \nonumber 
    \\ 
    & \leq \db
    C(t)\Lambda^{-\frac{1}{3}} + 
    C(t)\left(
        \Vert\varphi_t - \widetilde\varphi_t \Vert_2^2
        +
        \Vert\varphi_t - \widetilde\varphi_t \Vert_2^2
    \right).\ed
    \label{eq:sum-of-L12norms}
  \end{align}\ed 
\end{itemize} 
These
    ingredients yield the bound
    \begin{align*} \partial_t \Vert \varphi_t -
      \widetilde\varphi_t \Vert_2 \leq C(t)\left( \Lambda^{-\frac{1}{6}} + \Vert
      \varphi_t - \widetilde\varphi_t \Vert_2 + \Vert \varphi_t -
      \widetilde\varphi_t \Vert_2^2 \right).
    \end{align*} 
    With Grönwall's Lemma, $\varphi_0=\widetilde\varphi_0$, and a similar
    argument as used in the proof of Lemma~\ref{lem:phi-propagation},
    we may therefore conclude that 
    \begin{align}
      \Vert \varphi_t - \widetilde\varphi_t \Vert_2 \leq C(t)\Lambda^{-\frac{1}{6}}
      \label{eq:varphi-tildevarphi-diff-L2bound}
    \end{align} 
    holds for all $t\in\mathbb R$ provided $\Lambda$ is sufficiently large.
    This implies
    \begin{align*} \Vert \varphi_t \Vert_{2\wedge\infty} \leq \Vert
      \varphi_t - \widetilde\varphi_t \Vert_2 + \Vert \widetilde\varphi_t
      \Vert_\infty \leq C(t)(\Lambda^{-\frac{1}{6}} + 1)
    \end{align*} which proves the
    claim.
  \end{proof}

Lemma~\ref{lem:L1-norm-evolution} and Lemma~\ref{lem:varphi-L2infty-bound} imply
the following corollary.
\begin{cor}\label{lem:propagation-varphi}
    \db Let $U\in\mathcal C^\infty_c(\mathbb R^3,\mathbb R^+_0)$ be a repulsive
  potential. \ed
    There exists a
  $C\in\mathrm{Bounds}$ such that 
  \begin{align}
    \Vert\epsilon_t\Vert_\infty\leq1+\Vert\varphi_t\Vert_\infty\leq1+\Vert\widehat\varphi_t\Vert_1\leq1+C(t).
    \tag{\ref{eq:varphi-Linfty-bound}}
\end{align}
\end{cor}

\subsubsection{Estimates for $\epsilon_t$}\label{sec:propagation-epsilon}

\begin{lem}\label{lem:propagation-epsilon} 
\db Let $U\in\mathcal C^\infty_c(\mathbb R^3,\mathbb R^+_0)$ be a repulsive
  potential and $\Lambda$ be sufficiently large. \ed
    There exist
  $C_1,C_2\in\mathrm{Bounds}$ such that for all $1/4\leq r<1$
  \begin{align}
    \Vert \epsilon_t \Vert_2
    & \leq C_1(t), \tag{\ref{eq:epsilon-L2-bound}} \\ \Vert
    p^{(\mathrm{ref})}_t \epsilon_t \Vert_2 & \leq
    \frac{C_2(t)}{\Lambda^{1/2}}, \tag{\ref{eq:p-epsilon-L2-bound}} \\
    \Vert \chi_r\epsilon_t \Vert_2 & \leq C(t)\Lambda^{-\frac{1}{3}}.
    \label{eq:m-epsilon-L2}
  \end{align}
\end{lem}

\begin{proof} Thanks to definition (\ref{eq:epsilon-def}) and the evolution
    equations (\ref{eq:varphi-evolution}) and (\ref{eq:phi-evolution}) we
  find
  \begin{align*} \partial_t\Vert \epsilon_t\Vert^2 & \leq \left\Vert
    U*\left( |\varphi_t|^2-|\phi_{t}^{(\mathrm{ref})}|^2
    \right)\phi_t^{(\mathrm{ref})} \right\Vert_2 \\ & \leq \Vert U\Vert_{1,2}
    \, \left\Vert |\varphi_t|^2-|\phi_{t}^{(\mathrm{ref})}|^2
    \right\Vert_{1\wedge2} \, \Vert\phi_{t}^{(\mathrm{ref})}\Vert_\infty.
  \end{align*} The triangle inequality implies 
  \begin{align*}
      \left\Vert
    |\varphi_t|^2-|\phi_{t}^{(\mathrm{ref})}|^2 \right\Vert_{1\wedge2} \leq
    \left\Vert |\varphi_t^{\mathbb T}|^2-1 \right\Vert_{2} + \left\Vert
    |\varphi_{t}|^2-|\phi_0^{(\mathrm{ref})}|^2-|\varphi^{\mathbb T}_t|^2+1
  \right\Vert_{1\wedge2} + \left\Vert |\phi_{t}^{(\mathrm{ref})}|^2 -
  |\phi_{0}^{(\mathrm{ref})}|^2 \right\Vert_2.
    \end{align*}
The terms on the right-hand
side can be estimated as follows:
\begin{itemize} 
  \item Corollary~\ref{cor:epsilon-varphi-torus}, definition of
      $\widetilde\varphi_t^{\mathbb T}$ in (\ref{eq:tildevarphi-T}),
      (\ref{eq:varphi-tildevarphi-diff-torus}), 
      definition of $\epsilon^{\mathbb T}_t$ in
      (\ref{eq:epsilon-evolution-torus}), and
      (\ref{eq:epsilon-L2-bound-torus}) imply
      \begin{align*} 
        \left\Vert |\varphi_t^{\mathbb T}|^2 - 1 \right\Vert_2 
        & \leq 
        \left\Vert |\varphi_t^{\mathbb T}|^2 + 1 \right\Vert_\infty \, 
        \left\Vert |\varphi_t^{\mathbb T}| - 1 \right\Vert_2 
        \\
        & \leq 
        C(t) \left( \left\Vert \varphi_t^{\mathbb T} 
                      - \widetilde\varphi^{\mathbb T}_t 
                    \right\Vert_2 
               + \left\Vert |\widetilde\varphi_t^{\mathbb T}| - 1 \right\Vert_2
             \right) 
        \\ 
        & \leq 
        C(t) \left( 1 +  \left\Vert \epsilon_t^{\mathbb T} \right\Vert_2 \right) 
        \\ 
        & \leq 
        C(t);
\end{align*} 
    \item The definition of $\widetilde\phi_t$  in (\ref{eq:tilde-phi-def}) together with the
        identify $|\phi^{(\mathrm{(ref)}}_0|=|\widetilde \phi_t|$ and the bounds in
    (\ref{eq:phidifferenzen}) and (\ref{eq:varphi-tildevarphi-diff-L2bound})
    ensure
    \begin{align} 
      \left\Vert
        |\varphi_{t}|^2-|\phi_0^{(\mathrm{ref})}|^2+1-|\varphi^{\mathbb T}_t|^2
      \right\Vert_{1\wedge2} 
      \leq 
      C(t)\Lambda^{-\frac{1}{6}}; 
      \label{eq:sum-L12-norm}
  \end{align} \item Recalling (\ref{eq:moduls-phi-phi0-L2-bound}) we know that
    \[ \left\Vert |\phi_{t}^{(\mathrm{ref})}|^2 -
      |\phi_{0}^{(\mathrm{ref})}|^2 \right\Vert_2 \leq C(t)\Lambda^{-\frac{1}{6}}.  \]
  \end{itemize} 
  In consequence, we find 
  \begin{align}
      \left\Vert
    |\varphi_t|^2-|\phi_{t}^{(\mathrm{ref})}|^2 \right\Vert_{1\wedge2} 
    \leq
    C(t)
    \label{eq:varphi-phi-2}
    \end{align}
  and therefore 
  \[ \partial_t\Vert \epsilon_t\Vert^2
  \leq C(t) \] which by Grönwall's Lemma proves the claim
  (\ref{eq:epsilon-L2-bound}) of this lemma.\\

  We continue by recalling Condition~\ref{def:initial-conditions} which
  ensures
  \begin{align*} 
    \left\Vert
    p_{t}^{(\mathrm{ref})}\epsilon_{t}\right\Vert _{2} =
    \frac{1}{\Lambda}\left\Vert \phi_{t}^{(\mathrm{ref})}\right\Vert_{2}
    \left|\left\langle \phi_{t}^{(\mathrm{ref})},\epsilon_{t}\right\rangle
    \right| \nonumber \leq \frac{1}{\Lambda^{1/2}} \left|\left\langle
    \phi_{t}^{(\mathrm{ref})},\epsilon_{t}\right\rangle \right|.
  \end{align*}
  In order to estimate the right-hand side we recall the definition of $\epsilon_t$
  in (\ref{eq:epsilon-def}), the evolution equations
  (\ref{eq:varphi-evolution}) as well as (\ref{eq:phi-evolution}), and regard
  \begin{align}
    i\partial_t\left\langle\epsilon_t,\phi_{t}^{(\mathrm{ref})}\right\rangle &
    = i\partial_t\left\langle e^{it\Vert U \Vert_1}\varphi_t,\phi_{t}^{(\mathrm{ref})}\right\rangle
    \nonumber
    \\ 
    & = 
    \left\langle e^{it\Vert U \Vert_1} \varphi_t,
    U*(|\varphi_t|^2-|\phi_{t}^{(\mathrm{ref})}|^2)\phi_{t}^{(\mathrm{ref})}
    \right\rangle 
    \nonumber 
    \\ 
    & = 
    \left\langle \phi_{t}^{(\mathrm{ref})},
    U*(|\varphi_t|^2-|\phi_{t}^{(\mathrm{ref})}|^2)\phi_{t}^{(\mathrm{ref})}
    \right\rangle \label{eq:imag-term} 
    \\
    & \quad + \left\langle
    \epsilon_t,U*(|\varphi_t|^2-|\phi_{t}^{(\mathrm{ref})}|^2)
    \phi_{t}^{(\mathrm{ref})} \right\rangle.  
    \nonumber
  \end{align} 
  \db Note that term
  (\ref{eq:imag-term}) is real. Hence, 
  the bounds (\ref{eq:epsilon-L2-bound}), (\ref{eq:varphi-phi-2}), 
  and Corollary~\ref{lem:phi-propagationa} imply \ed
  \begin{align*}
    \partial_t\left|\left\langle \epsilon_t,\phi_{t}^{(\mathrm{ref})}
    \right\rangle\right| & \leq \left|\left\langle
    \epsilon_t,U*(|\varphi_t|^2-|\phi_{t}^{(\mathrm{ref})}|^2)
    \phi_{t}^{(\mathrm{ref})} \right\rangle\right| 
    \\ & \leq
    \Vert\epsilon_t\Vert_2 \, \Vert U\Vert_{1,2} \,
    \Vert\,|\varphi_t|^2-|\phi_{t}^{(\mathrm{ref})}|^2\,\Vert_{1\wedge2}\,
    \Vert\phi_{t}^{(\mathrm{ref})}\Vert_\infty \\ & \leq C(t).
  \end{align*}
  An application of Grönwall's Lemma concludes the proof of claim
  (\ref{eq:p-epsilon-L2-bound}) of this lemma.\\

  Finally, with the definition of $\widetilde\phi_t$ and $\epsilon_t^{\mathbb
  T}$ in (\ref{eq:tilde-phi-def}) and (\ref{eq:epsilon-evolution-torus}),
  respectively,  we find the estimate 
  \[ 
    \Vert \chi_r \epsilon_t \Vert_2 \leq \Vert
    \chi_r \widetilde\phi_t \epsilon_t^{\mathbb T} \Vert_2 + \Vert \chi_r
    (\epsilon_t - \widetilde\phi_t \epsilon_t^{\mathbb T}) \Vert_2 \leq \Vert
    \widetilde\phi_t\Vert_\infty \, \Vert \chi_r \epsilon_t^{\mathbb T} \Vert_2 +
  \Vert \epsilon_t - \widetilde\phi_t \epsilon_t^{\mathbb T} \Vert_2.
  \] 
  Applying the definition of $\widetilde\varphi_t$ in (\ref{eq:varphi-tilde}) we
  estimate
  \[ 
    \Vert \epsilon_t - \widetilde\phi_t\epsilon^{\mathbb T}_t \Vert_2 
    \leq 
    \Vert \varphi_t - \widetilde\varphi_t \Vert_2 + \Vert
    \phi^{(\mathrm{ref})}_t - \widetilde\phi_t \Vert_2. 
  \]
  The estimate in (\ref{eq:epsilon-cutoff}) in Theorem~\ref{lem:torus} and the
  bounds (\ref{eq:varphi-tildevarphi-diff-L2bound}),
  (\ref{eq:phi-tildephi-L2}) imply the claim
  (\ref{eq:m-epsilon-L2}).
\end{proof}

\subsection{Proof of Theorem~\ref{thm:main-epsilon}}\label{sec:eta}

In this last section we provide the proof of the fourth main result:

\begin{proof}[Proof of Theorem~\ref{thm:main-epsilon}]
Since the Laplace
  operator is self-adjoint we find \db by means of the evolution equations
  (\ref{eq:epsilon-evolution}) and (\ref{eq:approximate-time-evolution}) that\ed
  \begin{align*} 
    \left\Vert {\epsilon}_{t} - \eta_t \right\Vert _{2} 
    \leq & 
    \left\Vert U*2\Re \left({\epsilon}_{t}^{*}\phi^{(\mathrm{ref})}_t -\eta_{t}^{*} \right) \right\Vert_2 
    \\ 
    & + 
    \left\Vert U*2\Re \left({\epsilon}_{t}^{*}\phi^{(\mathrm{ref})}_t \right)
    \left( \phi^{(\mathrm{ref})}_t - 1 \right) \right\Vert_2 
    \\ 
    & + 
    \left\Vert U*\left[|\phi^{(\mathrm{ref})}_t|^2-1\right]\epsilon_t \right\Vert_2 
    \\ 
    & + 
    \left\Vert U*|{\epsilon}_{t}|^{2} \phi^{(\mathrm{ref})}_t\right\Vert_2
    \\ 
    & + 
    \left\Vert \left[U*|{\epsilon}_{t}|^{2}\right]{\epsilon}_{t} \right\Vert_2 
    \\ 
    & + \left\Vert 
          \left[U*2\Re{\epsilon}_{t}^{*}\phi^{(\mathrm{ref})}_t\right]
          {\epsilon}_{t}
        \right\Vert_2
  \end{align*} 
  we begin with the most crucial estimate
  \begin{align*} \left\Vert
    U*\left[|\phi^{(\mathrm{ref})}_t|^2-1\right]\epsilon_t \right\Vert_2 &
    \leq 
    \left\Vert \phi^{(\mathrm{ref})}_t|+1 \right\Vert_\infty \,
    \left\Vert U*\left[|\phi^{(\mathrm{ref})}_t|-1\right]\epsilon_t
    \right\Vert_2 \\ & \leq C(t) \left( \left\Vert
    U*\left[|\phi^{(\mathrm{ref})}_0|-1\right]\epsilon_t \right\Vert_2 +
    \left\Vert
    U*\left[|\phi^{(\mathrm{ref})}_t|-|\phi^{(\mathrm{ref})}_0|\right]\epsilon_t
    \right\Vert_2 \right).
  \end{align*} 
  \db Using the bounds 
  (\ref{eq:epsilon-L2-bound}), given in Lemma~\ref{lem:propagation-epsilon}, and
  (\ref{eq:moduls-phi-phi0-L2-bound}) we note\ed
  \[ \left\Vert
    U*\left[|\phi^{(\mathrm{ref})}_t|-|\phi^{(\mathrm{ref})}_0|\right]\epsilon_t
    \right\Vert_2 
    \leq 
    \db \Vert U\Vert_2 \ed \, \left\Vert
    |\phi^{(\mathrm{ref})}_t|-|\phi^{(\mathrm{ref})}_0| \right\Vert_2 \, \Vert
    \epsilon_t \Vert_2 \leq C(t)\Lambda^{-\frac{1}{6}}. 
  \] 
  Furthermore, $|\phi^{(\mathrm{ref})}_0|-1\leq \chi_\Lambda$ as required in
  (\ref{eq:phi-cutoff-constraint-epsilon-decay}), and (\ref{eq:m-epsilon-L2})
  imply
  \begin{align*} \left\Vert
      U*\left[|\phi^{(\mathrm{ref})}_0|-1\right]\epsilon_t \right\Vert_2 &
      \leq \left\Vert U*\chi_\Lambda\, \epsilon_t \right\Vert_2 
      \\ 
      & \leq 
      \left\Vert \int dy \, U(\cdot-y) \chi_\Lambda(\cdot) \epsilon_t(\cdot)
      \right\Vert_2
      + 
      \left\Vert \int dy \, U(\cdot-y)
      (\chi_\Lambda(y)-\chi_\Lambda(\cdot)) \epsilon_t(\cdot) \right\Vert_2 
      \\ &
      \leq \Vert U \Vert_1 \, \Vert \chi_\Lambda\epsilon_t \Vert_2 +
      CD\Lambda^{-\frac{1}{3}} \Vert U \Vert_1 \, \Vert \epsilon_t \Vert_2 \\ & \leq
      C(t)(\Lambda^{-\frac{1}{3}} + \Vert\epsilon_t\Vert^2),
    \end{align*} 
    where we used again
    (\ref{eq:Um-decay}) and that $U$ is supported in a ball of radius $D\geq 0$.
    Using Corollary~\ref{lem:phi-propagationa} we collect the following
    estimates:
    \begin{itemize} 
    \item For $\Lambda$ large enough one finds
          \begin{eqnarray*} \left\Vert
            U*2\Re\left({\epsilon}_{t}^{*}\phi^{(\mathrm{ref})}_t-\eta_{t}^{*}\right)
            \right\Vert_{2} & = & \left\Vert \int dy\, U(y)2\Re \left(
            {\epsilon}_{t}^{*}(\cdot-y)\phi^{(\mathrm{ref})}_t(\cdot-y) -
            \eta_{t}^{*}(\cdot-y) \right) \right\Vert _{2} \\ & \leq & 2\int
            dy\,|U(y)|\, \left\Vert
            {\epsilon}_{t}^{*}(\cdot-y)\phi^{(\mathrm{ref})}_t(\cdot-y)
            -\eta_{t}^{*}(\cdot-y) \right\Vert _{2} \\ & \leq & 2
            \left\Vert U\right\Vert _{1} \left( \left\Vert {\epsilon}_{t} -
            \eta_t \right\Vert _{2} +\left\Vert
          (1-{\phi^{(\mathrm{ref})}_t}^*){\epsilon}_{t}\right\Vert _{2}
        \right) \\ 
        & \leq & C \left\Vert {\epsilon}_{t}-\eta_{t}\right\Vert_2 +
        C(t)\Lambda^{-1/6},
  \end{eqnarray*} 
  where thanks to the ingredients:
    \begin{itemize}
        \item 
          $
            \widetilde\phi_{t} :=
            \exp\left(-it U*\left(|\phi_{0}^{(\mathrm{ref})}|^2-1\right)\right)
            \phi_{0}^{(\mathrm{ref})},
          $
          as defined in (\ref{eq:tilde-phi-def});
        \item $\left\Vert \phi^{(\mathrm{ref})}_t - \widetilde\phi_t
            \right\Vert_2 \leq C(t)\Lambda^{-\frac{1}{6}}$ from line
            (\ref{eq:phi-tildephi-L2});
        \item $\Vert \epsilon_t\Vert_\infty \leq C(t)$ from (\ref{eq:varphi-Linfty-bound});
        \item
          $\Vert \phi_0 \Vert_\infty \leq C$, as required in
          Condition~\ref{def:initial-conditions};
        \item $\Vert \epsilon_t \Vert_2\leq C(t)$ and $\Vert \chi_r\epsilon_t \Vert_2  \leq C(t)\Lambda^{-\frac{1}{3}}$
            for $1/4\leq r<1$ as proven in Lemma~\ref{lem:propagation-epsilon};
        \item Since $U$ is supported in a ball of radius $D\geq 0$ and due to
            (\ref{eq:phi-cutoff-constraint-epsilon-decay}) in
          Condition~\ref{def:initial-conditions} one has $U*\left(|\phi_{0}^{(\mathrm{ref})}|^2-1\right)(x)=0$
            for $x\in B_{1/2\Lambda^{1/3}-2D}$;
        \item Consequently, for sufficiently large $\Lambda$ one has
            $(1-\widetilde\phi^*_t)(1-\chi_r)(x)=0$ for $r=1/4$;
    \end{itemize}
  we used
  \begin{align*}
    \left\Vert (1-{\phi^{(\mathrm{ref})}_t}^*){\epsilon}_{t}\right\Vert_{2}
    & \leq 
    \left\Vert (1-{\widetilde\phi}_t^*){\epsilon}_{t}\right\Vert_{2} + 
    \left\Vert
    ({\phi^{(\mathrm{ref})}_t}^*-{\widetilde\phi}_t^*){\epsilon}_{t}\right\Vert_{2}
    \\
    & \leq \left\Vert (1-{\widetilde\phi}_t^*)(1-\chi_{1/4}){\epsilon}_{t}\right\Vert_{2} +
 \left\Vert (1-{\widetilde\phi}_t^*)\chi_{1/4}{\epsilon}_{t}\right\Vert_{2}
 +C(t)\Lambda^{-1/6} \\
 & \leq 0  
 +2C(t)\Lambda^{-1/6}.
  \end{align*}

  \item 
      \begin{align*}
      \left\Vert U*2\Re \left({\epsilon}_{t}^{*}\phi^{(\mathrm{ref})}_t \right)
        \left( \phi^{(\mathrm{ref})}_t - 1 \right) 
      \right\Vert_2 
      & \leq 
      \left\Vert U*2\Re \left({(1-\chi_{1/4})\epsilon}_{t}^{*}\phi^{(\mathrm{ref})}_t \right)
        \left( \widetilde\phi_t - 1 \right) 
      \right\Vert_2
      \\
      & \qquad +
      \left\Vert U*2\Re \left({\chi_{1/4}\epsilon}_{t}^{*}\phi^{(\mathrm{ref})}_t \right)
        \left( \widetilde\phi_t - 1 \right) 
      \right\Vert_2
      \\
      & \qquad +
      \left\Vert U*2\Re \left({\epsilon}_{t}^{*}\phi^{(\mathrm{ref})}_t \right)
        \left( \phi^{(\mathrm{ref})}_t - \widetilde\phi_t \right) 
      \right\Vert_2
      \\
      & \leq 
      \left\Vert U*2\Re \left((1-\chi_{1/4}){\epsilon}_{t}^{*}\phi^{(\mathrm{ref})}_t \right)
        \left( \widetilde\phi_t - 1 \right) 
      \right\Vert_2
      \\
      & \qquad 
      +2 \Vert U \Vert_1 \Vert \chi_{1/4} \epsilon_t \Vert_2 \, \Vert
      \phi^{(\mathrm{ref})}_t \Vert_\infty \left\Vert
      \widetilde\phi_t - 1 \right\Vert_\infty
      \\
      & \qquad 
      +2 \Vert U \Vert_1 \Vert \epsilon_t \Vert_\infty \, \Vert
      \phi^{(\mathrm{ref})}_t \Vert_\infty \left\Vert
      \phi^{(\mathrm{ref})}_t - \widetilde\phi_t \right\Vert_2 
      \\
      &\leq
      0+2C(t) \Lambda^{-1/6},
    \end{align*}
    where in addition to the ingredients for the previous term we have used:
    \begin{itemize}
        \item $\Vert
          \phi_t^{(\mathrm{ref})}\Vert_\infty\leq C(t)$ as proven in
          Corollary~\ref{lem:phi-propagationa};
        \item $\operatorname{supp}U*2\Re
            \left((1-\chi_{1/4}){\epsilon}_{t}^{*}\phi^{(\mathrm{ref})}_t \right)\subset
          B_{1/4\Lambda^{1/3}+2D}$;
        \item 
          Similarly as
          above
         one has
            $(\widetilde\phi_t-1)(1-\chi_r)(x)=0$ for $r=1/4$;
        and sufficiently large $\Lambda$.
    \end{itemize}
  \item
      \begin{eqnarray*} \left\Vert
        U*|{\epsilon}_{t}|^{2}\phi^{(\mathrm{ref})}_t\right\Vert _{2} & = &
        \left\Vert \phi^{(\mathrm{ref})}_t\right\Vert_\infty \, \left\Vert
        \int dy\, U(\cdot-y)|{\epsilon}_{t}(y)|^{2}\right\Vert _{2} \leq
        C(t)\int dy\,|{\epsilon}_{t}(y)|^{2}\,\left\Vert U(\cdot-y)\right\Vert
        _{2}\\ & \leq & C(t)\left\Vert U\right\Vert _{2}\left\Vert
        {\epsilon}_{t}\right\Vert _{2}^{2}\leq C\left\Vert
        {\epsilon}_{t}\right\Vert _{2}^{2};
      \end{eqnarray*}

\item \[ \left\Vert
    \left[U*|{\epsilon}_{t}|^{2}\right]{\epsilon}_{t}\right\Vert
    _{2}=\left\Vert \int dy\,
    U(\cdot-y)|{\epsilon}_{t}(y)|^{2}{\epsilon}_{t}\right\Vert
    _{2}\leq\left\Vert U\right\Vert _{2}\left\Vert {\epsilon}_{t}\right\Vert
    _{2}^{3}\leq C\left\Vert {\epsilon}_{t}\right\Vert _{2}^{3}; \]

\item \[ \left\Vert
    \left[U*2\Re{\epsilon}_{t}^{*}\phi^{(\mathrm{ref})}_t\right]{\epsilon}_{t}\right\Vert
    _{2} \leq  \left\Vert|\phi^{(\mathrm{ref})}_t\right\Vert_\infty \left\Vert
    U\right\Vert _{2}\left\Vert {\epsilon}_{t}\right\Vert _{2}^{2}\leq
  C(t)\left\Vert {\epsilon}_{t}\right\Vert _{2}^{2}.  \]

\end{itemize} Hence, we have shown \[ \partial_t\left\Vert
  \eta_{t}-{\epsilon}_{t}\right\Vert _{2} \leq C\left\Vert
  \eta_{t}-{\epsilon}_{t}\right\Vert _{2} + C(t)\Lambda^{-\frac{1}{6}}
  +C(t)\left( \left\Vert
  {\epsilon}_{t}\right\Vert _{2}^{2}+\left\Vert {\epsilon}_{t}\right\Vert
  _{2}^{3}\right) \] which together with Grönwall's Lemma proves the claim.

\end{proof}

\section{Appendix}

In several steps we have used the convenient computation formulas
(\ref{eq:hat-multipilication})-(\ref{eq:hat-Q-commutation-1bis}) concerning the counting operators that 
were established in in \cite[Lemma 1]{pickl2010a} and are repeated here for
easier reference:
\begin{lem} Given the definitions (\ref{eq:projectors})-(\ref{eq:def-w}), the following relations are true:
\begin{enumerate} \item
      \begin{equation}\tag{\ref{eq:hat-multipilication}}
        \widehat{v^{\varphi}}\widehat{w^{\varphi}}=\widehat{\left(vw\right)^{\varphi}}=\widehat{w^{\varphi}}\widehat{v^{\varphi}}
      \end{equation}

\item
  \begin{equation}
    \left[\widehat{w^{\varphi}},p_{k}^{\varphi}\right]=\left[\widehat{w^{\varphi}},q_{k}^{\varphi}\right]=0\tag{\ref{eq:hat-pq-commutation}}
  \end{equation}

\item \[ \left[\widehat{w^{\varphi}},P_{k}^{\varphi}\right]=0 \]

\item For $n(k)=\sqrt{\frac{k}{N}}$ we have
  \begin{equation}
    \left(\widehat{n^{\varphi}}\right)^{2}=\frac{1}{N}\sum_{k=1}^{N}q_{k}^{\varphi}\tag{\ref{eq:old-q1}}
  \end{equation}

\item For $\Psi\in\left(L^{2}\right)^{\odot N}$ we have that
  \begin{align}
    \left\Vert \widehat{w^{\varphi}}q_{1}^{\varphi}\Psi\right\Vert_2  & = 
    \left\Vert \widehat{w^{\varphi}}\widehat{n^{\varphi}}\Psi\right\Vert_2
    \tag{\ref{eq:old-q1-2}}
    \\ 
    \left\Vert
    \widehat{w^{\varphi}}q_{1}^{\varphi}q_{2}^{\varphi}\Psi\right\Vert_2  & \leq 
    \sqrt{\frac{N}{N-1}}\left\Vert
    \widehat{w^{\varphi}}\left(\widehat{n^{\varphi}}\right)^{2}\Psi\right\Vert_2
    \tag{\ref{eq:old-q1-3}}
  \end{align}

\item For any  function $Y:\in L^\infty(\mathbb{R}^3)$ and $Z:\in L^\infty(\mathbb{R}^6)$ and
  \[ A_{0}^{\varphi}=p_{1}^{\varphi},\qquad
    A_{1}^{\varphi}=q_{1}^{\varphi},\qquad
    B_{0}^{\varphi}=p_{1}^{\varphi}p_{2}^{\varphi},\qquad
    B_{1}^{\varphi}=p_{1}^{\varphi}q_{2}^{\varphi},\qquad
    B_{2}^{\varphi}=q_{1}^{\varphi}q_{2}^{\varphi} \] we have
    \begin{equation}   \widehat{w^{\varphi}}A_{j}^{\varphi}Y(x_{1})A_{l}^{\varphi}=A_{j}^{\varphi}Y(x_{1})A_{l}^{\varphi}\widehat{w_{j-l}^{\varphi}}\quad \text{with}\,\, j,l=0,1,\tag{\ref{eq:hat-Q-commutation-1}}
\end{equation}
and
\begin{equation}\widehat{w^{\varphi}}B_{j}^{\varphi}Z(x_{1},x_{2})B_{l}^{\varphi}=B_{j}^{\varphi}Z(x_{1},x_{2})B_{l}^{\varphi}\widehat{w_{j-l}^{\varphi}}\quad \text{with}\,\,j,l=0,1,2.\tag{\ref{eq:hat-Q-commutation-1bis}}
    \end{equation}
    \end{enumerate}
\end{lem}

    \begin{proof}\mbox{}
\begin{enumerate}

\item Since $p_k^{\varphi}$ is orthogonal to $q_k^{\varphi}$ for any $1\leq k
    \leq N $ it follows, that the $P_k^{\varphi}$,  $1\leq k \leq N $ (see
    (\ref{eq:projector})) are pairwise orthogonal projectors. Hence, by (\ref{eq:def-w})
$$\widehat{v^{\varphi}}\widehat{w^{\varphi}}=\sum_{k,j=0}^{N}v(k)P_{k}^{\varphi}w(j)P_{j}^{\varphi}=\sum_{k=0}^{N}v(k)w(k)P_{k}^{\varphi}=\widehat{\left(vw\right)^{\varphi}}\;.$$
Similarly one can show $\widehat{\left(vw\right)^{\varphi}}=\widehat{w^{\varphi}}\widehat{v^{\varphi}}$.

\item $p_{k}^{\varphi}$ commutes with $p_{j}^{\varphi}$ and $q_{j}^{\varphi}$
    for any $j,k$. It follows that $p_{k}^{\varphi}$ commutes with any
    $P_{j}^{\varphi}$ since the latter is a product of $p$'s and $q$'s. In view
    of (\ref{eq:def-w}) we observe that 
$p_{k}^{\varphi}$ commutes with any weighted counting operators
$\widehat{w^{\varphi}}$. A analogous argument can be made for $q_{k}^{\varphi}$.

\item Observing that $P_k^{\varphi}$ is given as a symmetric product of $p$'s
    and $q$'s (see (\ref{eq:projector})) the claim follows from
    (\ref{eq:hat-pq-commutation}).

\item

Note that $1=\prod_{k=1}^N(p_{k}^{\varphi}+q_{k}^{\varphi})$. Expanding this
product and sorting the summands according to the number of $q$-factors it
follows that 
$1=\sum_{k=0}^N P_k^{\varphi}$. Hence, the claim (\ref{eq:old-q1}) follows from
\begin{align*}
N^{-1}\sum_{k=1}^Nq_k^{\varphi}&=N^{-1}\sum_{k=1}^Nq_k^{\varphi}\sum_{j=0}^N
P_j^{\varphi}= N^{-1}\sum_{j=0}^N\sum_{k=1}^Nq_k^{\varphi}
P_j^{\varphi}=N^{-1}\sum_{j=0}^Nj P_j^{\varphi}=\left(\widehat{n^{\varphi}}\right)^{2},
\end{align*}
where in the last step we have used (\ref{eq:hat-multipilication}).

\item

Using symmetry we get
\begin{align*}
\left\Vert \widehat{w^{\varphi}}q_{1}^{\varphi}\Psi\right\Vert^2_2 
&=
\left\langle\Psi,q_{1}^{\varphi} \left(\widehat{w^{\varphi}}\right)^2q_{1}^{\varphi}\Psi\right\rangle
=N^{-1}\sum_{k=1}^N\left\langle\Psi,q_{k}^{\varphi} \left(\widehat{w^{\varphi}}\right)^2q_{k}^{\varphi}\Psi\right\rangle\;.
\end{align*}
Using (\ref{eq:hat-pq-commutation}), then (\ref{eq:hat-multipilication}) and then (\ref{eq:old-q1}) the latter equals
\begin{align*}
\left\langle\Psi,\left(N^{-1}\sum_{k=1}^Nq_{k}^{\varphi}\right) \left(\widehat{w^{\varphi}}\right)^2\Psi\right\rangle
&=\left\langle\Psi,\left(\widehat{n^{\varphi}}\right)^{2}\left(\widehat{w^{\varphi}}\right)^2\Psi\right\rangle
=\left\Vert \widehat{w^{\varphi}}\widehat{n^{\varphi}}\Psi\right\Vert^2_2 
\end{align*}
and (\ref{eq:old-q1-2}) follows.

In a similar way we get
\begin{align*}
\left\Vert
    \widehat{w^{\varphi}}q_{1}^{\varphi}q_{2}^{\varphi}\Psi\right\Vert_2^2 
    &=\left\langle\Psi,q_{1}^{\varphi}q_{2}^{\varphi}\left(\widehat{w^{\varphi}}\right)^2q_{1}^{\varphi}q_{2}^{\varphi}\Psi\right\rangle
   \\ &=\frac{1}{N(N-1)}\sum_{j\neq k}\left\langle\Psi,q_{j}^{\varphi}q_{k}^{\varphi}\left(\widehat{w^{\varphi}}\right)^2q_{j}^{\varphi}q_{k}^{\varphi}\Psi\right\rangle\;.
\end{align*}
Using that
$\left\langle\Psi,q_{k}^{\varphi}q_{k}^{\varphi}\left(\widehat{w^{\varphi}}\right)^2q_{k}^{\varphi}q_{k}^{\varphi}\Psi\right\rangle
$ is for any $k$ quadratic, and thus positive, we find
\begin{align*}
\left\Vert
    \widehat{w^{\varphi}}q_{1}^{\varphi}q_{2}^{\varphi}\Psi\right\Vert_2^2 
    &\leq \frac{1}{N(N-1)}\sum_{j,
     k=1}^N\left\langle\Psi,q_{j}^{\varphi}q_{k}^{\varphi}\left(\widehat{w^{\varphi}}\right)^2q_{j}^{\varphi}q_{k}^{\varphi}\Psi\right\rangle
     \\&=
     \frac{N^2}{N(N-1)}\left\langle\Psi,\left(N^{-1}\sum_{j=1}^Nq_{j}^{\varphi}\right)\left(N^{-1}\sum_{
     k=1}^Nq_{k}^{\varphi}\right)\left(\widehat{w^{\varphi}}\right)^2\Psi\right\rangle
     \\&=\frac{N}{N-1}\left\langle\Psi,\left(\widehat{n^{\varphi}}\right)^{4}\left(\widehat{w^{\varphi}}\right)^2\Psi\right\rangle
\\&=\frac{N}{N-1}\left\Vert
\widehat{w^{\varphi}}\left(\widehat{n^{\varphi}}\right)^{2}\Psi\right\Vert_2^2
\;.
\end{align*}
\item


The proof is very similar for all the combinations of $A$ and $B$ operators. Therefore, we only demonstrate one case and start with the following
computation. Denoting the tensor product by $\otimes$, we find
\begin{align*}
    p_1^\varphi Y(x_1) q_1^\varphi \, P_k^\varphi 
    & = 
    p_1^\varphi Y(x_1) q_1^\varphi \left[ (q^\varphi)^{\odot k}\odot
    (p^\varphi)^{\odot (N-k)}\right] \\
    & = 
    p_1^\varphi Y(x_1) \left[ q_1^\varphi \otimes (q^\varphi)^{\odot (k-1)}\odot
    (p^\varphi)^{\odot (N-k)} \right] \\
    & = 
    p_1^\varphi \, \left[ 1 \otimes (q^\varphi)^{\odot (k-1)}\odot
(p^\varphi)^{\odot (N-k)}\right]  Y(x_1) q_1^\varphi \\
    & = 
    \left[ p_1^\varphi \otimes (q^\varphi)^{\odot (k-1)}\odot
    (p^\varphi)^{\odot (N-k)} \right] Y(x_1) q_1^\varphi \\
    & = 
    \left[ 1\otimes (q^\varphi)^{\odot (k-1)}\odot
    (p^\varphi)^{\odot (N-k)} \right] p_1^\varphi Y(x_1) q_1^\varphi \\
    & = 
    \left[ (q^\varphi)^{\odot (k-1)}\odot
    (p^\varphi)^{\odot (N-k+1)} \right] p_1^\varphi Y(x_1) q_1^\varphi \\
    &= P^\varphi_{k-1} \, p_1^\varphi Y(x_1) q_1^\varphi.
\end{align*}

Similar arguments can be applied for the various combinations of $A$ and $B$
operators to show
 \begin{equation}
   P^\varphi_{k}A_{j}^{\varphi}Y(x_{1})A_{l}^{\varphi}=A_{j}^{\varphi}Y(x_{1})A_{l}^{\varphi}P^\varphi_{k+l-j}\quad \text{with}\,\, j,l=0,1,
\end{equation}
and
\begin{equation}
P^\varphi_{k}B_{j}^{\varphi}Z(x_{1},x_{2})B_{l}^{\varphi}=B_{j}^{\varphi}Z(x_{1},x_{2})B_{l}^{\varphi}P^\varphi_{k+l-j}\quad \text{with}\,\,j,l=0,1,2.
    \end{equation}
Using these identities together with the convention $P^\varphi_k=0$ for
$k\not\in\{0,1,\ldots,N\}$, see (\ref{def-Pk-2}), and the
definiton (\ref{eq:def-w}), we get
\begin{align*}
   \widehat{w^{\varphi}}A_{j}^{\varphi}Y(x_{1})A_{l}^{\varphi}
   &=\sum_{k=-\infty}^{\infty}w(k)P_{k}^{\varphi}A_{j}^{\varphi}Y(x_{1})A_{l}^{\varphi}
   \\&=\sum_{k=-\infty}^{\infty}w(k)A_{j}^{\varphi}Y(x_{1})A_{l}^{\varphi}P_{k+l-j}^{\varphi}
\end{align*}
Substituting the index of the sum by $m=k+l-j$ we get
\begin{align*}
   \widehat{w^{\varphi}}A_{j}^{\varphi}Y(x_{1})A_{l}^{\varphi}
  &=\sum_{m=-\infty}^{\infty}w(m+j-l)A_{j}^{\varphi}Y(x_{1})A_{l}^{\varphi}P_{m}^{\varphi}
  \\&=A_{j}^{\varphi}Y(x_{1})A_{l}^{\varphi}\sum_{m=-\infty}^{\infty}w(m+j-l)P_{m}^{\varphi}
   \\&=A_{j}^{\varphi}Y(x_{1})A_{l}^{\varphi}\widehat{w_{j-l}^{\varphi}}\;.
\end{align*}
In the same way we can prove the second formula:
\begin{align*}
\widehat{w^{\varphi}}B_{j}^{\varphi}Z(x_{1},x_{2})B_{l}^{\varphi}
&=\sum_{k=-\infty}^{\infty}w(k)P_{k}^{\varphi}B_{j}^{\varphi}Z(x_{1},x_{2})B_{l}^{\varphi}
\\&=\sum_{k=-\infty}^{\infty}w(k)B_{j}^{\varphi}Z(x_{1},x_{2})B_{l}^{\varphi}P_{k+l-j}^{\varphi}
\\&=\sum_{m=-\infty}^{\infty}w(m+j-l)B_{j}^{\varphi}Z(x_{1},x_{2})B_{l}^{\varphi}P_{m}^{\varphi}
\\&
=B_{j}^{\varphi}Z(x_{1},x_{2})B_{l}^{\varphi}\widehat{w_{j-l}^{\varphi}}\;.
\end{align*}

\end{enumerate}
\end{proof}

\vskip.5cm

\noindent \emph{Dirk - Andr\'e Deckert}\\
Department of Mathematics\\
University of California Davis\\
One Shields Avenue, Davis, California 95616, USA\\
\texttt{deckert@math.ucdavis.edu}

\vskip.5cm

\noindent \emph{J\"urg Fr\"ohlich}\\
Theoretische Physik, ETH Zürich\\
CH-8093 Zürich, Switzerland\\
\texttt{juerg@phys.ethz.ch}

\vskip.5cm

\noindent \emph{Peter Pickl}\\
Mathematisches Institut der LMU M\"unchen\\
Theresienstra\ss e 39, 80333 M\"unchen, Germany\\
\texttt{pickl@math.lmu.de}

\vskip.5cm

\noindent \emph{Alessandro Pizzo}\footnote{On leave of absence from the Mathematics Department at University of California Davis}\\
Dipartimento di Mathematica\\
Universit\`a di Roma Tor Vergata\\
Via della Ricerca Scientifica 1, Roma, 00133, Italy\\
\texttt{pizzo@mat.uniroma2.it}

\end{document}